\documentclass{llncs}

\usepackage{amssymb,amsmath,stmaryrd}
\usepackage{graphicx,color}
\usepackage{nicefrac}
\usepackage{hyperref}
\usepackage{algorithm2e}
\usepackage{wrapfig}
\usepackage{float}
\usepackage[calc]{adjustbox}
\usepackage{array, multirow, bigdelim, makecell, booktabs}
\usepackage[colorinlistoftodos,textsize=tiny,color=orange!70
,disable
]{todonotes}
\usepackage{tikz}
\usepackage{pgfplots}
\usepackage{mathtools}
\usepackage{csquotes}\usetikzlibrary{calc,matrix,arrows,shapes,automata,backgrounds,petri,decorations.pathreplacing}
\usepackage{amsfonts}
\usepackage[displaymath, mathlines]{lineno}

\allowdisplaybreaks

\newlength{\strutheight}

\newcommand\arxive{\cite{wbhhk:stochastic-decision-petri-nets-arxiv}}

\newcommand{\short}[1]{}\newcommand{\full}[1]{#1}

\short{\usepackage[appendix=strip,bibliography=common]{apxproof}}
\full{\usepackage[appendix=append,bibliography=common]{apxproof}}

\newtheoremrep{theorem}{Theorem}[section]
\newtheoremrep{lemma}[theorem]{Lemma}
\newtheoremrep{assumption}[theorem]{Assumption}
\newtheoremrep{definition}[theorem]{Definition}
\newtheoremrep{proposition}[theorem]{Proposition}
\newtheoremrep{corollary}[theorem]{Corollary}
\newtheoremrep{example}[theorem]{Example}

\newcommand{\mN}{{\mathbb{N}}}
\newcommand{\mR}{{\mathbb{R}}}
\newcommand{\mP}{{\mathbb{P}}}
\newcommand{\mC}{{\mathbb{C}}}
\newcommand{\mE}{{\mathbb{E}}}
\newcommand\funcset[2]{\left({#1}\rightarrow{#2}\right)}
\newcommand\setto[1]{\{1,\dots,#1\}}
\newcommand\val{\mathbb{V}}
\newcommand\valfunc{V}
\newcommand{\transitions}{\mathit{tr}}
\newcommand{\places}{\mathit{pl}}
\DeclareMathOperator{\length}{len}
\DeclareMathOperator{\pref}{pre}
\DeclareMathOperator{\exte}{ext}
\newcommand{\firingsequences}[1]{\mathcal{FS}(#1)}
\newcommand{\validconfs}[1]{\mathcal{C}^\omega(#1)}
\newcommand{\validsubconfs}[1]{\mathcal{C}(#1)}
\newcommand{\bc}[1]{\mathit{BC}(#1)}
\newcommand\pnc{\ensuremath{\mathcal{N}}}

\newcommand{\np}{\mathsf{NP}}
\newcommand{\p}{\mathsf{P}}
\newcommand{\pp}{\mathsf{PP}}
\newcommand{\hashp}{\mathsf{\#P}}

\newcommand{\half}{\nicefrac{1}{2}}

\newcommand{\pre}[1]{\ensuremath{\!~^\bullet{#1}}}
\newcommand{\post}[1]{\ensuremath{{#1} {^\bullet}}}

\newcommand\set[1]{\{#1\}}
\newcommand\transition[3]{#1\rightarrow_{#2}#3}
\DeclareMathOperator{\supp}{supp}

\renewcommand{\phi}{\varphi}
\renewcommand{\epsilon}{\varepsilon}

\newcommand{\del}[1]{}

\newcommand{\dmap}{\ensuremath{\mathsf{D\mbox{-}MAP}}}
\newcommand{\dpr}{\ensuremath{\mathsf{D\mbox{-}PR}}}
\newcommand{\sat}{\ensuremath{\mathsf{SAT}}}
\newcommand{\threesat}{\ensuremath{\mathsf{3\mbox{-}SAT}}}
\newcommand{\valn}[1]{\ensuremath{\mathsf{#1\mbox{-}VAL}}}
\newcommand{\fconval}{\ensuremath{\mathsf{FCON\mbox{-}VAL}}}
\newcommand{\safcval}{\ensuremath{\mathsf{SAFC\mbox{-}VAL}}}
\newcommand{\poln}[1]{\ensuremath{\mathsf{#1\mbox{-}POL}}}
\newcommand{\fconpol}{\ensuremath{\mathsf{FCON\mbox{-}POL}}}
\newcommand{\safcpol}{\ensuremath{\mathsf{SAFC\mbox{-}POL}}}

\def\orcidID#1{\href{https://orcid.org/#1}{\includegraphics[height=8pt]{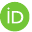}}}

\author{Florian Wittbold\inst{1}\orcidID{0000-0001-8307-503X}
   \and Rebecca Bernemann\inst{1}\orcidID{0000-0002-3240-0952}
   \and Reiko Heckel\inst{2}\orcidID{0000-0003-4719-0772}
   \and Tobias Heindel\inst{3}\orcidID{0000-0003-3371-8564}
   \and Barbara K\"onig\inst{1}\orcidID{0000-0002-4193-2889}}
\authorrunning{F. Wittbold et al.}  \institute{Universit\"at
  Duisburg-Essen, Duisburg, Germany \and University of Leicester, UK
  \and Heliax Technologies GmbH}

\title{Stochastic Decision Petri Nets}

\begin{document}

\maketitle


\begin{abstract}
  We introduce stochastic decision Petri nets (SDPNs), which are a
  form of stochastic Petri nets equipped with rewards and a control
  mechanism via the deactivation of controllable transitions.
  Such nets can be translated into Markov decision processes (MDPs),
  potentially leading to a combinatorial explosion in the number of
  states due to concurrency.
  Hence we restrict ourselves to instances where nets are either safe,
  free-choice and acyclic nets (SAFC nets) or even occurrence nets and
  policies are defined by a constant deactivation pattern.
  We obtain complexity-theoretic results for such cases via a close
  connection to Bayesian networks, in particular we show that for SAFC
  nets the question whether there is a policy guaranteeing a reward
  above a certain threshold is $\np^\pp$-complete.
  We also introduce a partial-order procedure which uses an SMT solver
  to address this problem.
\end{abstract}

\section{Introduction}\label{sec:introduction}

State-based probabilistic systems are typically modelled as Markov
chains~\cite{Tolver2016}, i.e., transition systems where transitions
are annotated with probabilities.
This admits an intuitive graphical visualization and efficient
analysis techniques~\cite{gs:markov-chains}.
By introducing additional non-determinism, one can model a system
where a player can make decisions, enriched with randomized choices.
This leads to the well-studied model of Markov decision processes
(MDPs)~\cite{bellman1957markovian,fs:handbook-mdps} and the challenge
is to synthesize strategies that maximize the reward of the player.

In this paper we study stochastic systems enriched with a mechanism
for decision making in the setting of concurrent systems.
Whenever a system exhibits a substantial amount of concurrency, i.e.,
events that may potentially happen in parallel, compiling it down to a
state-based system -- such as an MDP -- can result in a combinatorial
state explosion and a loss in efficiency of MDP-based methods.
We base our models on stochastic Petri nets~\cite{Marsan88}, where
Petri nets are a standard formalism for modelling concurrent systems,
especially such systems where resources are generated and consumed.
When considering the discrete-time semantics of such stochastic nets,
it is conceptually easy to transform them into Markov chains, but this
typically leads to a state space explosion.

There exist successful partial order methods for analyzing concurrent
systems that avoid explicit interleavings and the enumeration of all
reachable states.
Instead, they work with partial orders -- instead of total orders --
of events.
While such techniques are well understood in the absence of random
choices, leading for instance to methods such as unfoldings~\cite{eh:unfoldings-book}, there are considerable difficulties to
reconcile probability and partial order.
Progress has been made by the introduction of the concept
of branching cells~\cite{ab:true-concurrency-probabilistic} that
encapsulate independent choices, but to our knowledge there is no
encompassing theory that provides off-the-shelf partial order methods
for computing the probability of reaching a certain goal (e.g.
marking a certain place) in a stochastic net.

The contributions of this paper are the introduction of a new model:
stochastic decision Petri nets (SDPNs) and its connection to Markov
decision processes (MDPs).  The transformation of SDPNs into MDPs is
relatively straightforward, but may lead to state space explosion,
i.e., exponentially many markings, due to the concurrency inherent in
the Petri net.  This can make the computation of the optimal policy
infeasible.  We restrict ourselves to a subclass of nets which are
safe, acyclic and free-choice (SAFC) and to constant policies and
study the problem of determining a policy that guarantees a payoff
above some bound.  Our result is that the problem $\safcpol$ of
determining such a policy, despite the restrictions, is still
${\np}^{\pp}$-complete.  We reduce from the $\dmap$ problem for
Bayesian networks ~\cite{park2004complexity} (in fact the two problems
are interreducible under mild restrictions) and show the close
connection of reasoning about stochastic Petri nets and Bayesian
networks.  Furthermore, for SAFC nets, there is a partial-order
solution procedure via an SMT solver, for which we obtain encouraging
runtime results.  For the simpler free-choice occurrence nets, we
obtain an $\np$-completeness result.

Note that the \full{main body of the }paper contains some proof
sketches, while full proofs and an additional example can be found in
\short{\arxive}\full{the appendix}.

\section{Preliminaries}\label{sec:preliminaries}

By $\mN$ we denote the natural numbers without $0$, while $\mN_0$
includes $0$.

Given two sets $X,Y$ we denote by $\funcset{X}{Y}$ the set of all
functions from $X$ to $Y$.
Given a function $f\colon X\to\mN_0$ or $f\colon X\to \mR$ with $X$
finite, we define $\|f\|_\infty = \max_{x\in X} f(x)$ and
$\supp(f)=\set{x\in X\mid f(x)\neq 0}$.

\paragraph*{Complexity Classes:} In addition to well-known complexity
classes such as $\p$ and $\np$, our results also refer to $\pp$
(see~\cite{Papadimitriou94}).
This class is based on the notion of a probabilistic Turing machine,
i.e., a non-deterministic Turing machine whose transition function is
enriched with probabilities, which means that the acceptance function
becomes a random variable.
A language $L$ lies in $\pp$ if there exists a probabilistic Turing
machine $M$ with polynomial runtime on all inputs such that a word
$w\in L$ iff it is accepted with probability strictly greater than
$\half$.
As probabilities we only allow numbers $\rho$ that are efficiently
computable, meaning that the $i$-th bit of $\rho$ is computable in a
time polynomial in $i$. (See~\cite{ab:computational-complexity-modern}
for a discussion on why such probabilistic Turing machines have equal
expressivity with those based on fair coins, which is not the case if
we allow arbitrary numbers.)

Given two complexity classes $A,B$ and their corresponding machine
models, by $A^B$ we denote the class of languages that are solved by a
machine of class $A$, which is allowed to use an oracle answering
yes/no-questions for a language $L\in B$ at no extra cost in terms of
time or space complexity.
In particular $\np^\pp$ denotes the class of languages that can be
accepted by a non-deterministic Turing machine running in polynomial
time that can query a black box oracle solving a problem in $\pp$.

By Toda's theorem~\cite{t:todas-theorem}, a polynomial time Turing
machine with a $\pp$ oracle ($\p^\pp$) can solve all problems in the
polynomial hierarchy.

In order to prove hardness results we use the standard polynomial-time
many-one reductions, denoted by $A\le_p B$ for problems $A,B$
(see~\cite{gj:intractability}).

\paragraph*{Stochastic Petri Nets:}
A stochastic Petri net~\cite{Marsan88} is given by a tuple
$N = (P,T,\pre{(\,)}, \post{(\,)}, \Lambda, m_0)$ where $P$ and $T$
are finite sets of places and transitions,
$\pre{(\ )}, \post{(\ )}:T \rightarrow \funcset{P}{\mN_0}$ determine
for each transition its pre-set and post-set including multiplicities,
$\Lambda\colon T\to \mR_{>0}$ defines the firing rates and
$m_0\colon P\to \mN_0$ is the initial marking.
By $\mathcal{M}(N)$ we denote the set of all markings of $N$, i.e.,
$\mathcal{M}(N) = \funcset{P}{\mN_0}$.

We will only consider the discrete-time semantics of such nets.
The firing rates determine stochastically which transition is fired in
a marking where multiple transitions are enabled: When transitions
$t_1,\dots,t_n\in T$ are enabled in a marking $m\in\mathcal{M}(N)$
(i.e., $^\bullet t_i\leq m$ pointwise), then transition $t_i$ fires
with probability $\Lambda(t_i)/\sum_{j=1}^n\Lambda(t_j)$, resulting in
a discrete step
$\transition{m}{t_i}{m'\coloneqq m-\pre{t_i}+\post{t_i}}$.
In particular, the firing rates have no influence on the reachability
set $\mathcal{R}(N)\coloneqq\set{m\in\mathcal{M}(N)
\mid m_0\rightarrow^*m}$ but only define the probability of reaching
certain places or markings.
Defining \enquote{empty} transitions $\transition{m}{\epsilon}{m}$ for
markings $m\in\mathcal{R}(N)$ where no transition is enabled, such a
stochastic Petri net can be interpreted as a Markov chain on the set
of markings $\mathcal{M}(N)$.

This Markov chain thus generates a (continuous) probability space over
sequences $(m_0,m_1,\dots)\in\mathcal{M}(N)^\omega$ where a sequence
is called valid if $m_0$ is the initial marking of the Petri net and
for a prefix $(m_0,\dots,m_n)$ all cones
$\set{(m'_0,m'_1,\dots)\in\mathcal{M}(N)^\omega \mid\forall
  k=0,\dots,n:m'_k=m_k}$ have non-zero probability.  We write
$\firingsequences{N}\coloneqq\set{\mu\in\mathcal{M}(N)^\omega
  \mid\mu\text{ is valid}}$ to denote the set of valid sequences.
We assume that no two transitions have the same pre- and
postconditions to have a one-to-one-correspondence between valid
sequences and firing sequences
$\mu:(\transition{m_0}{t_1}{\transition{m_1}{t_2}{\dots}})$.

For a firing sequence $\mu$, we write
$\mu^k:\transition{m_0}{t_1}{
  \transition{m_1}{t_2}{\transition{\dots}{t_k}{m_k}}}$ to denote the
finite subsequence of the first $k$ steps,
$\length(\mu)\coloneqq\min\set{k\in\mN\mid t_k=\varepsilon}-1$, for
its length, as well as
\[\places(\mu)\coloneqq\bigcup_{n=0}^\infty\supp(m_n) \qquad\qquad
  \transitions(\mu)\coloneqq\set{t_n\mid
    n\in\mN}\setminus\set{\varepsilon} \]
to denote the set of places reached in $\mu$ (or, analogously,
$\mu^k$), and the set of fired transitions in $\mu$
(independent of their firing order), respectively.

We are, furthermore, interested in the following properties of Petri
nets: A Petri net $N$ as above is called
\begin{itemize}
  \item
    \emph{ordinary} iff all transitions require and produce at most
    one token in each place
    ($\|\pre{t}\|_\infty,\|\post{t}\|_\infty \le 1$ for all $t\in T$);
  \item
    \emph{safe} iff it is ordinary and all reachable markings also
    only have at most one token in each place ($\|m\|_\infty\leq 1$
    for all $m\in\mathcal{R}(N)$);
  \item
    \emph{acyclic} iff the transitive closure $\prec_N^+$ of the
    causal relation $\prec_N$ (with $p\prec_N t$ if $\pre{t}(p) > 0$
    and $t \prec_N p$ if $\post{t}(p) > 0$) is irreflexive;
  \item an \emph{occurrence net} iff it is safe, acyclic, free of
    backward conflicts (all places have at most one predecessor
    transition, i.e., $|\{t\mid \post{t}(p) > 0|\leq 1$ for all
    $p\in P$) and self-conflicts (for $x\in P\cup T$, there exist no
    two distinct conflicting transitions $t,t'\in T$, i.e.,
    transitions sharing preconditions, on which $x$ is causally
    dependent, i.e., $t,t'\prec_{N}^+ x$), and the initial marking has
    no causal predecessors (for all $p\in P$ with $m_0(p)=1$, we have
    $\post{t}(p) = 0$ for all $t\in T$);
  \item
    \emph{free-choice}~\cite{DeselEsparza95} iff it is ordinary and
    all transitions $t,t'\in T$ are either both enabled or disabled in
    all markings (i.e., $^\bullet t=^\bullet t'$ or
    $\supp(^\bullet t)\cap\supp(^\bullet t')=\emptyset$);
  \item
    \emph{$\phi$-bounded} (for $\phi\colon \mN_0\to \mN_0$) iff all
    its runs, starting from $m_0$, have at most length
    $\phi(|P|+|T|)$, i.e., iff $\length(\mu) \le \phi(|P|+|T|)$ for
    all firing sequences $\mu\in\firingsequences{N}$.
\end{itemize}

We will abbreviate the class of free-choice occurrence Petri nets as
FCON, safe and acyclic free-choice nets as SAFC nets, and the class of
$\phi$-bounded Petri nets as $[\phi]$BPN.
Note that FCON $\subseteq$ SAFC and also
SAFC $\subseteq$ $[\mathit{id}]$BPN for the identity $\mathit{id}$.\footnote{Indeed, $[\mathit{id}]$BPN contains any safe
  and acyclic Petri net, omitting the free-choice constraint.}

\medskip

We also introduce some notation specifically for SAFC nets: As common
in the analysis of safe Petri nets, we will interpret markings as well
as pre- and postconditions of transitions as subsets of the set $P$ of
places rather than functions $P\rightarrow\{0,1\}\subseteq\mN_0$.

The set of maximal configurations will be denoted by
$\validconfs{N} \coloneqq
\set{\transitions(\mu)\mid\mu\in\firingsequences{N}}$ and
configurations by
$\validsubconfs{N} \coloneqq
\set{\transitions(\mu^k)\mid\mu\in\firingsequences{N},k\in\mN_0}$.

An important notion in the analysis of a (free-choice) net are
branching cells (see
also~\cite{BruniMM20,ab:true-concurrency-probabilistic}).  We will
define a cell to be a subset of transitions $\mC\subseteq T$ where all
transitions $t\in\mC$ share their preconditions and all
$t'\in T\setminus\mC$ share no preconditions with $t\in\mC$.  In other
words, $\mC$ is an equivalence class of a relation $\leftrightarrow$
on $T$ defined by
\[\forall t,t'\in T: t\leftrightarrow t'\Longleftrightarrow
\pre{t}=\pre{t'}.\]
We will write $\mC_t\coloneqq[t]^\leftrightarrow$ to denote the
equivalence class of transition $t\in T$ and
$\pre{\mC}\coloneqq\bigcup_{t\in\mC}\pre{t}$ as well as
$\post{\mC}\coloneqq\bigcup_{t\in\mC}\post{t}$ to denote the sets of
pre- and postplaces of $\mC$, respectively.
The set of all cells of a net $N$ is denoted by $\bc{N}$.

\paragraph*{Markov decision processes:}
A Markov decision process (MDP) is a tuple
$(S,A,\delta,\linebreak r,s_0)$ consisting of finite sets $S$, $A$ of
states and actions, a function
$\delta\colon S\times A\rightarrow\mathcal{D}(S)$ of probabilistic
transitions (where $\mathcal{D}(S)$ is the set of probability
distributions on $S$), a reward function
$r\colon S\times A\times S\rightarrow\mR$ of rewards and an initial
state $s_0\in S$ (see also~\cite{bellman1957markovian,fs:handbook-mdps}).

A policy (or strategy) for an MDP is some function
$\pi\colon S\rightarrow A$.
It has been shown that such stationary deterministic policies can act
optimally in such an (infinite-horizon) MDP setting (see also~\cite{fs:handbook-mdps}).
A policy gives rise to a Markov chain on the set of states with
transitions $s\mapsto \delta(s,\pi(s))\in\mathcal{D}(S)$.
The associated probability space is $s_0S^\omega$, the set of all
infinite paths on $S$ starting with $s_0$, which -- due to its
uncountable nature -- has to be dealt with using measure-theoretic
concepts.
As before we equip the probability space with a $\sigma$-algebra generated by
all cones, i.e., all sets of words sharing a common prefix.

The value (or payoff) of a policy $\pi$ is then given as the
expectation of the (undiscounted) total reward (where $\mathbf{s}_i$,
$i\in\mathbb{N}_0$ are random variables, mapping an infinite path to
the $i$-th state, i.e., they represent the underlying Markov chain):
\[\mE\left[\sum_{n\in\mN_0}
    r(\mathbf{s}_{n},\pi(\mathbf{s}_{n}),\mathbf{s}_{n+1})\right].\]
To avoid infinite values, we have to assume that the sum is bounded.

The problem of finding an optimal policy $\pi\colon S\rightarrow A$
for a given MDP $(S,A,\delta,r,s_0)$ with finite state and action
space is known to be solvable in polynomial time using linear
programming~\cite{fs:handbook-mdps,ldk:complexity-mdps}.

\paragraph*{Bayesian Networks:} Bayesian networks are graphical models
that give compact representations of discrete probability
distributions, exploiting the (conditional) independence of random
variables.

A (finite) probability space $(\Omega,\mP)$ consists of a finite set
$\Omega$ and a probability function $\mP\colon \Omega\to [0,1]$ such
that $\sum_{\omega\in\Omega} \mP(\omega) = 1$.
A Bayesian network~\cite{Pearl00} is a tuple $(X,\Delta,P)$ where
\begin{itemize}
  \item
    $X=(X_i)_{i=1,\dots,n}$ is a (finite) family of random variables
    $X_i\colon \Omega\rightarrow V_i$, where $V_i$ is finite.
  \item
    $\Delta\subseteq\set{1,\dots,n}\times\set{1,\dots,n}$ is an
    acyclic relation that describes dependencies between the
    variables, i.e., its transitive closure $\Delta^+$ is irreflexive.
    By $\Delta^i=\set{j\mid (j,i)\in\Delta}$ we denote the parents of
    node $i$ according to $\Delta$.
  \item
    $P=(P_i)_{i=1,\dots,n}$ is a family of probability matrices
    $P_i\colon\prod_{j\in\Delta^i}V_j\to\mathcal{D}(V_i)$, whose
    entries are given by $P_i(v_i\mid(v_j)_{j\in\Delta^i})$.
\end{itemize}

A probability function $\mP$ is consistent with such a Bayesian
network whenever for $v=(v_i)_{i=1,\dots,n}\in\prod_{i=1}^n V_i$ we
have
\[\mP(X=v)=\prod_{i=1}^n P_i(v_i\mid(v_j)_{j\in\Delta^i}).\]

The size of a Bayesian network is not just the size of the graph, but
the sum of the size of all its matrices (where the size of an
$m\times n$-matrix is $m\cdot n$).
In particular, note that a node with $k$ parents in a binary Bayesian
network (i.e., with $|V_i|=2$ for all $i$) is associated with a
$2\times 2^k$ probability matrix.

\begin{wrapfigure}[10]{r}{0.35\textwidth}
  \vspace{-1.5cm}

  \begin{center}
    \includegraphics[scale=0.5]{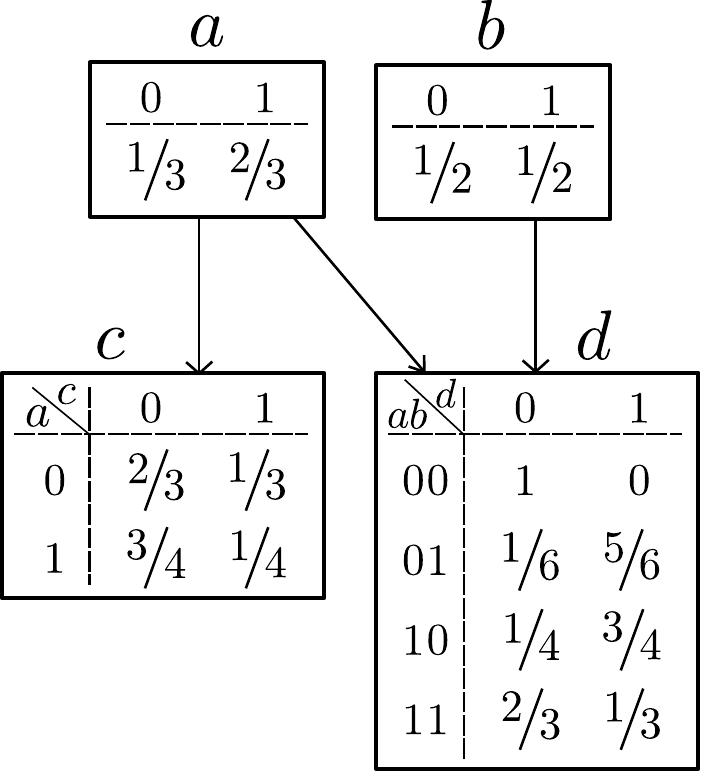}
  \end{center}
  \vspace*{-0.6cm}
  \caption{A Bayesian Network}
  \label{fig:BN}
\end{wrapfigure}
\begin{example}
  An example Bayesian network is given in Figure~\ref{fig:BN}.
  There are four random variables ($a,b,c,d$) with codomain $\{0,1\}$.
  The tables in the figure denote the conditional probabilities, for
  instance $P_d(0\mid01)=\mP(X_d=0\mid X_a=0,X_b=1)=\nicefrac{1}{6}$,
  i.e., one records the probability that a random variable has a
  certain value, dependent on the value of its parents in the
  graph.
  The probability $\mP(X = 0100) = \mP(X_a=0,X_b=1,X_c=0,X_d=0)$ is
  obtained by multiplying
  $P_a(0)\cdot P_b(1)\cdot P_c(0\mid 0)\cdot P_d(0\mid 01) =
  \nicefrac{1}{3}\cdot
  \nicefrac{1}{2}\cdot\nicefrac{2}{3}\cdot\nicefrac{1}{6} =
  \nicefrac{1}{54}$.
\end{example}

We are interested in the following two problems for Bayesian networks
(see also~\cite{park2004complexity}):

\begin{itemize}
  \item
    \dpr: Given the Bayesian network $(X,\Delta,P)$ and
    $E = \{X_{i_1},\dots,X_{i_\ell}\}\subseteq X$,
    $e\in V_E \coloneqq \prod_{j=1}^\ell V_{i_j}$ (the evidence) and a
    rational $p>0$, does it hold that $\mP(E=e) > p$?
    This problem is known to be $\pp$-complete
    ~\cite{littman2001stochastic}.
  \item
    \dmap: Given a Bayesian network $(X,\Delta,P)$, a rational
    number $p>0$, disjoint subsets $E,F\subseteq X$,\footnote{The
    variables contained in $F$ are called MAP variables.} and evidence
    $e\in V_E$, does there exist $f\in V_F$ such that
    $\mP(F=f,E=e)>p$, or, if $\mP(E=e)\neq\emptyset$, equivalently,
    $\mP(F=f\mid E=e) > p$ (by adapting the bound $p$).
    It is known that this problem, also known as maximum a-posteriori
    problem, is $\np^{\pp}$-complete
    (see~\cite{littman2001stochastic,c:new-complexity-map}).
\end{itemize}

The corresponding proof in~\cite{park2004complexity} also shows that
the \dmap\ problem remains $\np^{\pp}$-complete if $F$ only contains
uniformly distributed \enquote*{input} nodes, i.e., nodes $X_i$ with
$\Delta^i=\emptyset$ and $P_i(x_i)=1/|V_i|$, as well as
$V_i=\set{0,1}$ for all $i=1,\dots,n$.

In particular, the following problem (where $E,F$ are switched!) is
still $\np^{\pp}$-complete: Given a binary Bayesian network
$(X,\Delta,P)$ (i.e., $V_i=\set{0,1}$ for all $i$), a rational $p>0$,
disjoint subsets $E,F\subseteq X$ where $F$ only contains uniformly
distributed input nodes, as well as evidence $e\in V_E$, does there
exist $f\in V_F$ such that $\mP(E=e\mid F=f) > p$ (as
$\mP(F=f)=1/2^{|F|}$ is independent of $f$
and known due to uniformity)?
We will, in the rest of this paper, refer to this modified problem as
\dmap\ instead of the original problem above.

\begin{example}[\dmap]
  Given the Bayesian Network in Figure~\ref{fig:BN} with $F=\{X_a\}$
  (MAP variable), $E = \{X_c,X_d\}$, $e=(0,1)\in V_c\times V_d$
  (evidence) and $p=\nicefrac{1}{3}$, we ask whether
  $\exists f\in\{0,1\}\colon \mP(X_c=0, X_d=1 \mid X_a = f) >
  \nicefrac{1}{3}$.
  When choosing $f = 1\in V_a$, the probability
  $\mP(X_c=0, X_d=1 \mid X_a = 1) = \nicefrac{3}{4} \cdot (
  \nicefrac{1}{2} \cdot \nicefrac{3}{4}+ \nicefrac{1}{2} \cdot
  \nicefrac{1}{3}) = \nicefrac{13}{32} > \nicefrac{1}{3}$ exceeds the
  bound.
  Note that to compute the value in this way, one has to sum up over
  all possible valuations of those variables that are neither evidence
  nor MAP variables, indicating that this is not a trivial task.
\end{example}

\section{Stochastic decision Petri nets}\label{sec:sdpns}

We will enrich the definition of stochastic Petri nets to allow for
interactivity, similar to how MDPs~\cite{bellman1957markovian} extend
the definition of Markov chains.

\begin{definition}
  A stochastic decision Petri net (SDPN) is a tuple
  $(P,T,\pre{()},\post{()}, \linebreak \Lambda,m_0,C,R)$ where
  $(P,T,\pre{()},\post{()},\Lambda,m_0)$ is a stochastic Petri net;
  $C\subseteq T$ is a set of controllable transitions;
  $R\colon \mathcal{P}(P)\rightarrow\mR$ is a reward function.
\end{definition}

Here we describe the semantics of such SDPNs in a semi-formal way.
The precise semantics is obtained by the encoding of SDPNs into MDPs
in Section~\ref{sec:mdp}.

Given an SDPN, an external agent may in each step choose to manually
deactivate any subset $D\subseteq C$ of controllable transitions
(regardless of whether their preconditions are fulfilled or
not).
As such, if transitions $D\subseteq C$ are deactivated in marking
$m\in\mathcal{M}(N)$, the SDPN executes a step according to the
semantics of the stochastic Petri net
$N_D = (P,T\setminus D,\pre{()},\post{()},\Lambda_D,m_0)$ where the
pre- and post-set functions and $\Lambda_D$ are restricted
accordingly.

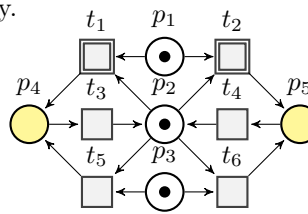
\begin{wrapfigure}[8]{r}{0.35\textwidth}
  \vspace*{-1.2cm}
  \centering
  \begin{tikzpicture}[node distance=0.9cm,>=stealth',bend angle=45,auto,every label/.style={align=left}]
  \tikzstyle{place}=[circle,thick,draw=black,fill=white,minimum size=5mm]
  \tikzstyle{transition}=[rectangle,thick,draw=black!75,
                fill=black!5,minimum size=4mm]
  \begin{scope}
    \node [place, tokens=1] (p1) [label=90:$p_1$]{};
    \node [place, tokens=1] (p2) [label=90:$p_2$, below of=p1, xshift=0mm, yshift=0mm]{};
    \node [place, tokens=1] (p3) [label=90:$p_3$, below of=p2, xshift=0mm, yshift=0mm]{};
    \node [place] (p4) [label=90:$p_4$, left of=p2, xshift=-9mm, yshift=0mm, fill=yellow!50]{};
    \node [place] (p5) [label=90:$p_5$, right of=p2, xshift=9mm, yshift=0mm, fill=yellow!50]{};

    \node [transition,double] (t1) [label=90:$t_1$, left of=p1, xshift=0mm, yshift=0mm] {}
    edge [pre] (p1)
    edge [pre] (p2)
    edge [post] (p4);

    \node [transition,double] (t2) [label=90:$t_2$, right of=p1, xshift=0mm, yshift=0mm] {}
    edge [pre] (p1)
    edge [pre] (p2)
    edge [post] (p5);

    \node [transition] (t3) [label=90:$t_3$, left of=p2, xshift=0mm, yshift=0mm] {}
    edge [pre] (p4)
    edge [post] (p2);

    \node [transition] (t4) [label=90:$t_4$, right of=p2, xshift=0mm, yshift=0mm] {}
    edge [pre] (p5)
    edge [post] (p2);

    \node [transition] (t5) [label=90:$t_5$, left of=p3, xshift=0mm, yshift=0mm] {}
    edge [pre] (p3)
    edge [pre] (p2)
    edge [post] (p4);

    \node [transition] (t6) [label=90:$t_6$, right of=p3, xshift=0mm, yshift=0mm] {}
    edge [pre] (p3)
    edge [pre] (p2)
    edge [post] (p5);

  \end{scope}
  \end{tikzpicture}
  \vspace*{-0.6cm}
  \caption{Example SDPN}\label{fig:ex-PN}
\end{wrapfigure}

For all rewarded sets $Q\in\supp(R)$, the agent receives an
\enquote{immediate} reward $R(Q)$ once all the places $p\in Q$ are
reached at one point in the execution of the Petri net (although not
necessarily simultaneously).
In particular, any reward is only received once.
Note that this differs from the usual definition of rewards as in
MDPs, where a reward is received each time certain actions is taken
in given states.
However, logical formulae over reached places (such as
\enquote{places $p_1$ and $p_2$ are reached without reaching place
  $q$}) are more natural to represent by such one-time
rewards instead of cumulative rewards.\footnote{Firings of transitions
can also easily be rewarded by adding an additional place.}
The framework can be extended to reward markings instead of places but
at the cost of an exponential explosion, since to be able to compute
the one-time step-wise rewards not only already reached places but
already reached markings would have to be memorized.
Note that a reward need not be positive.

More formally, given a firing sequence
$\mu:\transition{m_0}{t_1}{\transition{m_1}{t_2}{\dots}}$, the agent
receives a value or payoff of $\valfunc(\places(\mu))$ where
$\valfunc(M)\coloneqq\sum_{Q\subseteq M}R(Q)$.

\begin{example}\label{exa:sdpn}
  As an example consider the SDPN in Figure~\ref{fig:ex-PN}.
  The objective is to mark both places coloured in yellow at some
  point in time (not necessarily at the same time).
  This can be described by a reward function $R$ which assigns $1$ to
  the set $\{p_4,p_5\}$ containing both yellow places and $0$ to all
  other sets.

  The transitions with double borders ($t_1,t_2$) are controllable
  and it turns out that the optimal strategy is to deactive both
  $t_1$ and $t_2$ first, in order to let $t_5$ or $t_6$ mark either of
  the two goal places before reaching the marking $(1,1,0,0,0)$ from
  which no information can be gained which of the two goal places have
  been marked.
  An optimal strategy thus has to have knowledge of already achieved
  sub-goals in terms of visited places.
  In this case, the strategy can deactivate one of the transitions
  ($t_1,t_2$) leading to the place already visited.
\end{example}

Policies may be dependent on the current marking and the places
accumulated so far.
Now, for a given policy
$\pi:\mathcal{M}(N)\times \mathcal{P}(P)\rightarrow\mathcal{P}(C)$,
determining the set $\pi(m,Q)\subseteq C$ of deactivated transitions
in marking $m$ for the set $Q$ of places seen so far, we consider the
(continuous) probability space $m_0 \mathcal{M}(N)^\omega$, describing
the infinite sequence $m_0\rightarrow_{t_1}m_1\rightarrow_{t_2}\dots$
of markings generated by the Petri net under the policy $\pi$ (i.e.,
if in step $n$ the transitions
$D_{n}\coloneqq\pi(m_{n-1},\bigcup_{k=0}^{n-2}\supp(m_k))$
are deactivated).

Then we can consider the expectation of the random variable
$V\circ \places$, i.e.,
\[\val^\pi\coloneqq\mE^\pi\left[\valfunc\circ\places\right], \]
over the probability space $m_0 \mathcal{M}(N)^\omega$.
We will call this the value of $\pi$ and, if $\pi\equiv D\subseteq C$
is constant, simply write $\val^D$ which we will call the value of
$D$.

For the complexity analyses we assume that $R$ is only stored on its
support, e.g., as a set $R\subseteq\mathcal{P}(P)\times\mR$ which we
will interpret as a dictionary with entries $[Q:R(Q)]$ for some
$Q\subseteq P$, as for many problems of interest the size of the
support of the reward function can be assumed to be polynomially
bounded w.r.t.~to the set of places and transitions.

We consider the following problems for stochastic Petri nets, where we
parameterize over a class $\mathcal{N}$ of SDPNs and (for the second
problem) over a class
$\Psi\subseteq \funcset{\mathcal{M}(N)\times
  \mathcal{P}(P)}{\mathcal{P}(C)}$ of policies:

\begin{itemize}
\item \valn{\pnc}: Given a rational $p>0$, a net $N\in\mathcal{N}$ and
  a policy $\pi\in\Psi$ for $N$, decide whether $\val^\pi > p$.
\item \poln{\pnc}: Given a rational $p>0$ and a net $N\in\mathcal{N}$,
  decide whether there exist a policy $\pi\in \Psi$ such that
  $\val^\pi > p$.

  Although paramterized over sets of policies, we will omit $\Psi$ if
  is clear from the context (in fact we will restrict to constant
  policies from Section~\ref{sec:complexity} onwards).
\end{itemize}

\section{Stochastic decision Petri nets as Markov decision processes}
\label{sec:mdp}

We now describe how to transform an SDPN into an MDP, thus fixing
the semantics of such nets.
For unbounded Petri nets, the resulting MDP has an infinite state
space, but we will restrict to the finite case later.

\begin{definition}\label{def:sdpn}
  Given an SDPN $N=(P,T,F,\Lambda,C,R,m_0)$ where $m_0$ is not the
  constant zero function, the MDP for $N$ is defined as the tuple
  $(S,A,\delta,r,s_0)$ where
  \begin{itemize}
    \item
      $S = \mathcal{R}(N)\times\mathcal{P}(P)$ (product of
      reachable markings and places collected),
    \item
      $A = \mathcal{P}(C)$ (sets of deactivated transition as
      actions),
    \item
      $\delta\colon (\mathcal{R}(N)\times\mathcal{P}(P)) \times
      \mathcal{P}(C) \rightarrow
      \mathcal{D}(\mathcal{R}(N)\times\mathcal{P}(P))$, with
      \[ \delta((m,Q),D)((m',Q'))\coloneqq\begin{cases}
          p(m'\mid m,D)&\text{if }Q'=Q\cup\supp(m),\\
          0&\text{otherwise,}
        \end{cases} \] where
      \[ p(m'\mid m,D)=\frac{\sum_{t\in \mathit{En}(m,D),
        m\rightarrow_t m'}\Lambda(t)}{\sum_{t\in \mathit{En}(m,D)}
        \Lambda(t)} \]
      whenever
      $\mathit{En}(m,D) \coloneqq \{t\in T\backslash D\mid
      \pre{t}\le m\} \neq\emptyset$.
      If $\mathit{En}(m,D) = \emptyset$, we set $p(m'\mid m,D) = 1$
      if $m=m'$ and $0$ if $m\neq m'$.
      That is, $p(m'\mid m, D)$ is the probability of reaching $m'$
      from $m$ when transitions $D$ are deactivated.
    \item
      $r\colon S\times A\times S\to \mathbb{R}$ (reward function) with
      \[ r((m,Q),D,(m',Q'))\coloneqq
          \begin{cases}
            \sum_{Q\subseteq Y\subseteq Q'}R(Y)
              &\text{if }Q=\emptyset,\\
            \sum_{Q\subsetneq Y\subseteq Q'}R(Y)
              &\text{if }Q\neq\emptyset.
        \end{cases}\]
    \item $s_0 = (m_0,\emptyset)$
  \end{itemize}
\end{definition}

The transition probabilities are determined as for regular stochastic
Petri nets where we consider only the rates of those transitions that
have not been deactivated and that can be fired for the given marking.
If no transition is enabled, we stay at the current marking with
probability $1$.

Note that the reward for the places reached in a marking $m$ is only
collected when we fire a transition leaving $m$.
This is necessary as in the very first step we also obtain the reward
for the empty set, which might be non-zero, and due to the fact that
the initial marking is assumed to be non-empty, this reward for the
empty set is only collected once.

The following result shows that the values of policies
$\pi:S\rightarrow A$ (note that these are exactly the policies for the
underlying SDPN) over the MDP are equal to the ones over the
corresponding SDPN.

\begin{propositionrep}
  Let $N=(P,T,F,\Lambda,C,R,m_0)$ be an SDPN and
  $M=(S,A,\delta,\linebreak r,s_0)$ the corresponding MDP.
  For any policy $\pi:S\rightarrow A$, we have
  \[(\val^\pi=)\mE^\pi\left[\valfunc\circ \places\right]
    =\mE^\pi\left[\sum_{n\in\mN_0}r(\mathbf{s}_{n},
      \pi(\mathbf{s}_n),\mathbf{s}_{n+1})\right]\]
  where $(\mathbf{s}_n)_n$ is the Markov chain resulting from
  following policy $\pi$ in $M$.
\end{propositionrep}

\begin{proof}
  Consider a sequence of states $s_1,\dots, s_n$, $n\in\mN_0$ of the
  MDP where $s_i = (m_i,Q_i)$, $m_i\in\mathcal{R}(N)$, $Q_i\subseteq
  P$.
  Using the notation of Definition~\ref{def:sdpn} we obtain
  \begin{align*}
    &\mP^{\pi}(\mathbf{s}_1=s_1,\dots,\mathbf{s}_n=s_n)\\
    &=\prod_{k=1}^n\delta(s_{k-1},\pi(s_{k-1}))(s_k) \\
    &=\begin{cases}
      \prod_{k=1}^n p(m_k\mid m_{k-1},\pi(s_{k-1}))
        &\text{if for all }k=1,\dots,n:Q_k=\bigcup_{i=0}^{k-1}
          \supp(m_i),\\
      0&\text{otherwise}
    \end{cases}\\
    &=\begin{cases}
        \mP^{\pi}((m_1,\dots,m_n))
          &\text{if for all }k=1,\dots,n:Q_k=\bigcup_{i=0}^{k-1}
            \supp(m_i),\\
      0&\text{otherwise}
    \end{cases}
  \end{align*}
  where $\mP^{\pi}((m_1,\dots,m_n))$ is the probability of the
  sequence $(m_1,\dots,m_n)$ in the Petri net $N$ under policy $\pi$.

  The result can then be seen from the following equations
  \begin{align*}
    &\mE^\pi\left[\sum_{n\in\mN}
      r(\mathbf{s}_{n-1},\pi(\mathbf{s}_{n-1}),\mathbf{s}_{n})
    \right]\\
    &=\lim_{n\rightarrow\infty}\mE^\pi\left[\sum_{k=1}^n
      r(\mathbf{s}_{k-1},\pi(\mathbf{s}_{k-1}),\mathbf{s}_{k})
    \right]\\
    &=\lim_{n\rightarrow\infty}\sum_{s_1,\dots,s_{n}\in S}
      \mP^{\pi}(\mathbf{s}_1=s_1,\dots,\mathbf{s}_n=s_n)
      \sum_{k=1}^nr(s_{k-1},\pi(s_{k-1}),s_{k})\\
    &=\lim_{n\rightarrow\infty}\sum_{m_1,\dots,m_{n}\in\mathcal{M}(N)}
      \mP^\pi((m_0,\dots,m_n))
      \sum_{Q\subseteq\bigcup_{k=0}^n\supp(m_k)}R(Q)\\
    &=\lim_{n\rightarrow\infty}\sum_{M\subseteq P}
      \mP^\pi(\places_{\leq n}=M)\sum_{Q\subseteq M}R(Q)\\
    &=\sum_{M\subseteq P}\mP^\pi(\places=M)\sum_{Q\subseteq M}R(Q)\\
    &=\mE^\pi\left[\valfunc\circ\places\right]\\
  \end{align*}
  where we use in the first equation that the random variables
  \[\mathbf{r}_n\coloneqq\sum_{k=1}^n
    r(\mathbf{s}_{k-1},\pi(\mathbf{s}_{k-1}),\mathbf{s}_{k})\]
  are uniformly bounded by construction.\qed
\end{proof}

\begin{toappendix}
  \begin{example}
    We consider the Petri net from Example~\ref{exa:sdpn} and spell
    out its corresponding MDP (see Figure~\ref{fig:ex-PN-to-mdp}).

  \begin{figure} [h]
    \centering
    \begin{tikzpicture}[node distance=2.4cm,>=stealth',bend angle=45,auto,every label/.style={align=left} ]
    \tikzstyle{place}=[ellipse,thick,draw=black,fill=white]

    \begin{scope}
      \node [place](s0){$s_0$};
      \node [place, right of=s0, xshift=-16mm, yshift=16mm] (s3){$s_3$};
      \node [place, right of=s0, xshift=12mm, yshift=6mm] (s1){$s_1$};
      \node [place, right of=s0, xshift=12mm, yshift=-6mm] (s2){$s_2$};
      \node [place, right of=s0, xshift=-16mm, yshift=-16mm] (s4){$s_4$};

      \node [place, right of=s3, xshift=-16mm, yshift=16mm] (s7){$s_7$};
      \node [place, right of=s1, xshift=-12mm] (s5){$s_5$};
      \node [place, right of=s2, xshift=-12mm] (s6){$s_6$};
      \node [place, right of=s4, xshift=-16mm, yshift=-16mm] (s8){$s_8$};

      \node [place, right of=s7, xshift=-16mm, yshift=-14mm] (s13){$s_{13}$};
      \node [place, right of=s7, xshift=8mm] (s14){$s_{14}$};
      \node [place, above of=s5, yshift=-12mm] (s9){$s_9$};
      \node [place, right of=s5, xshift=-4mm] (s10){$s_{10}$};
      \node [place, right of=s6, xshift=-4mm] (s11){$s_{11}$};
      \node [place, below of=s6, yshift=12mm] (s12){$s_{12}$};
      \node [place, right of=s8, xshift=8mm] (s15){$s_{15}$};
      \node [place, right of=s8, xshift=-16mm, yshift=14mm] (s16){$s_{16}$};

      \node [place, left of=s9, xshift=12mm] (s17){$s_{17}$};
      \node [place, right of=s11, xshift=-8mm, yshift=6mm] (s18){$s_{18}$};
      \node [place, left of=s12, xshift=12mm] (s19){$s_{19}$};

      \draw[thick,->] (s0) -- (s3) node[midway,left] {$\frac{1}{4-|\pi(s_0)|}$ [$t_5$]};
      \draw[thick,->] (s0) -- (s1) node[midway,above] {$\frac{\chi_{T\setminus\pi(s_0)}(t_1)}{4-|\pi(s_0)|}$ [$t_1$]};
      \draw[thick,->] (s0) -- (s2) node[midway,below] {$\frac{\chi_{T\setminus\pi(s_0)}(t_2)}{4-|\pi(s_0)|}$ [$t_2$]};
      \draw[thick,->] (s0) -- (s4) node[midway,left] {$\frac{1}{4-|\pi(s_0)|}$ [$t_6$]};

      \draw[thick,->] (s3) -- (s7) node[midway,left] {[$t_3$]};
      \draw[thick,->] (s1) -- (s5) node[midway,below] {[$t_3$]};
      \draw[thick,->] (s2) -- (s6) node[midway,above] {[$t_4$]};
      \draw[thick,->] (s4) -- (s8) node[midway,left] {[$t_4$]};

      \path (s7) edge [loop left] node {$\chi_{\pi(s_7)}(t_1)\chi_{\pi(s_7)}(t_2)$ [$\varepsilon$]} (s7);
      \path (s8) edge [loop left] node {$\chi_{\pi(s_8)}(t_1)\chi_{\pi(s_8)}(t_2)$ [$\varepsilon$]} (s8);
      \draw[thick,->] (s5) -- (s9) node[midway,right] {$\nicefrac{1}{2}$ [$t_5$]};
      \draw[thick,->] (s5) -- (s10) node[midway,below] {$\nicefrac{1}{2}$ [$t_6$]};
      \draw[thick,->] (s6) -- (s11) node[midway,above] {$\nicefrac{1}{2}$ [$t_5$]};
      \draw[thick,->] (s6) -- (s12) node[midway,right] {$\nicefrac{1}{2}$ [$t_6$]};
      \draw[thick,->] (s7) -- (s13) node[midway,right] {$\frac{\chi_{T\setminus\pi(s_7)}(t_1)}{2-|\pi(s_7)|}$ [$t_1$]};
      \draw[thick,->] (s7) -- (s14) node[midway,above] {$\frac{\chi_{T\setminus\pi(s_7)}(t_2)}{2-|\pi(s_7)|}$ [$t_2$]};
      \draw[thick,->] (s8) -- (s15) node[midway,below] {$\frac{\chi_{T\setminus\pi(s_8)}(t_1)}{2-|\pi(s_8)|}$ [$t_1$]};
      \draw[thick,->] (s8) -- (s16) node[midway,right] {$\frac{\chi_{T\setminus\pi(s_8)}(t_2)}{2-|\pi(s_8)|}$ [$t_2$]};
      \draw[thick,->] (s9) -- (s17) node[midway,below] {[$t_3$]};
      \draw[thick,->] (s10) -- (s18) node[midway,above] {[$t_4$]};
      \draw[thick,->] (s11) -- (s18) node[midway,below] {[$t_3$]};
      \draw[thick,->] (s12) -- (s19) node[midway,above] {[$t_4$]};
      \draw[thick,->] (s13) -- (s17) node[midway,below] {[$t_3$]};
      \draw[thick,->] (s14) to[bend left=20]  node[midway,above] {[$t_4$]} (s18);
      \draw[thick,->] (s15) to[bend right=20] node[midway,below] {[$t_3$]} (s18) ;
      \draw[thick,->] (s16) -- (s19) node[midway,above] {[$t_4$]};

    \end{scope}
    \end{tikzpicture}
    \caption{MDP from example Petri net}\label{fig:ex-PN-to-mdp}
  \end{figure}
  \begin{figure} [h]
    \centering
    \begin{tikzpicture}[node distance=1.2cm,>=stealth',bend angle=45,auto,every label/.style={align=left} ]
    \tikzstyle{place}=[circle,thick,draw=black,fill=white]

    \begin{scope}
      \node [place](s0){$s_0$};
      \node [place, right of=s0, yshift=6mm] (s3){$s_3$};
      \node [place, right of=s0, yshift=-6mm] (s4){$s_4$};

      \node [place, right of=s3] (s7){$s_7$};
      \node [place, right of=s4] (s8){$s_8$};

      \node [place, right of=s7] (s14){$s_{14}$};
      \node [place, right of=s8] (s15){$s_{15}$};

      \node [place, right of=s14, yshift=-6mm] (s18){$s_{18}$};

      \draw[thick,->] (s0) -- (s3) node[midway,above] {$\nicefrac{1}{2}$};
      \draw[thick,->] (s0) -- (s4) node[midway,below] {$\nicefrac{1}{2}$};

      \draw[thick,->] (s3) -- (s7) node[midway,above] {};
      \draw[thick,->] (s4) -- (s8) node[midway,below] {};

      \draw[thick,->] (s7) -- (s14) node[midway,above] {};
      \draw[thick,->] (s8) -- (s15) node[midway,below] {};

      \draw[thick,->] (s14) -- (s18) node[midway,above] {};
      \draw[thick,->] (s15) -- (s18) node[midway,below] {};

    \end{scope}
    \end{tikzpicture}
    \caption{Markov chain of optimal strategy in MDP for example Petri net.}\label{fig:ex-PN-optimal-strategy}
  \end{figure}
  \begin{figure} [h]
    \centering
    \begin{tikzpicture}[node distance=2.4cm,>=stealth',bend angle=45,auto,every label/.style={align=left} ]
    \tikzstyle{place}=[circle,thick,draw=black,fill=white]

    \begin{scope}
      \node [place](m0){$\begin{psmallmatrix}1\\1\\1\\0\\0\end{psmallmatrix}$};
      \node [place, right of=m0, yshift=36mm] (m1){$\begin{psmallmatrix}0\\0\\1\\1\\0\end{psmallmatrix}$};
      \node [place, right of=m0, yshift=12mm] (m2){$\begin{psmallmatrix}0\\0\\1\\0\\1\end{psmallmatrix}$};
      \node [place, right of=m0, yshift=-12mm] (m3){$\begin{psmallmatrix}1\\0\\0\\1\\0\end{psmallmatrix}$};
      \node [place, right of=m0, yshift=-36mm] (m4){$\begin{psmallmatrix}1\\0\\0\\0\\1\end{psmallmatrix}$};

      \node [place, right of=m2] (m5){$\begin{psmallmatrix}0\\1\\1\\0\\0\end{psmallmatrix}$};
      \node [place, right of=m3] (m6){$\begin{psmallmatrix}1\\1\\0\\0\\0\end{psmallmatrix}$};

      \node [place, right of=m5] (m7){$\begin{psmallmatrix}0\\0\\0\\1\\0\end{psmallmatrix}$};
      \node [place, right of=m6] (m8){$\begin{psmallmatrix}0\\0\\0\\0\\1\end{psmallmatrix}$};

      \node [place, right of=m7, yshift=-12mm] (m9){$\begin{psmallmatrix}0\\1\\0\\0\\0\end{psmallmatrix}$};

      \draw[thick,->] (m0) -- (m1) node[midway,above] {[$t_1$]};
      \draw[thick,->] (m0) -- (m2) node[midway,below] {[$t_2$]};
      \draw[thick,->] (m0) -- (m3) node[midway,below] {[$t_5$]};
      \draw[thick,->] (m0) -- (m4) node[midway,below] {[$t_6$]};

      \draw[thick,->] (m1) -- (m5) node[midway,above] {[$t_3$]};
      \draw[thick,->] (m2) -- (m5) node[midway,above] {[$t_4$]};
      \draw[thick,->] (m3) -- (m6) node[midway,above] {[$t_3$]};
      \draw[thick,->] (m4) -- (m6) node[midway,above] {[$t_4$]};

      \draw[thick,->] (m5) -- (m7) node[midway,above] {[$t_5$]};
      \draw[thick,->] (m5) -- (m8) node[pos=0.3,above] {[$t_6$]};
      \draw[thick,->] (m6) -- (m7) node[pos=0.1,above] {[$t_1$]};
      \draw[thick,->] (m6) -- (m8) node[midway,above] {[$t_2$]};

      \draw[thick,->] (m7) -- (m9) node[midway,above] {[$t_3$]};
      \draw[thick,->] (m8) -- (m9) node[midway,below] {[$t_4$]};

    \end{scope}
    \end{tikzpicture}
    \caption{Reachability graph of example Petri net
      (Example~\ref{exa:sdpn}) }
    \label{fig:ex-PN-reachability}
  \end{figure}
  We note that the reachability graph of the net consists of ten
  markings:
  \[m_0=\begin{psmallmatrix}1\\1\\1\\0\\0\end{psmallmatrix},\quad
    m_1=\begin{psmallmatrix}0\\0\\1\\1\\0\end{psmallmatrix},\quad
    m_2=\begin{psmallmatrix}0\\0\\1\\0\\1\end{psmallmatrix},\quad
    m_3=\begin{psmallmatrix}1\\0\\0\\1\\0\end{psmallmatrix},\quad
    m_4=\begin{psmallmatrix}1\\0\\0\\0\\1\end{psmallmatrix},\]
  \[m_5=\begin{psmallmatrix}0\\1\\1\\0\\0\end{psmallmatrix},\quad
    m_6=\begin{psmallmatrix}1\\1\\0\\0\\0\end{psmallmatrix},\quad
    m_7=\begin{psmallmatrix}0\\0\\0\\1\\0\end{psmallmatrix},\quad
    m_8=\begin{psmallmatrix}0\\0\\0\\0\\1\end{psmallmatrix},\quad
    m_9=\begin{psmallmatrix}0\\1\\0\\0\\0\end{psmallmatrix}\]
  (see also Figure~\ref{fig:ex-PN-reachability}).
  Firing rates are constant: $\Lambda\equiv 1$.
  The corresponding MDP has twenty states:
  \begin{gather*}
    s_0=(m_0,\emptyset),
    s_1=(m_1,M^{-,-}),
    s_2=(m_2,M^{-,-}),
    s_3=(m_3,M^{-,-}),\\
    s_4=(m_4,M^{-,-}),
    s_5=(m_5,M^{+,-}),
    s_6=(m_5,M^{-,+}),
    s_7=(m_6,M^{+,-}),\\
    s_8=(m_6,M^{-,+}),
    s_9=(m_7,M^{+,-}),
    s_{10}=(m_8,M^{+,-}),
    s_{11}=(m_7,M^{-,+}),\\
    s_{12}=(m_8,M^{-,+}),
    s_{13}=(m_7,M^{+,-}),
    s_{14}=(m_8,M^{+,-}),
    s_{15}=(m_7,M^{-,+}),\\
    s_{16}=(m_8,M^{-,+}),
    s_{17}=(m_9,M^{+,-}),
    s_{18}=(m_9,M^{+,+}),
    s_{19}=(m_9,M^{-,+}),\\
  \end{gather*}
  (see also Figure~\ref{fig:ex-PN-to-mdp}) where
  $M^{-,-}=\set{p_1,p_2,p_3}$ signifies that none of the places $p_4$
  and $p_5$ have been reached previously,
  $M^{+,-}=\set{p_1,p_2,p_3,p_4}$ and $M^{-,+}=\set{p_1,p_2,p_3,p_5}$
  that $p_4$ or $p_5$ were reached, respectively, and
  $M^{+,+}=\set{p_1,p_2,p_3,\linebreak p_4,p_5}$ that both were
  reached (and, thus, the reward received).

  As previously remarked for the Petri net, in this MDP, the optimal
  strategy is to deactivate both $t_1$ and $t_2$ in $s_0$ and then
  activate only whichever transition leads to the place that was not
  yet visited, yielding the Markov chain in
  Figure~\ref{fig:ex-PN-optimal-strategy}.
\end{example}
\end{toappendix}

This provides an exact semantic for SDPNs via MDPs. Note, however,
that for analysis purposes, even for safe Petri nets, the reachability
set $\mathcal{R}(N)$ (as a subset of $\mathcal{P}(P)$) is generally of
exponential size whence the transformation into an MDP can at best
generally only yield algorithms of exponential worst-case-time.
Hence, we will now restrict to specific subproblems and it will turn
out that even with fairly severe restrictions to the type of net and
the policies allowed, we obtain completeness results for complexity
classes high in the polynomial hierarchy.

\section{Complexity analysis for specific classes of Petri nets}
\label{sec:complexity}

For the remainder of this paper, we will consider the problem of
finding optimal \emph{constant} policies for certain classes of nets.
In other words, the agent chooses \emph{before} the execution of the
Petri net which transitions to deactivate for its \emph{entire}
execution.
For a net $N$, the policy space is thus given by
\[ \Psi(N) =
  \set{\pi:\mathcal{M}(N)\rightarrow\mathcal{P}(C)\mid\pi\equiv
    D\subseteq C}\text{ }\hat{=}\text{ }\mathcal{P}(C). \]

Since one can non-deterministically guess the maximizing policy (there
are only exponentially many) and compute its value, it is clear that
the complexity of the policy optimization problem \poln{\pnc} is
bounded by the complexity of the corresponding value problem
\valn{\pnc} as follows: If, for a given class $\pnc$ of Petri nets,
\valn{\pnc} lies in the complexity class $\mathsf{C}$, then
\poln{\pnc} lies in $\np^\mathsf{C}$.

We will now show the complexity of these problems for the three Petri
net classes FCON, SAFC, and $[\phi]$BPN and work out the connection to
Bayesian networks.
In the following we will assume that all probabilities are efficiently
computable, allowing us to simulate all probabilistic choices with
fair coins.

\subsection{Complexity of safe and acyclic free-choice decision nets}
\label{sec:safc-complexity}

We will first consider the case of Petri nets where the length of runs
is bounded.

\begin{propositionrep}\label{prop:pbpn-np-pp}
  For any polynomial $\phi$, the problem \valn{[\phi]BPN} is in $\pp$.
  In particular, \poln{[\phi]BPN} is in $\np^{\pp}$.
\end{propositionrep}

\begin{proofsketch}
  Given a Petri net $N$, a policy $\pi$ and a bound $p$, a
  $\pp$-algorithm for \valn{[\phi]BPN} can simulate the execution of
  the Petri net and calculate the resulting value, checking whether
  the expected value for $\pi$ is greater than the pre-defined bound
  $p$.  For this, we have to suitably adapt the threshold (with an
  affine function $\psi$) so that the probabilistic Turing machine
  accepts with probability greater than $\nicefrac{1}{2}$ iff the
  reward for the given policy is strictly greater than $p$.

  As the execution of the Petri net takes only polynomial time in the
  size of the Petri net ($\phi$), this can be performed by a
  probabilistic Turing machine in polynomial time whence
  \valn{[\phi]BPN} lies in $\pp$.

  Since a policy can be guessed in polynomial time, we can also infer
  that \poln{[\phi]BPN} is in $\np^{\pp}$.
  \qed
\end{proofsketch}

\begin{proof}
  We  give a $\pp$-algorithm that solves \valn{[\phi]BPN}, i.e.,
  given a $[\phi]$-bounded stochastic Petri net $N$ and a lower bound
  $p\in\mR$, determine whether the expected value of the net is
  greater than $p$.

  As, by definition, the length of runs of the net $N$ is bounded by
  $\phi(|T|+|P|)$, each run either terminates after a polynomial
  number of steps.
  Hence we can simulate the net with a probabilistic Turing machine
  and whenever the run terminates, we can determine the value for this
  run.

  Hence the main difficulty is the following: A probabilistic Turing
  machine can check whether a given property is satisfied in more than
  half of all runs of the algorithm.
  As such, it could, for example, easily check whether, in more than
  half of all runs of the Petri net, its value (or payoff) is greater
  than a pre-defined threshold.
  Our goal, however, is to check whether the \textit{expected value}
  of the total reward is greater than a threshold.

  We show that this can also be checked in polynomial time.
  We calculate
  \[\valfunc_{\min}\coloneqq\sum_{Q\in\supp(R^-)}R(Q)\]
  where $R^-$ is the negative part of the reward function $R$ (i.e.,
  $\valfunc_{\min}$ is the sum of all negative rewards), and similarly
  \[\valfunc_{\max}\coloneqq\sum_{Q\in\supp(R^+)}R(Q)\]
  where $R^+$ is the positive part of $R$.

  As any reward can only be granted once, the value of any run
  of the Petri net (and, in particular, its value $\val$) lies in
  $[\valfunc_{\min},\valfunc_{\max}]$.
  If $p\notin[\valfunc_{\min},\valfunc_{\max}]$, we can, therefore,
  already safely output that $\val>p$ is true (if $p<\valfunc_{\min}$)
  or false (if $p>\valfunc_{\max}$).

  If $p\in[\valfunc_{\min},\valfunc_{\max}]$, we define an affine
  linear transformation
  \[\psi\colon\mR\rightarrow\mR,x \mapsto \frac{x-\valfunc_{\min}}{\valfunc_{\max}-\valfunc_{\min}}\]
  which maps any possible value to a number in $[0,1]$, and define
  $\tilde{p}\coloneqq\psi(p)\in[0,1]$.

  The probabilistic Turing machine can then simulate a run $\mu$ of
  the Petri net and calculate the total value for its execution which
  can be done in polynomial time by definition of $[\phi]$BPN.
  Then, the linear transformation $\psi$ is applied to the value
  $V_\mu$ which yields $\rho_\mu = \psi(V_\mu)\in[0,1]$.
  Assume for now that we terminate with probability $\rho_\mu$ with
  success and otherwise with a failure.
  This ensures that the algorithm terminates with probability greater
  than $\tilde{p}$ if and only if the value $\psi(V(\places(\mu)))$ of
  the run $\mu$ is greater than $\tilde{p} = \psi(p)$, which is
  equivalent to the value $\val$ being greater than $p$.
  In more detail:
  \begin{align*}
    &\tilde{p}
    < \sum_{\mu}\mP(\mu) \cdot
    \rho_\mu = \sum_{\mu}\mP(\mu) \cdot
    \psi(\valfunc(\places(\mu)))
    = \mathbb{E}[\psi(\valfunc(\places(\mu)))]\\
    \Leftrightarrow \quad & \psi(p) <
    \psi(\sum_{\mu}\mP(\mu)\cdot \valfunc(\places(\mu)))
    =\psi(\mathbb{E}[\valfunc(\places(\mu))])\\
    \Leftrightarrow \quad &p<\sum_{\mu}\mP(\mu)\cdot
    \valfunc(\places(\mu))=\mathbb{E}[\valfunc(\places(\mu))]=\val
  \end{align*}
  where we use that $\psi$ is a strictly monotone affine linear
  function (and thus, in particular, commutes with expected values).

  Note however that it should be the case that a word is in the
  language iff the $\pp$ machine accepts it with probability strictly
  greater than $\nicefrac{1}{2}$.

  To adapt the threshold accordingly, we define
  \[ \sigma\coloneqq(\nicefrac{1}{2} -
    \tilde{p})/(|\nicefrac{1}{2} -
    \tilde{p}|+\nicefrac{1}{2}),\]
  and insert an additional step before simulating the net:

  If $\sigma<0$, the algorithm outputs false with probability
  $-\sigma$ (and continues its execution otherwise), and if
  $\sigma\geq 0$, the algorithm outputs true with probability $\sigma$
  (and continues its execution otherwise).
  This ensures that the probability of reaching a success in the rest
  of the algorithm has to be at least $\tilde{p}$ instead of
  $\nicefrac{1}{2}$.

  This can be seen as follows, where $p_S$ denotes the probability of
  success for the rest of the execution:
  \begin{itemize}
    \item
      $\tilde{p}>\nicefrac{1}{2}$, which implies $\sigma<0$: then the
      success probability is $(1+\sigma)p_S$ and we have
      \[ \frac{1}{2}<(1+\sigma)p_S \iff \frac{1}{2} <
        (1+\frac{\nicefrac{1}{2}-\tilde{p}}{\tilde{p}})p_S
        \iff \frac{1}{2} < \frac{1}{2\tilde{p}}p_S \iff
        \tilde{p} < p_S \]
    \item
      $\tilde{p}\leq\nicefrac{1}{2}$, which implies $\sigma\geq 0$:
      then the success probability is $\sigma + (1-\sigma)p_S$ and we
      have
      \begin{align*}
        && \frac{1}{2} < \sigma+(1-\sigma)p_S \iff \frac{1}{2} <
        \frac{\nicefrac{1}{2}-\tilde{p}}{1-\tilde{p}}+(1-\frac{\nicefrac{1}{2}-\tilde{p}}{1-\tilde{p}})p_S
        \iff \\
        && \frac{\nicefrac{1}{2}\cdot(1-\tilde{p})}{1-\tilde{p}} <
        \frac{\nicefrac{1}{2}-\tilde{p}}{1-\tilde{p}}+\frac{(1-\tilde{p})-(\nicefrac{1}{2}-\tilde{p})}{1-\tilde{p}})p_S
        \iff \\
        && \frac{\nicefrac{1}{2}\cdot\tilde{p}}{1-\tilde{p}} <
        \frac{\nicefrac{1}{2}}{1-\tilde{p}}p_S \iff
        \tilde{p} < p_S,
      \end{align*}
  \end{itemize}
  which concludes the proof. \qed
\end{proof}

This easily gives us the following corollary for SAFC nets.

\begin{corollary}\label{cor:safc-pp}
  The problem \safcval\ is in $\pp$ and \safcpol\ in $\np^{\pp}$.
\end{corollary}

\begin{proof}
  This follows directly from Proposition~\ref{prop:pbpn-np-pp} and the
  fact that SAFC $\subseteq$ $[id]$BPN. \qed
\end{proof}

\begin{propositionrep}\label{prop:safc-np-pp-hard}
  The problem \safcpol\ is $\np^{\pp}$-hard and, therefore, also
  $\np^{\pp}$-complete.
\end{propositionrep}

\begin{proofsketch}
  This can be proven via a reduction $\dmap \leq_p \safcpol$, i.e.,
  representing the modified \dmap\ problem for Bayesian networks as a
  decision problem in safe and acyclic free-choice nets.
  $\np^{\pp}$-completeness then follows together with
  Corollary~\ref{cor:safc-pp}.
  Note that we are using the restricted version of the \dmap\ problem
  as explained in Section~\ref{sec:preliminaries} (uniformly
  distributed input nodes, binary values).

  We sketch the reduction via an example: we take the Bayesian network
  in Figure~\ref{fig:BN} and consider a \dmap\ instance where
  $E = \{X_c, X_d\}$ (evidence, where we fix the values of $c,d$ to be
  $0,1$), $F=\{X_a\}$ (MAP variables) and $p$ is a threshold.
  That is, the question being asked for the Bayesian network is
  whether there exists a value $x$ such that
  $\mP(X_c = 0, X_d = 1\mid X_a = x) > p$.

  
  \noindent
  \begin{minipage}{0.37\linewidth}
    \quad
    This Bayesian network is encoded into the SAFC net in 
    Figure~\ref{fig:bn-safc-reduction}, where transitions with double
    borders are controllable and the yellow places give a reward of
    $1$ when both are reached (not necessarily at the same time).
    Transitions either have an already indicated rate of $1$ or the
    rate can be looked up in the corresponding matrix of the BN.
    The rate of a transition $t^i_{x_1 x_2 \rightarrow x_3}$ is the
    probability value $P_i(x_3 \mid x_1 x_2)$, where $P_i$ is the
    probability matrix for $i\in \{a,b,c,d\}$.

    \quad
    Intuitively the first level of
    transitions \hfill simulates \hfill the
    \vspace*{1mm}
  \end{minipage}
  \begin{minipage}{0.63\linewidth}
    \vspace*{-1cm}
    \begin{figure}[H]
      \centering
        \begin{tikzpicture}[node distance=1.8cm,>=stealth',bend angle=45,auto,every label/.style={align=left}]
          \tikzstyle{place}=[circle,thick,draw=black,fill=white,minimum size=5mm]
          \tikzstyle{transition}=[rectangle,thick,draw=black!75,
          fill=black!5,minimum size=4mm]
    
          \begin{scope}
            \node [place, tokens=1] (start) [label=left:$P_{()}^{(\bot, a)}$]{};
            \node [place, tokens=1] (start1) [right of=start, xshift=1.5cm, label=left:$P_{()}^{(\bot, b)}$]{};
            \node [place] (a0) [below of=start, xshift=-7mm, label=left:$P_0^{a}$]{};
            \node [place] (a1) [below of=start, xshift=9mm, label=right:$P_1^{a}$]{};
            \node [place] (b0) [below of=start1, xshift=-6mm, label=left:$P_0^{b}$]{};
            \node [place] (b1) [below of=start1, xshift=5mm, label=right:$P_1^{b}$]{};
            \node [place, draw=white] (helper2) [below of=b1, yshift=0mm]{$\dots$};
    
            \node [place] (ac0) [below of=a0, yshift=-12mm, xshift=-10mm, label={180:\tiny{$P_{0}^{a,c}$}}]{};
            \node [place] (ad00) [below of=a0, yshift=1mm, label={-90:\tiny{$P_{00}^{a,d}$}}]{};
            \node [place] (ad01) [below of=a0, yshift=1mm, xshift=0.9cm, label={90:\tiny{$P_{01}^{a,d}$}}]{};
    
            \node [place] (ac1) [below of=a0, yshift=-12mm, xshift=9mm, label={180:\tiny{$P_{1}^{a,c}$}}]{};
            \node [place] (ad10) [below of=a1, xshift=1mm, label={-90:\tiny{$P_{10}^{a,d}$}}]{};
            \node [place] (ad11) [below of=a1, xshift=8mm, label={-90:\tiny{$P_{11}^{a,d}$}}]{};
    
            \node [place] (bd00) [below of=b0, xshift=-3mm, label={95:\tiny{$P_{00}^{b,d}$}}]{};
            \node [place] (bd10) [below of=b0, xshift=3mm, label={-90:\tiny{$P_{10}^{b,d}$}}]{};
    
            \node [place] (c0) [below of=a0, yshift=-2.7cm, xshift=-0.5cm, label={left:$P_{0}^{c}$}, fill=yellow!50]{};
            \node [place] (c1) [below of=a0, yshift=-2.7cm, xshift=0.5cm, label={right:$P_{1}^{c}$}]{};
    
            \node [place] (d0) [right of=c1, xshift=0.3cm, label={left:$P_{0}^{d}$}]{};
            \node [place] (d1) [right of=d0, xshift=-6mm, label={right:$P_{1}^{d}$}, fill=yellow!50]{};
    
            \node [transition, double] (control1) [below of=start, yshift=9mm, xshift=-7mm] {$t_{() \rightarrow 0}^a$}
            edge [pre] (start)
            edge [post] (a0);
    
            \node [transition, double] (control2) [below of=start, yshift=9mm, xshift=9mm] {$t_{() \rightarrow 1}^a$}
            edge [pre] (start)
            edge [post] (a1);
    
            \node [transition] (init0) [below of=start1, yshift=9mm, xshift=-6mm] {$t_{() \rightarrow 0}^b$}
            edge [pre] (start1)
            edge [post] (b0);
    
            \node [transition] (init1) [below of=start1, yshift=9mm, xshift=5mm] {$t_{() \rightarrow 1}^b$}
            edge [pre] (start1)
            edge [post] (b1);
    
            \node [transition] (ta0) [below of=a0, yshift=9mm, label={[text=blue]180:$1$}] {$t_{0}^{a}$}
            edge [pre] (a0)
            edge [post] (ac0)
            edge [post] (ad00)
            edge [post] (ad01);
    
            \node [transition] (ta1) [below of=a1, yshift=9mm, label={[text=blue]0:$1$}] {$t_{1}^{a}$}
            edge [pre] (a1)
            edge [post] (ac1)
            edge [post] (ad10)
            edge [post] (ad11);
    
            \node [transition] (tb0) [below of=b0, yshift=9mm, label={[text=blue]180:$1$}] {$t_{0}^{b}$}
            edge [pre] (b0)
            edge [post] (bd00)
            edge [post] (bd10);
    
            \node [transition] (tb1) [below of=b1, yshift=9mm, label={[text=blue]0:$1$}] {$t_{1}^{b}$}
            edge [pre] (b1);
    
            \node [transition] (tc0_0) [below of=a0, yshift=-1.9cm, xshift=-1.5cm] {$t_{0 \rightarrow 0}^{c}$}
            edge [pre] (ac0)
            edge [post] (c0);
    
            \node [transition] (tc0_1) [below of=a0, yshift=-1.9cm, xshift=-0.5cm] {$t_{0 \rightarrow 1}^{c}$}
            edge [pre] (ac0)
            edge [post] (c1);
    
            \node [transition] (tc1_0) [below of=a0, yshift=-1.9cm, xshift=0.5cm] {$t_{1 \rightarrow 0}^{c}$}
            edge [pre] (ac1)
            edge [post] (c0);
    
            \node [transition] (tc1_1) [below of=a0, yshift=-1.9cm, xshift=1.5cm] {$t_{1 \rightarrow 1}^{c}$}
            edge [pre] (ac1)
            edge [post] (c1);
    
            \node [transition] (td00_0) [right of=tc1_1, xshift=-7mm] {$t_{00 \rightarrow 0}^{d}$}
            edge [pre] (ad00)
            edge [pre] (bd00)
            edge [post] (d0);
    
            \node [transition] (td00_1) [right of=td00_0, xshift=-6mm] {$t_{00 \rightarrow 1}^{d}$}
            edge [pre] (ad00)
            edge [pre] (bd00)
            edge [post] (d1);
    
            \node [place, draw=white] (helper) [right of=td00_1, xshift=-8mm]{$\dots$};
          \end{scope}
        \end{tikzpicture}
        \vspace*{-0.2cm}
      \caption{SAFC net corresponding to BN in Figure~\ref{fig:BN}}
      \label{fig:bn-safc-reduction}
    \end{figure}
  \end{minipage}

  \noindent 
  probability tables of $P^a,P^b$, the nodes without predecessors in
  the Bayesian network, where for instance the question of whether
  $P_0^a$ or $P_1^a$ are marked corresponds to the value of the random
  variable $X_a$ associated with node $a$.
  Since $X_a$ is a MAP variable, its two transitions are controllable.
  Note that enabling both transitions will never give a higher reward
  than enabling only one of them.
  (This is due to the fact that $\max\{x,y\} \ge p_1\cdot x +
  p_2\cdot y$ for $p_1,p_2 \ge 0$ with $p_1+p_2=1$.)

  The second level of transitions (each with rate $1$) is inserted
  only to obtain a free-choice net by creating sufficiently many
  copies of the places in order to make all conflicts free-choice.

  The third level of transitions simulates the probability tables of
  $P^c$, $P^d$, only to ensure the net being free-choice we need
  several copies.
  For instance, transition $t_{0 \rightarrow 0}^{c}$ consumes a token
  from place $P_0^{a,c}$, a place specifically created for the entry
  $P^c(c=0 \mid a=0)$ in the probability table of node~$c$.

  In the end the aim is to mark the places $P_0^c$ and $P_1^d$, and we
  can find a policy (deactivating either $t^a_{()\to 0}$ or
  $t^a_{()\to 0}$) such that the probability of reaching both places
  exceeds $p$ if and only if the \dmap\ instance specified above has a
  solution.

  This proof idea can be extended to more complex Bayesian networks,
  for a more formal proof see \short{\arxive}\full{the appendix}.
  \qed
\end{proofsketch}

\begin{proof}
  We show $\np^{\pp}$-hardness by reduction of the $\np^{\pp}$-hard
  \dmap\ problem for Bayesian networks (see also Section~\ref{sec:preliminaries}) to the \safcpol\ problem for (safe and
  acyclic free-choice) decision Petri nets.

  As such, we are given a Bayesian network
  $((X_1,\dots,X_n),\Delta,P),$ disjoint sets $E,F\subseteq X$ of
  evidence and MAP variables, evidence $e\in E$, and a threshold
  $p\in[0,1]$.
  To prove $\np^{\pp}$-hardness of \safcpol, we construct an instance
  of the \safcpol\ problem such that
  \[\max_{D\subseteq C}\val^D>p\qquad\Leftrightarrow\qquad\max_{f\in
      V_F}\mP(E=e\mid F=f)>p.\]
  The idea behind the construction is to simulate the computation in
  the Bayesian network by a SAFC net $N$.
  More intuition for this proof is given in the proof sketch in the
  main body of the paper.

  We define the set of places $P$ of $N$ to contain all places of the
  form
  \begin{itemize}
    \item
      $p^i_v$ for $i\in\setto{n},v\in \{0,1\}$ (places representing
      the fact that the random variable of node $i$ has value $v$), as
      well as
    \item
      $p^{(j,i)}_{\tilde{v}}$ for $i\in\setto{n}$,
      $\tilde{v}\in \{0,1\}^{\Delta^i}$, and $j\in\Delta^i$ if
      $\Delta^i\neq\emptyset$ and $j=\perp$ otherwise (auxiliary
      places, ensuring the fact that the net is free-choice).
  \end{itemize}

  Similarly, its transition set $T$ contains two types of transitions:
  \begin{itemize}
    \item
      transitions $t^i_{\tilde{v}\rightarrow v}$ for $i\in\setto{n}$,
      $\tilde{v}\in\prod_{j\in\Delta^i} \{0,1\}$, and $v\in \{0,1\}$
      with firing rates
      $\Lambda(t^i_{\tilde{v}\rightarrow v}) = P_i(v\mid\tilde{v})$
      (transitions simulating the probabilistic choices in the
      Bayesian network),
    \item
      transitions $t^i_{v}$ for $v\in \{0,1\}$ with firing rate $1$
      (auxiliary transitions).
  \end{itemize}

  All in all, this amounts to
  $2 + \max(1,|\Delta^i|)\cdot 2^{|\Delta^i|} \in\mathcal{O}(n\cdot
  2^{|\Delta^i|})$ places and
  $2+2\cdot 2^{|\Delta^i|}\in\mathcal{O}(2^{|\Delta^i|})$ transitions
  for each $i\in\setto{n}$.
  As the matrices $P_i$ of the Bayesian network contain
  $2^{|\Delta^i| +1}$ entries (hence the size of the input is already
  exponential in the $|\Delta^i|$), this SAFC-net is thus of
  polynomial size in the size of the network.

  The initial marking then puts one token in precisely each place
  $p^{(\perp,j)}_{()}$ and the flow relation is defined as follows
  (all values not given are $0$):
  \begin{align*}
    \pre{(t^i_v)}(p^i_v) & = & 1\\
    \post{(t^j_v)}(p^{(j,i)}_{\tilde{v}}) & = & 1\\
    \pre{(t^i_{\tilde{v}\rightarrow v_i})}(p^{(j,i)}_{\tilde{v}})
    & = & 1 \\
    \post{(t^i_{\tilde{v}\rightarrow v})}(p^i_{v}) & = & 1
  \end{align*}
  Note that the only direct (free-choice) conflicts exist between
  transitions of the form $t^i_{\tilde{v}\rightarrow v}$ for
  different values $v$ but equal $\tilde{v}$ and $i$.
  Intuitively, these conflicts simulate the probabilistic decision of
  choosing a value $v$ for the random variable $X_i$ in the Bayesian
  network when given the parent values $\tilde{v}$.
  In this sense, the places $p^i_{v}$ being marked in the execution of
  the net can be understood as the variable $X_i$ taking the value $v$
  in this simulation of the Bayesian network.
  The transitions $t^i_{v}$ then forward this signal to the child
  nodes in the network by duplicating the token in $p^i_{v}$ to the
  places $p^{(i,j)}_{\tilde{v}}$ which can be seen as the input from
  node $i$ to the node $j$ for the decision when $\tilde{v}$ are all
  parent values in node $j$.
  For our construction, this duplication step is necessary to ensure
  that the net remains free-choice (removing the duplication and
  directly feeding places $p^i_{v}$ into transitions
  $t^j_{\tilde{v}\rightarrow v'}$ would yield a smaller correct but
  non-free-choice net).

  Now, for the given evidence $e\in V_E$, we define the reward
  function as
  \[R(Q)=\begin{cases}
      1&\text{if }Q=\set{p^i_{e_i}\mid X_i\in E},\\
      0&\text{otherwise,}
    \end{cases}\]
  such that the reward is only received if all the places
  corresponding to $e$ are marked and the value thus represents the
  probability of reaching all these places.

  Finally, we encode the MAP-variables $F$ in the controllable
  transitions by defining
  \[C = \set{t^i_{()\rightarrow v}\mid X_i\in F,v\in \{0,1\}}\]
  where we use the fact that $F$ only contains input nodes.

  By construction, the resulting net is an SAFC net and can be
  constructed in polynomial time from the Bayesian network (as well as
  evidence and MAP variables).
  Furthermore, we have for all $f\in V_F$ that
  \[\val^{D_f}=\mP(E=e\mid F=f)\]
  for
  $D_f\coloneqq\set{t^i_{()\rightarrow g_i}\mid X_i\in F,g_i\in
    \{0,1\}\setminus\{f_i\}},$
  i.e., the set $D_f$ that deactivates all transitions that would mark
  a place corresponding to a value $g_i\in \{0,1\}\setminus\{f_i\}$
  for $X_i$ that differs from $f_i$, hence ensuring that all places
  corresponding to $f$ are marked (with probability $1$).

  As such, if
  \[\max_{f\in V_F}\mP(E=e\mid F=f)>p,\]
  we also have
  \[\max_{D\subseteq C}\val^D>p.\]

  On the other hand, we note that the maximal value $\val^D$ has to be
  reached for a set $D$ that deactivates all but one transition for
  all two-element sets of $\set{t^i_{()\rightarrow v}\mid
  v\in\{0,1\}}$ for $i\in\setto{n}$.

  To see this, note that clearly deactivating both transitions will
  not maximize the probability of reaching the goal places.
  Assume that $F = F' \uplus \{\bar{f}\}$ and for node $\bar{f}$ we
  activate both transitions.
  However, this cannot result in a higher reward, due to the fact that
  \footnote{Note that since the input nodes are uniformly
    distributed, the denominators are always non-zero.}
  \begin{align*}
    && \max_{b\in\{0,1\}} \mP(E=e \mid F'=f, \bar{f}=b) \\
    & = & \mP(E=e \mid F'=f, \bar{f}=0) \lor \mP(E=e \mid
    F'=f, \bar{f}=1) \\
    & = & \frac{\mP(E=e, \bar{f} = 0 \mid
      F'=f)}{\mP(\bar{f} = 0 \mid F'=f)} \lor
    \frac{\mP(E=e, \bar{f} = 1 \mid F'=f)}{\mP(\bar{f} =
      1 \mid F'=f)} \\
    & \ge & \mP(\bar{f} = 0 \mid F'=f)\cdot
    \frac{\mP(E=e, \bar{f} = 0 \mid F'=f)}{\mP(\bar{f} =
      0 \mid F'=f)} \mathop{+} \\
    && \qquad\qquad \mP(\bar{f} = 1 \mid F'=f) \cdot
    \frac{\mP(E=e, \bar{f} = 1 \mid F'=f)}{\mP(\bar{f} =
      1 \mid F'=f)} \\
    & = & \mP(E=e, \bar{f} = 0 \mid F'=f) + \mP(E=e,
    \bar{f} = 1 \mid F'=f) \\
    & = & \mP(E=e \mid F'=f)
  \end{align*}
  where the latter would be the reward that this policy gives us.
  The inequality above holds since
  $\max\{x,y\} \ge p_1\cdot x + p_2\cdot y$ for $p_1,p_2 \ge 0$ with
  $p_1+p_2=1$.

  Hence, defining $f_D\in D_F$ by $(f_D)_i\coloneqq v$ for the unique
  $v\in \{0,1\}$ with $t^i_{()\rightarrow v}\notin D$, we have that
  \[\val^{D}=\mP(E=e\mid F=f_D).\]

  Therefore, also if
  \[\max_{D\subseteq C}\val^D>p,\]
  we have
  \[\max_{f\in V_F}\mP(E=e\mid F=f)>p\]
  and vice versa.

  All in all, this shows that the \dmap\ problem can be reduced to the
  \safcpol\ problem in polynomial time with the same threshold
  $p\in[0,1]$.\qed
\end{proof}

In fact, a reduction in the opposite direction (from Petri nets to
Bayesian networks) is possible as well under mild restrictions, which
shows that these problems are closely related.

\begin{propositionrep}
  \label{prop:safc-bn-reduction}
  For two given constants $k,\ell$, consider the following problem:
  let $N$ be a SAFC decision Petri net, where for each branching cell
  the number of controllable transitions is bounded by some constant
  $k$.
  Furthermore, given its reward function $R$, we assume that
  $|\cup_{Q\in \supp(R)} Q| \le \ell$.
  Given a rational number $p$, does there exist a constant policy
  $\pi$ such that $\val^\pi > p$?

  This problem can be polynomially reduced to $\dmap$.
\end{propositionrep}

\vspace{-3.4cm}
\begin{adjustbox}{valign=C,raise=\strutheight,minipage={1\linewidth}}
  \begin{wrapfigure}[34]{r}{0.45\linewidth} 
  \vspace*{4.2cm}
  \centering
  \resizebox{\linewidth}{!}{%
    \begin{tikzpicture}
      \draw [->] (0,-2) -- (0,-7.83);
      \draw [->] (2.1,-8.8) -- (1.33,-9.2);

      \draw [->] (2.1,-5.3) -- (2.2,-5.63);
      \draw [->] (5,-4) -- (4.5,-5.63);
      \draw [->] (6.5,-4) -- (5,-5.63);

      \node[draw, rectangle, label={90:$X_{p_1}$}] (P1) {%
      \begin{tabular}{c c} 0 \; & 1 \\ \hline
        $0$ \; & $1$ \\
      \end{tabular}};

      \node[draw, rectangle, right=of P1, xshift=-3mm, label=$X_{t_1}$] (T1) {%
      \begin{tabular}{c c} 0 \; & 1 \\ \hline
        \multicolumn{2}{c}{$F$} \\
      \end{tabular}};

      \node[draw, rectangle, right=of T1, xshift=-2mm, label=$X_{p_2}$] (P2) {%
      \begin{tabular}{c c} 0 \; & 1 \\ \hline
        $0$ \; & $1$ \\
      \end{tabular}};

      \node[draw,rectangle,fill=white, right=of P2, yshift=-4.05cm, xshift=-5mm, label=$X_{t_5}$] (T5) {%
      \begin{tabular}{c c} 0 \; & 1 \\ \hline
        \multicolumn{2}{c}{$F$} \\
      \end{tabular}};

      \node[draw,rectangle,fill=white, right=of T5, xshift=-9mm, label=$X_{t_6}$] (T6) {%
      \begin{tabular}{c c} 0 \; & 1 \\ \hline
        \multicolumn{2}{c}{$F$} \\
      \end{tabular}};

      \node [draw,rectangle,fill=white, below=of P1, xshift=0.8cm, yshift=5mm, label={90:$X_{\mathbb{C}_1}$}] (C1) {%
      \begin{tabular}{c c | c c c}
        $X_{p_1}$ & $X_{t_1}$ & $\epsilon$ & $t_1$ & $t_2$ \\ \hline
        0 & $\ast$ & $1$ & $0$ & $0$ \\
        1 & 0 & $0$ & $0$ & $1$ \\
        1 & 1 & $0$ & $\nicefrac{1}{2}$ & $\nicefrac{1}{2}$
      \end{tabular}};

      \draw[->] (P1) --node[left]{} (C1);
      \draw[->] (T1) --node[left]{} (C1);

      \node [draw,rectangle,fill=white, below=of C1, xshift=1cm, yshift=5mm, label={170:$X_{p_3}$}] (P3) {%
      \begin{tabular}{c | c c}
        $X_{\mathbb{C}_1}$ & 0 & 1 \\ \hline
        $\epsilon$ & $1$ & $0$ \\
        $t_1$ & $1$ & $0$ \\
        $t_2$ & $0$ & $1$
      \end{tabular}};

      \draw[->] (C1) --node[left]{} (P3);

      \node [draw, rectangle, below=of P2, yshift=5mm, label={10:$X_{\mathbb{C}_2}$}] (C2) {%
      \begin{tabular}{c | c c c}
        $X_{p_2}$ & $\epsilon$ & $t_3$ & $t_4$ \\ \hline
        0 & $1$ & $0$ & $0$ \\
        1 & $0$ & $\nicefrac{1}{2}$ & $\nicefrac{1}{2}$
      \end{tabular}};

      \draw[->] (P2) --node[left]{} (C2);

      \node [draw, rectangle, below=of C2, yshift=1mm, label={80:$X_{p_4}$}] (P4) {%
      \begin{tabular}{c | c c}
        $X_{\mathbb{C}_2}$ & 0 & 1 \\ \hline
        $\epsilon$ & $1$ & $0$ \\
        $t_3$ & $0$ & $1$ \\
        $t_4$ & $1$ & $0$
      \end{tabular}};

      \draw[->] (C2) --node[left]{} (P4);

      \node [draw,rectangle,fill=white, below=of P4, yshift=0.8cm, label={20:$X_{\mathbb{C}_3}$}] (C3) {%
      \begin{tabular}{c c c c | c c c}
        $X_{p_3}$ & $X_{p_4}$ & $X_{t_5}$ & $X_{t_6}$ & $\epsilon$ & $t_5$ & $t_6$ \\ \hline
        0 & $\ast$ & $\ast$ & $\ast$ & $1$ & $0$ & $0$ \\
        $\ast$ & 0 & $\ast$ & $\ast$ & $1$ & $0$ & $0$ \\
        1 & 1 & 0 & 0 & $1$ & $0$ & $0$ \\
        1 & 1 & 0 & 1 & $0$ & $0$ & $1$ \\
        1 & 1 & 1 & 0 & $0$ & $1$ & $0$ \\
        1 & 1 & 1 & 1 & $0$ & $\nicefrac{1}{2}$ & $\nicefrac{1}{2}$
      \end{tabular}};

      \draw[->] (P4) --node[left]{} (C3);

      \node [draw, rectangle, left=of C3, xshift=0.9cm, yshift=-2.2cm, label={100:$X_{p_5}$}] (P5) {%
      \begin{tabular}{c c | c c}
        $X_{\mathbb{C}_1}$ & $X_{\mathbb{C}_3}$ & 0 & 1 \\ \hline
        $\epsilon$ & $\epsilon$ & $1$ & $0$ \\
        $\epsilon$ & $t_6$ & $1$ & $0$ \\
        $t_2$ & $\epsilon$ & $1$ & $0$ \\
        $t_2$ & $t_6$ & $1$ & $0$ \\
        $t_1$ & $\ast$ & $0$ & $1$ \\
        $\ast$ & $t_5$ & $0$ & $1$
      \end{tabular}};

      \node [draw, rectangle, below=of C3, yshift=8mm, label={180:$X_{p_6}$}] (P6) {%
      \begin{tabular}{c | c c}
        $X_{\mathbb{C}_3}$ & 0 & 1 \\ \hline
        $\epsilon$ & $1$ & $0$ \\
        $t_5$ & $0$ & $1$ \\
        $t_6$ & $0$ & $1$
      \end{tabular}};

      \draw[->] (C3) --node[left]{} (P6);

      \node [draw, rectangle, below=of C3, xshift=2cm, yshift=8mm, label={80:$X_{p_7}$}] (P7) {%
      \begin{tabular}{c | c c}
        $X_{\mathbb{C}_3}$ & 0 & 1 \\ \hline
        $\epsilon$ & $1$ & $0$ \\
        $t_5$ & $1$ & $0$ \\
        $t_6$ & $0$ & $1$
      \end{tabular}};

      \draw[->] (C3) --node[left]{} (P7);

      \node [draw, rectangle, below=of P6, yshift=0.8cm, label={180:$X_{\mathit{rew}}$}] (Rew) {%
      \begin{tabular}{c | c c c c c c c c}
        $X_{p_{5}}$ & 0 & 0 & 0 & 0 & 1 & 1 & 1 & 1 \\
        $X_{p_6}$ & 0 & 0 & 1 & 1 & 0 & 0 & 1 & 1 \\
        $X_{p_7}$ & 0 & 1 & 0 & 1 & 0 & 1 & 0 & 1 \\ \hline
        0 & $\nicefrac{1}{2}$ & $\nicefrac{1}{2}$ & $\nicefrac{1}{3}$ & $\nicefrac{1}{2}$ &
            $\nicefrac{1}{3}$ & $\nicefrac{1}{2}$ & $\nicefrac{1}{3}$ & $\nicefrac{1}{2}$ \\
        1 & $\nicefrac{1}{2}$ & $\nicefrac{1}{2}$ & $\nicefrac{2}{3}$ & $\nicefrac{1}{2}$ &
            $\nicefrac{2}{3}$ & $\nicefrac{1}{2}$ & $\nicefrac{2}{3}$ & $\nicefrac{1}{2}$
      \end{tabular}};

      \draw[->] (P5) --node[left]{} (Rew);
      \draw[->] (P6) --node[left]{} (Rew);
      \draw[->] (P7) --node[left]{} (Rew);

    \end{tikzpicture}}
    \vspace*{-0.8cm}
    \caption{Bayesian network obtained
      from the SAFC net in Figure~\ref{fig:SAFC} below. Entries $\ast$ are
      \enquote*{don't-care} values.} \label{fig:safc-bn-reduction}
  \end{wrapfigure}
  \strut{}
  \vspace{2.6cm}
  \begin{proofsketch}
    We sketch the reduction, which is inspired by~\cite{BruniMM20},
    via an example: consider the SAFC net in Figure~\ref{fig:SAFC},
    where the problem is to find a deactivation pattern such that the
    payoff exceeds $p$.
    We encode the net into a Bayesian network
    (Figure~\ref{fig:safc-bn-reduction}), resulting in an instance of
    the \dmap\ problem.

    We have four types of random variables: place variables ($X_p$,
    $p\in P$), which record which place is marked;
    transition variables ($X_{t_1}, X_{t_5}, X_{t_6}$), one for each
    controllable transition, which are the MAP variables;
    cell variables ($X_{\mathbb{C}_i}$ for
    $\mathbb{C}_1 = \{t_1,t_2\}$, $\mathbb{C}_2 = \{t_3,t_4\}$,
    $\mathbb{C}_3 = \{t_5,t_6\}$) which are non-binary and which
    record which transition in the cell was fired or whether no
    transition was fired ($\epsilon$); a reward variable
    ($X_{\mathit{rew}}$) such that $\mP(X_{\mathit{rew}} = 1)$ equals
    the function $\psi$ applied to the payoff.
    Note that we use the affine function $\psi$ from the proof of
    Proposition~\ref{prop:pbpn-np-pp} to represent rewards as
    probabilities in the interval $[0,1]$.
    The threshold for the $\dmap$ instance is $\psi(p)$.

    Dependencies are based on the structure of the given SAFC net.
    For instance, $X_{\mathbb{C}_3}$ is dependent on $X_{p_3}$,
    $X_{p_4}$ (since $\pre{\mathbb{C}_3} = \{p_3,p_4\}$) and
    $X_{t_5}$, $X_{t_6}$ (since $t_5,t_6$ are the controllable
    transitions in $\mathbb{C}_3$).
    
    Both the matrices of cell and place variables could become
    exponentially large, however this problem can be resolved easily
    by dividing the matrices into smaller ones and cascading
    them.
    Since the number of controllable transitions is bounded by $k$ and
    the number of rewarded places by $\ell$, they will not cause an
    exponential blowup of the corresponding matrix.  \qed
\end{proofsketch}
\end{adjustbox}

\begin{proof}
  Given a net $N = (P,T,\pre{()},\post{()},\Lambda,m_0,C,R)$
  satisfying the restrictions and a threshold $p$ we construct a
  \dmap\ problem as follows:

  First, we define a Bayesian network $(X, \Delta, P)$ with a set of
  random variables of the form:
  \begin{itemize}
    \item
      $X_p$ for $p \in P$ (variables representing the presence of a
      token in each place)
    \item
      $X_t$ for $t \in C$ (variables representing whether a
      controllable transition is activated)
    \item
      $X_{\mathbb{C}}$ for every branching cell $\mathbb{C}$ (cf.
      Section~\ref{sec:preliminaries}) and finally
    \item
      $X_{\mathit{rew}}$ as the only evidence variable in $E$
  \end{itemize}
  The subscripts ($p,t,\mathbb{C},\mathit{rew}$) correspond to the
  nodes of the Bayesian network.

  Second, we clarify which variables/nodes are dependent on one
  another:
  \begin{itemize}
    \item
      $\Delta^p=\{\mathbb{C}\in\bc{N}\mid p\in\post{\mathbb{C}}\}$
    \item
      $\Delta^t = \emptyset$
    \item
      $\Delta^{\mathbb{C}} = \pre{\mathbb{C}} \cup (C\cap \mathbb{C})$
    \item
      $\Delta^{\mathit{rew}} = \cup_{Q \in \supp(R)} Q$
  \end{itemize}

  \medskip\hrule\medskip

  To complete the description of the Bayesian network, we now specify
  the probability matrices.
  \begin{itemize}
    \item
      For nodes representing controllable transitions ($X_t$,
      $t\in C$) we have no predecessor variables, hence they are all
      input nodes.
      These are the MAP variables $F$ and will later be set to a
      specific boolean value according to the chosen policy $\pi$,
      when solving the \dmap\ problem.
      As required by the considered variant of the \dmap\ problem, we
      assume that they are uniformly distributed.
    \item
      For random variables representing places ($X_p$), whenever
      $\Delta_p=\emptyset$, we set $P_p(1) = 1$ if $p\in m_0$ and $0$
      otherwise.
      If $p$ is in the post-set of a transition let
      $\Delta^p = \{\mathbb{C}_1,\dots,\mathbb{C}_n \}$.
      Keep in mind that cell variables as non-binary variables return
      a transition or $\epsilon$.
      We define, for $t_j \in \mathbb{C}_j\cup \{\epsilon\}$:
      \begin{align*}
        P_p (1 \mid t_1\dots t_n)
        &= \bigvee_{j \in \{1 \dots n\}} [p\in \post{t_j}]
      \end{align*}
      The binary operator $[p \in \post{t}]$ returns $1$ if place $p$
      is in the post set of transition $t$ and $0$ otherwise.
      If $t = \epsilon$, the value is also $0$.

      Furthermore
      $P_p(0\mid t_1\dots t_n) = 1 - P_p(1\mid t_1\dots t_n)$.
    \item
      For a cell variable $X_{\mathbb{C}}$, let
      $\Delta^{\mathbb{C}} = \{p_1,\dots,p_m,t_1,\dots,t_k\}$,
      $v_i,u_j\in \{0,1\}$ where $i\in\{1,\dots,m\}$,
      $j\in\{1,\dots,k\}$.
      That is $v_i$ tells us if place $p_i$ is marked and $u_j$
      specifies if transition $t_j\in C$ is activated.
      Let
      \[\emph{Act}(\mathbb{C},u)=\{t\in\mathbb{C}\mid t\not\in C
        \lor (t\in C \land \exists j (t=t_j \land u_j=1)) \} \]
      be the set of transitions that are activated in $\mathbb{C}$
      (since they are either not controllable or controllable and
      activated).
      Now for every $t \in \mathbb{C}$ we have:
      \begin{align*}
        P_{\mathbb{C}} (t \mid v_1\dots v_m u_1\dots u_k) &=
        \frac{\Lambda(t)}{\sum_{t'\in\mathit{Act}(\mathbb{C},u)}
          \Lambda(t')}
      \end{align*}
      if $v_1\dots v_m = 1\dots 1$ and
      $t\in \mathit{Act}(\mathbb{C},u)$.
      The value is $0$ otherwise.
      Instead:
      \[ P_{\mathbb{C}}(\epsilon \mid v_1\dots v_m u_1\dots u_k) = 1 -
        \sum_{t \in \mathbb{C}} P_{\mathbb{C}} (t \mid v_1\dots v_m
        u_1\dots u_k), \] in particular the value is $1$ if
      $v_1\dots v_m \neq 1\dots 1$.
    \item For the reward node, we make use of the affine linear
      transformation $\psi$ introduced in the proof of
      Proposition~\ref{prop:pbpn-np-pp}, using the lower and upper
      bounds $\valfunc_{\min}$, $\valfunc_{\max}$ in order to
      represent the rewards as probabilities (mapping to $[0,1]$).  As
      already mentioned above, we also have to adapt the threshold $p$
      to $\tilde{p} \coloneqq \psi(p)$.  Let
      $\Delta^{\mathit{rew}} = \{p_1,\dots,p_m\}$ and $v_i\in\{0,1\}$,
      $i\in\{1,\dots,m\}$ binary values indicating whether $p_i$ will
      be marked.  Furthermore let $P_v = \{ p_i \mid v_i=1\}$ the
      corresponding set of marked places.  Then
      \begin{align*}
        P_{\mathit{rew}} (1 \mid v_1\dots v_m)
        &= \psi(\sum_{\substack{Q\in\supp(R)\\Q\subseteq P_v}}R(Q))\\
        P_{\mathit{rew}}(0\mid v_1\dots v_m)
        &=1-P_{\mathit{rew}}(1 \mid v_1\dots v_m)
      \end{align*}
  \end{itemize}
  In order to completely define the \dmap\ instance, we fix the
  evidence variables to $E = \{X_{\mathit{rew}}\}$ with
  $e = 1 \in V_e = \{0,1\}$, the MAP variables to
  $F = \{X_t \mid t \in C\}$ and the threshold to $\tilde{p}$.

  \medskip\hrule\medskip

  This $\dmap$ instance has a solution if there is a deactivation
  pattern $f\in V_F$ such that
  $\mP(X_{\mathit{rew}} = 1 \mid F = f) > \tilde{p}$.
  Assuming that $\mathit{UC}$ is the set of all functions
  \footnote{These functions choose which transition is fired in each
  cell. We have to sum over all these functions to determine the
  probability.}
  $u\colon \bc{N}\to T\cup\{\epsilon\}$ such that
  $u(\mathbb{C}) = \epsilon\lor u(\mathbb{C})\in\mathbb{C}$, we ask --
  by evaluating the Bayesian network -- whether there exists
  $f\colon C\to\{0,1\}$ such that
  \begin{align*}
    \tilde{p} < & \sum_{u\in\mathit{UC}} \sum_{v\colon P\to \{0,1\}}
    \prod_{p\in P} P_p(v(p) \mid (u(\mathbb{C}))_{\mathbb{C}\in
      \Delta^p})
    \mathop{\cdot} \\
    & \qquad \prod_{\mathbb{C}\in\mathit{BC}}
    P_{\mathbb{C}} (u(\mathbb{C}) \mid
    (v(p))_{p\in\Delta^{\mathbb{C}} \cap P}, (f(t))_{t\in
      \Delta^{\mathbb{C}}\cap C})\cdot P_{\mathit{rew}}(1\mid
    (v(p))_{p\in\Delta^{\mathit{rew}}})
  \end{align*}
  We observe that for a given $u,v$, the product equals $0$, unless
  $v$ satisfies: $v(p)=1$ iff $p\in m_0$ or there exists $t\in T$ such
  that ($p\in\post{t}\land u(\mathbb{C}) = t$), i.e., $p$ is either
  initial or is generated by a transition that was fired.
  We denote this specific $v$ by $v[u]$ and the term above becomes:
  \begin{align*}
    && \sum_{u\in\mathit{UC}}
    \prod_{\mathbb{C}\in\mathit{BC}}P_{\mathbb{C}}(u(\mathbb{C}) \mid
    (v[u](p))_{p\in\Delta^{\mathbb{C}} \cap P}, (f(t))_{t\in
      \Delta^{\mathbb{C}}\cap C})\cdot
    \psi\big(\sum_{\substack{Q\in\supp(R)\\ Q\subseteq
        P_{v[u]}}} R(Q)\big) \\
    & = & \psi\big(\sum_{u\in\mathit{UC}}
    \prod_{\mathbb{C}\in\mathit{BC}}P_{\mathbb{C}}(u(\mathbb{C}) \mid
    (v[u](p))_{p\in\Delta^{\mathbb{C}} \cap P}, (f(t))_{t\in
      \Delta^{\mathbb{C}}\cap C})\cdot \sum_{\substack{Q\in\supp(R)\\
        Q\subseteq P_{v[u]}}} R(Q)\big)
  \end{align*}
  This equality is true since $\psi$ commutes with expected values
  (cf. proof of Proposition~\ref{prop:pbpn-np-pp}).
  Note that $P_{v[u]}=m_0\cup\post{(u[\bc{N}]\setminus\{\epsilon\})}$.

  Due to the fact that $\psi$ is strictly monotone, this value in turn
  is larger than or equal to $\tilde{p} = \psi(p)$ iff
  \begin{align*}
    p & < & \sum_{u\in\mathit{UC}}
    \prod_{\mathbb{C}\in\mathit{BC}} P_{\mathbb{C}} (u(\mathbb{C}) \mid
    (v[u](p))_{p\in\Delta^{\mathbb{C}} \cap P}, (f(t))_{t\in
      \Delta^{\mathbb{C}}\cap C})\cdot \sum_{\substack{Q\in\supp(R)\\
        Q\subseteq P_{v[u]}}} R(Q)
  \end{align*}
  Now we observe that any maximal configuration
  $\tau\in\validconfs{N_D}$ (where $D = f^{-1}(\{0\})$) can be
  represented by a function $u\colon \bc{N}\to T\cup\{\epsilon\}$
  defined as $u(\mathbb{C}) = t$ if $\mathbb{C}\cap \tau = \{t\}$ and
  $\epsilon$ otherwise.
  This function $u$ clearly satisfies
  $u(\mathbb{C}) = \epsilon\lor u(\mathbb{C}) \in \mathbb{C}$.

  Vice versa, given such a function $u$ it only corresponds to a
  configuration $\tau = u[\bc{N}]\backslash\{\epsilon\}$ if
  the places in the initial marking and those generated by transitions
  in $\tau$ can cover every $\pre{\tau}$, i.e., every transition in
  $\tau$ is enabled.
  In other words: $\pre{t}\subseteq P_{v[u]}$ for all $t\in \tau$.
  Assume that $t\in\mathbb{C}$.
  If the inclusion $\pre{\mathbb{C}} = \pre{t}\subseteq P_{v[u]}$ does
  not hold, by definition:
  \[ P_\mathbb{C}(u(\mathbb{C}) \mid
    (v[u](p))_{p\in\Delta^{\mathbb{C}} \cap P}, (f(t))_{t\in
      \Delta^{\mathbb{C}}\cap C}) = 0, \]
  which means that such summands disappear.

  Furthermore, if $u$ does correspond to a configuration $\tau$, we
  have that
  \[ \mP(\transitions=\tau) = \prod_{\mathbb{C}\in\mathit{BC}}
    P_{\mathbb{C}}
    (u(\mathbb{C}) \mid (v[u](p))_{p\in\Delta^{\mathbb{C}} \cap P},
    (f(t))_{t\in \Delta^{\mathbb{C}}\cap C}), \]
  that is, the probability of a configuration is obtained by
  multiplying the probability that its transitions \enquote*{win}
  against the other transitions in their cells, taking deactivated
  transitions into account.

  Summarizing, this means that we check the inequality:
  \begin{align*}
    p < \sum_{\tau\in\validconfs{N_D}}
    \mP(\transitions=\tau) \cdot
    \sum_{\substack{Q\in\supp(R)\\ Q\subseteq m_0\cup\post{\tau}}}
    R(Q) = \mathbb{E}[V\circ \mathit{pl}],
  \end{align*}
  that is, we add up the rewards for each configuration, weighted by
  its probability, which is exactly the answer to the $\safcpol$
  problem.

  \medskip\hrule\medskip

  We give some additional intuition for this construction:

  In the reduction above, it is apparent that there are only two types
  of variables that have matrix entries unequal to $0$ or $1$:
  variables representing cells and the reward variable.
  Cell variables are responsible for choosing and returning the
  transition firing in that specific cell according to the enabled
  transitions and their respective firing rates.
  All other variables (apart from the aforementioned reward variable)
  simply forward these information by adequately setting which places
  are marked or empty.

  Because we work with acyclic Petri nets, there will be a final
  marking, in which no further transitions can fire.
  This implies that we will reach a point in time, where all places
  involved in a reward function have either been marked at least once
  or will never be marked.
  We can take note of this information by introducing a final reward
  marking consisting of bits for each of these places representing
  whether it was ever marked or always empty.
  By choosing a policy $\pi$, the transition probabilities in the cell
  variables are manipulated in order to fit the firing rates and
  therefore also how likely it is to reach each possible final reward
  marking.

  Finally, given a policy $\pi$ we obtain how probable each final
  reward marking is and we simply have to multiply this with the
  respective reward, which is already coded into the reward variable
  (albeit fit to the $[0,1]$ interval) and sum up these products.
  This is achieved through matrix multiplication in the BN and results
  in the expected reward for policy $\pi$.

  Hence, if the policy problem for the SAFC net has a solution for
  bound $p$, the \dmap\ problem also has a solution for bound
  $\tilde{p}$.

  \medskip\hrule\medskip

  Finally, while the size of the graph underlying the Bayesian network
  is linear in the size of the Petri net, note that the size of the
  Bayesian network itself, i.e., the sum of the size of its matrizes,
  could still be exponential.
  In particular, this occurs for random variables of type
  $X_\mathbb{C}$ or $X_p$, for which the number of parents is
  unbounded.
  Both corresponding types of nodes can easily be split up by
  cascading multiple variables with only two input variables, where
  the sum of the size of the matrices is only linear, giving us a
  polynomial reduction.

  \begin{center}
    \includegraphics[width=0.9\textwidth]{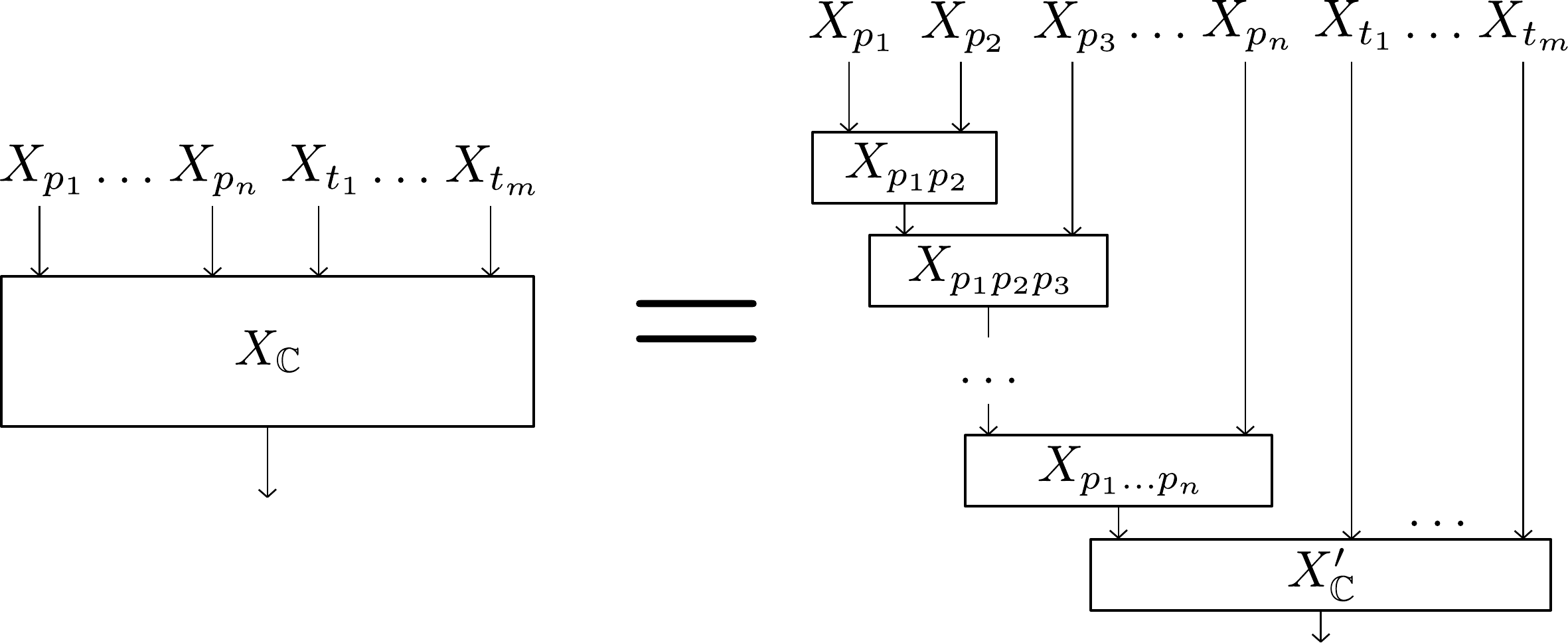}
  \end{center}

  For the splitting of cell variables $X_{\mathbb{C}}$ we remember
  intermediate results of whether the places seen so far are all
  marked, basically by implementing a binary $\land$-operator.
  The matrix corresponding to the random variable $X_{p_1\dots p_j}$
  is denoted by $P_{p_1\dots p_j}$ and we denote the matrix
  corresponding to $X'_{\mathbb{C}}$ by $P'_{\mathbb{C}}$.
  We define:
  \begin{align*}
    P_{p_1 \dots p_i} (1 \mid y_1 y_2) &= y_1 \wedge y_2 \qquad
    y_i = \{0,1\} \\
    P'_{\mathbb{C}} (t \mid v \, u_1 \dots u_m) &=
    \begin{cases}
      P_{\mathbb{C}}(t\mid\overbrace{1\dots 1}^{n}\,u_1\dots u_m)
      &\text{ , if } v = 1\\
      0 &\text{ , otherwise}
    \end{cases}\\
    P'_{\mathbb{C}} (\epsilon \mid v \, u_1 \dots u_m) &=
    1-\sum_{t \in \mathbb{C}} P'_{\mathbb{C}}(t\mid v\,u_1\dots u_m)
  \end{align*}
  The last node is given the information whether all places are
  marked, checks which controllable transitions are activated and
  returns the entries of the original matrix $P_{\mathbb{C}}$.

  We now argue why this construction is correct: we define the
  probability function specified by the new network (on the right-hand
  side) by $\bar{\mP}$ and the one by the original network (on the
  left-hand side) by $\mP$.
  Then we have, given $t\in\mathbb{C}\cup\{\epsilon\}$:
  \begin{align*}
    &\bar{\mP}(X'_{\mathbb{C}} = t \mid X_{p_1} = y_1,\dots,
    X_{p_n}=y_n, X_{t_1} = u_1,\dots, X_{t_m} = u_m) \\
    = & \sum_{w\colon\{2,\dots,n\}\to\{0,1\}}
    P'_{\mathbb{C}}(t\mid w(n)\, u_1\dots u_m)\cdot \prod_{j=3}^{n}
    P_{p_1\dots p_j} (w(j)\mid w(j-1)\,y_j) \mathop{\cdot} \\
    & \qquad\qquad P_{p_1p_2}(w(2) \mid y_1y_2) \\
  \end{align*}
  Here $w$ is a function that assigns (boolean) values to the
  intermediate wires.
  The product under the sum is only non-zero if $w(n) = 1$, due to the
  definition of $P'_{\mathbb{C}}$, and -- by induction -- $w(j) = 1$
  for all other indices $j$, otherwise the matrix entry of
  $P_{p_1\dots p_j}$ equals $0$.
  Hence, the above sum simplifies to
  \begin{align*}
    &
    P'_{\mathbb{C}}(t\mid 1\, u_1\dots u_m)\cdot \prod_{j=3}^{n}
    P_{p_1\dots p_j} (1\mid 1\,y_j) \cdot P_{p_1p_2}(1 \mid y_1y_2) \\
    = &
    \begin{cases}
      P_{\mathbb{C}}(t\mid 1\dots 1u_1\dots u_m) & \mbox{if
      }y_1=\dots=y_n=1 \\
      0 & \mbox{otherwise}
    \end{cases} \\
     = & P_{\mathbb{C}}(t\mid y_1\dots y_n u_1\dots u_m) \\
     = & \mP(X'_{\mathbb{C}} = t \mid X_{p_1} = y_1,\dots,
    X_{p_n}=y_n, X_{t_1} = u_1,\dots, X_{t_m} = u_m)
  \end{align*}

  Similarly, when splitting a place variable $X_p$, we save
  intermediate results on whether the transitions chosen and fired by
  the parent cells mark $p$.
  Therefore, here we produce a cascade with the binary
  $\lor$-operator.
  \begin{center}
    \includegraphics[width=0.6\textwidth]{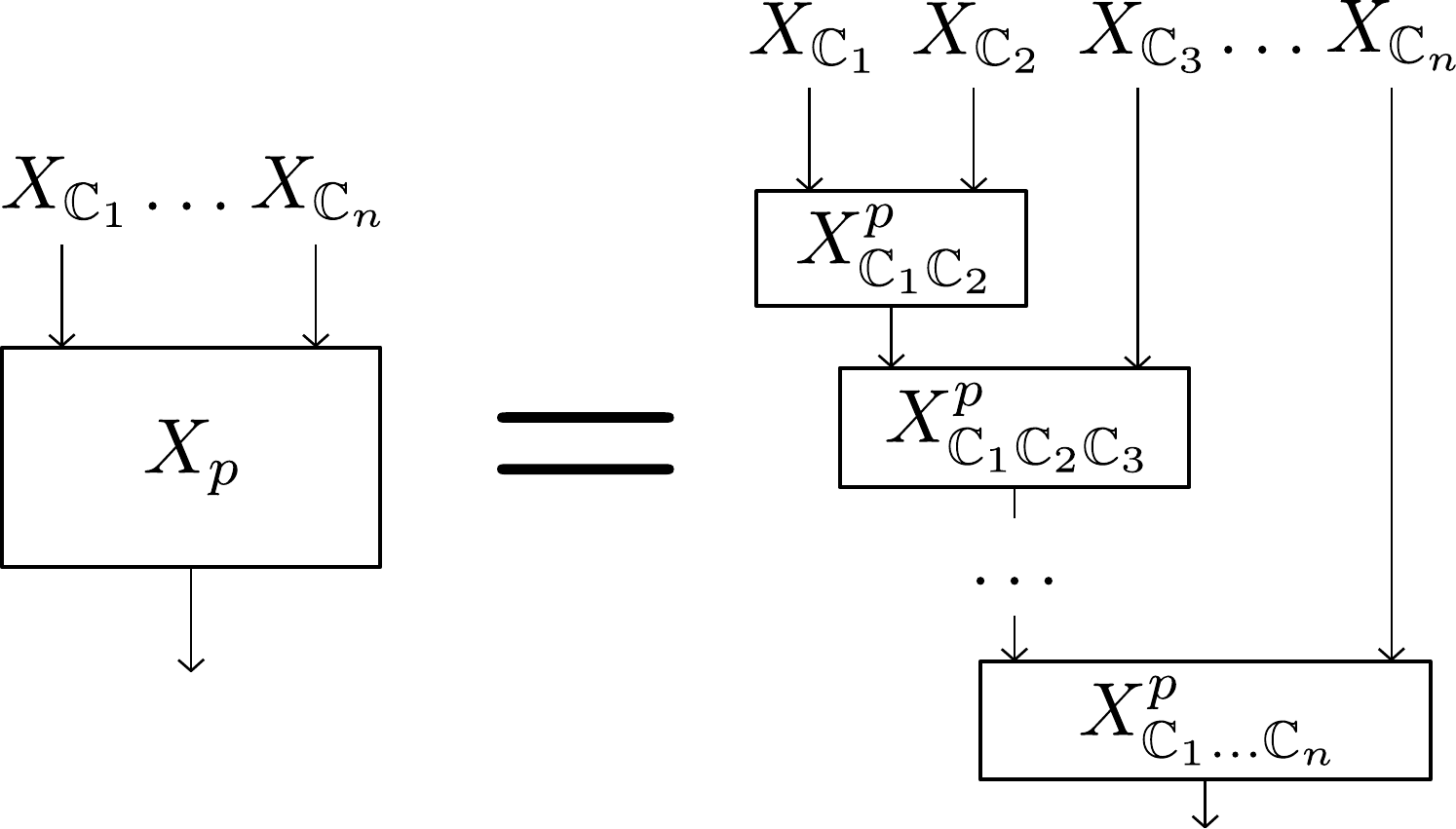}
  \end{center}
  The probability matrix associated with
  $X^p_{\mathbb{C}_1 \dots \mathbb{C}_l}$ is denoted by
  $P^p_{\mathbb{C}_1 \dots \mathbb{C}_l}$ and we define:
  \begin{align*}
    P_{\mathbb{C}_1 \mathbb{C}_2}^p (1 \mid t_1 t_2)
    &= [p \in \post{t_1}] \vee [p \in \post{t_2}]\\
    P_{\mathbb{C}_1 \dots \mathbb{C}_l}^p (1 \mid y \, t)
    &= y \vee [p \in \post{t}]
  \end{align*}
  The correctness argument is analogous to the one above.

  The fact that there are always at most $k$ controllable transitions
  for each cell $\mathbb{C}$ and the reward refers to at most $\ell$
  places, all possible combinations can be encoded into the
  probability matrices for $X_{\mathbb{C}}$ and $X_{\mathit{rew}}$,
  which does not cause an exponential blowup.

  \qed
\end{proof}

\medskip

\begin{corollary}\label{cor:safc-pp-hard}
  The problem \safcval\ is $\pp$-hard and, therefore, also
  $\pp$-complete.
\end{corollary}
\begin{proof}
  We note that using the construction in the proof of
  Proposition~\ref{prop:safc-np-pp-hard} with the set $F$ of MAP
  variables being empty, we can reduce the \dpr\ problem for Bayesian
  networks to the \safcval\ problem, showing that \safcval\ is
  $\pp$-hard.
  Using Corollary~\ref{cor:safc-pp}, this yields that \safcval\ is
  $\pp$-complete. \qed
\end{proof}

\begin{corollary}
    For any polynomial $\phi:\mN_0\rightarrow\mN_0$ fulfilling
    $\phi(n)\geq n$ for all $n\in\mN_0$, the problem \valn{[\phi]BPN}
    is $\pp$-complete and \poln{[\phi]BPN} is $\np^{\pp}$-complete.
\end{corollary}

\begin{proof}
    As any safe and acyclic free-choice net is an $id$-bounded net, it
    is, in particular, a $\phi$-bounded net with $\phi$ as above, and
    we have \safcval\ $\leq_p$ \valn{[\phi]BPN} and \safcpol\ $\leq_p$
    \poln{[\phi]BPN}.
    Propositions~\ref{prop:pbpn-np-pp} and~\ref{prop:safc-np-pp-hard}
    as well as Corollary~\ref{cor:safc-pp-hard}, therefore show that
    \valn{[\phi]BPN} is $\pp$-complete and \poln{[\phi]BPN} is
    $\np^{\pp}$-complete. \qed
\end{proof}

\subsection{Complexity of free-choice occurrence decision nets}
\label{sec:fcon-complexity}

Now we further restrict SAFC nets to occurrence nets, which leads to a
substantial simplification.
The main reason for this is the absence of backwards-conflicts, which
means that each place is uniquely generated, making it easier to trace
causality, i.e., there is a unique minimal configuration that
generates each place.

\begin{propositionrep}
  \label{prop:fconval-p-np}
  The problem \fconval\ is in $\p$.
  In particular, \fconpol\ is in $\np$.
\end{propositionrep}

\begin{proofsketch}
  Determining the probability of reaching a set of places $Q$ in an
  occurrence net amounts to multiplying the probabilities of the
  transitions on which the places in $Q$ are causally dependent.
  This can be done for every set $Q$ in the support of the reward
  function $R$, which enables us to determine the expected value in
  polynomial time, implying that \fconval\ lies in $\p$.
  By guessing a policy for an occurrence net with controllable
  transitions, we obtain that \fconpol\ lies in $\np$. \qed
\end{proofsketch}

\begin{proof}
  To show that \fconval\ is in P, we explain how to compute the value
  of a free-choice occurrence net without any controllable decisions.
  First we notice that
  \begin{align*}
    \val &= \mathbb{E}\left[V\circ \places\right]
    = \sum_{\mu\in\firingsequences{N}} \mP(\mu)\cdot V(\places(\mu))
    = \sum_{M\subseteq P} \mP(\places=M)\cdot V(M) \\
    &= \sum_{M\subseteq P}\mP(\places=M)\cdot \sum_{Q\subseteq M} R(Q)
    = \sum_{M\subseteq P}\sum_{Q\subseteq M}\big( \mP(\places=M)\cdot
    R(Q) \big) \\
    &= \sum_{Q\subseteq M\subseteq P} \big( \mP(\places=M)\cdot
    R(Q) \big)
    = \sum_{Q\subseteq P} \bigg( \sum_{Q\subseteq M\subseteq P}
    \mP(\places=M) \bigg)\cdot R(Q) \\
    &= \sum_{Q\subseteq P} \mP(\places\supseteq Q) \cdot R(Q)
    = \sum_{Q\in\supp(R)} \mP(\places\supseteq Q) \cdot R(Q)
  \end{align*}
  In particular, since $R$ has polynomial support, it suffices to show
  that we can compute $\mP(\places \supseteq Q)$ in polynomial time
  for any $Q\subseteq P$ which is exactly the probability of reaching
  at least all places $p\in Q$ (not necessarily simultaneously).

  Now, as we are dealing with an occurrence net, we have that reaching
  $p\in P$ is equivalent to firing all transitions on which $p$ is
  causally dependent.
  So let $T' = \{t\in T\mid \exists q\in Q\colon t \prec_N q\}$ be the
  set of transitions that are causes of places in $Q$.
  If two transitions $t,t'\in T'$ are now in conflict (which can be
  checked in polynomial time), the probability of reaching $Q$ is
  zero.
  Otherwise, due to the net being free-choice, we can multiply the
  local firing probabilities of all transitions in $T'$ to obtain the
  probability of reaching $Q$ in polynomial time.

  All in all, this procedure can be used to calculate $\val$ whence
  \fconval\ is in~$\p$.

  That \fconpol\ lies in $\np$ follows from the fact that, given an
  occurrence net with controllable transitions, we can guess a policy
  in polynomial time and then solve the resulting \fconval\ instance
  again in polynomial time. \qed
\end{proof}

\begin{propositionrep}\label{prop:fconpol-np-hard}
  The problem \fconpol\ is $\np$-hard and, therefore, also
  $\np$-complete.
\end{propositionrep}

\begin{proofsketch}
  To show $\np$-hardness we reduce $\threesat$ (the problem of
  deciding the satisfiability of a propositional formula in
  conjunctive normal form with at most three literals per clause) to
  \fconpol.
  Given a formula $\psi$, this is done by constructing a simple
  occurrence net with parallel controllable transitions, one for each
  atomic proposition $\ell$ in $\psi$.
  Then we define a reward function with polynomial support in such a
  way that the expected reward for the constructed net is larger or
  equal than the number of clauses iff the formula has a model.
  The correspondence between the model and the policy is such that
  transitions whose atomic propositions are evaluated as true are
  deactivated. \qed
\end{proofsketch}

\begin{proof}
  We show $\np$-hardness by a polynomial reduction from \threesat\ to
  \fconpol.

  It is well-known that \sat, the problem of deciding whether a given
  propositional formula $\psi$ is satisfiable, is $\np$-complete
  ~\cite{Papadimitriou94}.
  Its variant \threesat\ is still $\np$-complete, where the
  propositional formula $\psi$ is in conjunctive normal form with
  exactly three literals per clause, i.e.,
  $\psi = \bigwedge_{i\in I}(\ell^{i}_1\lor \ell^{i}_2\lor
  \ell^{i}_3)$, where
  $\ell^{i}_j\in\set{x,\neg x\mid x\in \mathcal{X}}$ for a set of
  atomic propositions $\mathcal{X}=\set{x_1,\dots,x_n}$.

  Assume that we are given an \threesat\ instance, i.e., a
  propositional formula
  $\psi = \bigwedge_{i\in I}(\ell_1^i\lor \ell2^i\lor \ell3^i)$ where
  $\ell_j^i\in \{x,\lnot x\mid x\in \mathcal{X}\}$ for a set of
  propositional formulas $\mathcal{X}$.
  Based on $\psi$ we construct a free-choice occurrence SDPN $N$ as
  follows:
  \begin{itemize}
  \item
    $P= \set{p_x,q_x\mid x\in\mathcal{X}}, T=\set{t_x\mid
      x\in\mathcal{X}}$, where $\pre{t_x}(p_x) = 1$ and
    $\pre{t_x}(p) = 0$ if $p\neq p_x$.
    Similarly $\post{t_x}(q_x) = 1$ and $0$ for all other places.
    Furthermore $\Lambda\equiv 1$ (all rates are equal to $1$) and
    $m_0(p_x) = 1$, $m_0(q_x) = 0$ for all $x\in \mathcal{X}$.

    In other words, the net consists of $n=|\mathcal{X}|$ separate
    subnets, each with a single transition $t_x$ that transfers a
    token from an input place $p_x$ into an output place $q_x$.
    This construction can be performed in polynomial time in $n$ and
    obviously results in a free-choice occurrence net.
  \item
    $C=T$, i.e., all transitions are controllable.
  \end{itemize}

  Now, only the reward function, which is central to this result,
  remains to be constructed.
  For this, we note that a place $q_x$ for an atomic proposition
  $x\in\mathcal{X}$ is reached if and only if the transition $t_x$ is
  not deactivated.
  We use this observation to encode the propositional formula $\psi$
  given above into a reward function as a formula on deactivated
  transitions $t_x$.

  The reward function is constructed as follows: For each positive
  literal $\ell_j^i=x\in\mathcal{X}$, we define a reward function as
  \[R_x:\mathcal{P}(P)\rightarrow\mR,Q\mapsto\begin{cases}
      1&\text{if }Q=\emptyset,\\
      -1&\text{if }Q=\set{q_x},\\
      0&\text{otherwise,}
    \end{cases}\]
  and for negative literals $\ell_j^i=\neg x$ for some
  $x\in\mathcal{X}$ as
  \[R_{\neg x}:\mathcal{P}(P)\rightarrow\mR,Q\mapsto\begin{cases}
      1&\text{if }Q=\set{q_x},\\
      0&\text{otherwise.}
    \end{cases}\]
  Intuitively, this interprets $x$ being true as not reaching $q_x$
  which in turn, due to the construction of the underlying net, is
  equivalent to $t_x$ being deactivated.
  This gives us for sets $D\subseteq C=T$ of deactivated transitions
  and a literal $\ell_j^i\in\set{x,\neg x\mid x\in\mathcal{X}}$ the
  corresponding value function
  \[
    \valfunc_{\ell_j^i}(M)=\sum_{Q\subseteq M}R_x(Q)=
    \begin{cases}
      1&\text{if $\ell_j^i=x$ and $q_x\notin M$ or
        $\ell_j^i=\neg x$ and $q_x\in M$,}\\
      0&\text{otherwise,}
    \end{cases}
  \]
  and, thus, using the interpretation that $x$ is true iff $t_x\in D$,
  we obtain the expected reward (expectation of the random variable
  $V_{\ell_j^i}^D$ for the constant policy $D$)
  \[\val_{\ell_j^i}^D=\begin{cases}
      1&\text{if $\ell_j^i$ is true,}\\
      0&\text{if $\ell_j^i$ is false.}
    \end{cases}\]
  In order to extend this construction to clauses
  $c^i=\ell_1^i\lor \ell_2^i\lor \ell_3^i$, we define a disjunction
  operator as follows:
  \[(R_1\lor R_2)(Q)\coloneqq R_1(Q)+R_2(Q)-\sum_{Q_1\cup Q_2=Q}
    R_1(Q_1)\cdot R_2(Q_2).\]
  This gives us for the corresponding value function
  \begin{align*}
    (\valfunc_1\lor \valfunc_2)(M)&=\sum_{Q\subseteq M}(R_1\lor R_2)(Q)\\
    &=\sum_{Q\subseteq M}(R_1(Q)+R_2(Q)-\sum_{Q_1\cup Q_2=Q}R_1(Q_1)\cdot R_2(Q_2))\\
    &=\sum_{Q\subseteq M}R_1(Q)+\sum_{Q\subseteq M}R_2(Q)-\sum_{Q\subseteq M}\sum_{Q_1\cup Q_2=Q}R_1(Q_1)\cdot R_2(Q_2)\\
    &=\valfunc_1(M)+\valfunc_2(M)-(\sum_{Q\subseteq M}R_1(Q))(\sum_{Q\subseteq M}R_2(Q))\\
    &=\valfunc_1(M)+\valfunc_2(M)-\valfunc_1(M)\cdot \valfunc_2(M)
  \end{align*}
  which for binary rewards (i.e., rewards whose corresponding value
  function has a value space of $\{0,1\}$) can be written as the
  maximum of $\valfunc_1(M)$ and $\valfunc_2(M)$, yielding
  \[
    \val_{c^i}^D\coloneqq\val_{\ell_1^i}^D\lor \val_{\ell_2^i}^D\lor
    \val_{\ell_3^i}^D=\max\{\val_{\ell_1^i}^D,\val_{\ell_2^i}^D,\val_{\ell_3^i}^D\}
    =
    \begin{cases}
      1&\text{if $c^i$ is true,}\\
      0&\text{if $c^i$ is false.}
    \end{cases}
  \]

  To prove that the disjunction $R_1\lor R_2$ can be computed in
  polynomial time, we emphasize that, if the supports of $R_1$ and
  $R_2$ are polynomial in size (which they clearly are), the sum
  $\sum_{Q_1\cup Q_2=Q}R_1(Q_1)\cdot R_2(Q_2)$ only adds a polynomial
  number of non-zero elements.
  In fact, viewing $R_1$ and $R_2$ as dictionaries or partial
  functions taking on only their non-zero values, we can construct
  $R_1\lor R_2$ by simply iterating over the supports of $R_1$ and
  $R_2$, adding entries in $R_1\lor R_2$ if necessary.

  Finally, with a construction of a reward function for clauses in
  mind, we define the reward function of the Petri net as
  \[R(Q)\coloneqq\sum_{i\in I}R_{c^i}(Q),\] where $R_{c^i}$ is the
  reward function constructed for clause $c^i$, whence
  \[\val^D=\sum_{i\in I}\val^D_{c^i}.\]

  This construction gives us a bijection
  $\Phi:\funcset{\mathcal{X}}{\{0,1\}}\rightarrow\mathcal{P}(C),$
  mapping assignments $\mathcal{A}:\mathcal{X}\rightarrow\{0,1\}$ to
  sets
  $D_{\mathcal{A}}=\Phi(\mathcal{A})\coloneqq\set{t_x\in T\mid
    \mathcal{A}(x)=1}$ of deactivated transitions, satisfying
  \[\mathcal{A}\text{ is a model for } \psi =
    \bigwedge_{i\in I}c^i \iff
    \val^{D_{\mathcal{A}}}\geq|I| \iff \val^{D_{\mathcal{A}}} >
    |I|-1.\] This shows that the propositional formula
  $\bigwedge_{i\in I}c^i$ is satisfiable if and only if there exists a
  policy $D\subseteq C$ such that $\val^{D_{\mathcal{A}}} > |I|-1$,
  proving the reduction $\threesat \leq_p \fconpol$ and, thus, the
  $\np$-hardness of \fconpol.

  The result that \fconpol\ lies in $\np$ has been shown in
  Proposition~\ref{prop:fconval-p-np}.  \qed
\end{proof}

\section{An algorithm for SAFC decision nets}\label{sec:numerical}

Here we present a partial-order algorithm for solving the policy
problem for SAFC (decision) nets.
It takes such a net and converts it into a formula for an SMT solver.
We will assume the following, which is also a requirement for
occurrence nets:
\begin{assumption}\label{ass:initial_marking_input}
  For all places $p\in m_0$:
  $\pre{p}\coloneqq\{t\in T\mid p\in\post{t}\}=\emptyset$.
\end{assumption}
This is a mild assumption since any transition $t\in{^\bullet}p$ for a
place $p\in m_0$ in a safe and acyclic net has to be dead as all
places can only be marked once.

We are now using the notion of (branching) cells, introduced in
Section~\ref{sec:preliminaries}: The fact that the SDPN is safe,
acyclic and free-choice ensures that choices in different cells are
taken independently from another, so that the probability of a
configuration $\tau\in\validsubconfs{N}$ under a specific deactivation
pattern $D\subseteq C$ is given by
\[\mP^D(\transitions\supseteq\tau)
  =\prod_{t\in\tau}\frac{\chi_{T\setminus D}(t)\cdot\Lambda(t)}{
  \sum_{t\in\mC_t\setminus D} \Lambda(t)} =
  \begin{cases}
    0 & \mbox{if } \tau\cap D \neq \emptyset \\
    \prod_{t\in\tau} \frac{\Lambda(t)}{\sum_{t'\in\mC_t\setminus
        D}\Lambda(t')} & \mbox{otherwise}
  \end{cases}
\]
where $\chi_{T\setminus D}$ is the characteristic function of
$T\setminus D$ and $0/0$ is defined to yield $0$.

The general idea of the algorithm is to rewrite the reward function
$R:\mathcal{P}(P)\rightarrow\mR$ on sets of places to a reward
function on sets of transitions that yields a compact formula for
computing the value $\val^D$ for specific sets $D$ (i.e., solving
\safcval), that we can also use to solve the policy problem \safcpol\
via an SMT solver.

We first need some definitions:

\begin{definition}
  For a maximal configuration $\tau\in\validconfs{N_D}$ for a given
  deactivation pattern $D\subseteq C$, we define its set of prefixes
  in $\validsubconfs{N_D}$ to be
  \[\pref^D(\tau)\coloneqq\set{\tau'\in\validsubconfs{N_D}
  \mid\tau'\subseteq\tau}\]
  which corresponds to all configurations that can lead to the
  configuration $\tau$.
  We also define the set of extensions of a configuration
  $\tau\in\validsubconfs{N_D}$ in $\validconfs{N_D}$, which
  corresponds to all maximal configurations that $\tau$ can lead to,
  as
  \[\exte^D(\tau)\coloneqq\set{\tau'\in\validconfs{N_D}
  \mid\tau\subseteq\tau'}.\]
\end{definition}

\begin{definition}
  \label{def:consistent-reward-fct}
  Let $N$ be a Petri net with a reward function
  $R\colon \mathcal{P}(P)\to\mR$ on places and a deactivation pattern
  $D$.
  A reward function $[R]\colon \mathcal{P}(T)\to\mR$ on transitions is
  called \emph{consistent} with $R$ if for each firing sequence
  $\mu\in \firingsequences{N_D}$:
  \[ V(\places(\mu))=\sum_{Q\subseteq\places(\mu)}R(Q)
    =\sum_{\tau\in\pref^D{(\transitions(\mu))}}[R](\tau). \]
\end{definition}

This gives us the following alternative method to determine the
expected value for a net (with given policy $D$):

\begin{lemmarep}\label{lem:consistent-reward-fct}
  Using the setting of Definition~\ref{def:consistent-reward-fct},
  whenever $[R]$ is consistent with the reward function $R$ and
  $[R](\tau)=0$ for all $\tau\not\in \validsubconfs{N}$, the
  expected value for the net $N$ under the constant policy $D$ is:
  \[\val^D =
    \sum_{\tau\subseteq T}\mP^D(\transitions\supseteq\tau)\cdot
    [R](\tau).\]
\end{lemmarep}

\begin{proof}
  \begin{align*}
    \val^D & = \sum_{\mu\in \firingsequences{N_D}} \mP^D(\mu)
    \sum_{Q\subseteq \places(\mu)} R(Q) \\
    & =
    \sum_{\mu\in\firingsequences{N_D}} \mP^D(\mu)
    \sum_{\tau'\in\pref^D(\transitions(\mu))} [R](\tau') \\
    & =
    \sum_{\tau\in\validconfs{N_D}}\mP^D(\transitions=\tau)
    \sum_{\tau'\in\pref^D(\tau)}[R](\tau') \\
    & = \sum_{\tau\in\validconfs{N_D}}\mP^D(\transitions=\tau)
    \sum_{\tau\in\exte^D(\tau')}[R](\tau') \\
    & = \sum_{\tau\in\validconfs{N_D}}\sum_{\tau\in\exte^D(\tau')}
    \mP^D(\transitions=\tau)\cdot [R](\tau') \\
    & = \sum_{\tau'\in\validsubconfs{N_D}}
    \sum_{\tau\in\exte^D(\tau')} \mP^D(\transitions = \tau)\cdot
    [R](\tau') \\
    & = \sum_{\tau'\in\validsubconfs{N_D}}
    \mP^D(\transitions\supseteq\tau')\cdot [R](\tau'),\\
    & = \sum_{\tau'\in\mathcal{P}(T)}
    \mP^D(\transitions\supseteq\tau')\cdot [R](\tau'),
  \end{align*}
  where we use that $\tau\in\pref^D(\tau')$ if and only if
  $\tau'\in\exte^D(\tau)$ for (maximal) configurations
  $\tau'\in\validconfs{N_D}$ and $\tau\in\validsubconfs{N_D}$.
  Furthermore, we rely on the fact that
  $\mP^D(\transitions\supseteq\tau') =
  \sum_{\tau\in\exte^D(\tau')}\mP^D(\transitions=\tau)$ for
  $\tau'\in\validsubconfs{N_D}$ and
  $\mP^D(\transitions\supseteq\tau')=0$ for
  $\tau'\in\validsubconfs{N}\setminus\validsubconfs{N_D}$.  \qed
\end{proof}

Note that $[R](\transitions(\mu))\coloneqq V(\places(\mu))$ for
$\mu\in\firingsequences{N}$ fulfills these properties
trivially.
However, rewarding only maximal configurations can lead, already in
occurrence nets with some concurrency, to an exponential support
(w.r.t. the size of the net and its reward function).
The goal of our algorithm is to instead make use of the sum over the
configurations by rewarding reached places immediately in the
corresponding configuration, generating a function $[R]$ that fulfills
the properties above and whose support remains of polynomial size in occurrence
nets.
Hence, we have some form of partial-order technique, in particular
concurrent transitions receive the reward independently of each other
(if the reward is not dependent on firing both of them).

The rewriting process is performed by iteratively \enquote*{removing
  maximal cells} and resembles a form of backward-search algorithm.
First of all, $\preceq^*_N$ (the reflexive and transitive closure of
causality $\prec_N$) induces a partial order $\sqsubseteq$ on the set
$\bc{N}$ of cells via
\[\forall \mC,\mC'\in \bc{N}:\mC\sqsubseteq\mC' \Longleftrightarrow
  \exists t\in\mC,t'\in\mC':t\preceq_N^* t'.\]

Let all cells $(\mC_1,\dots,\mC_m)$ with $m=|\bc{N}|$ be ordered
conforming to $\sqsubseteq$, then we let $N_k$ denote the Petri net
consisting of places
$P_k\coloneqq P\setminus(\bigcup_{l>k}\post{\mC_l})\cup(\bigcup_{l\leq
  k}\post{\mC_l})$
(where the union with the post-sets is only necessary if
backward-conflicts exist) and transitions
$T_k\coloneqq\bigcup_{l\leq k}\mC_l$, the remaining components being
accordingly restricted (note that the initial marking $m_0$ is still
contained in $P_k$ by Assumption~\ref{ass:initial_marking_input}).
In particular, it holds that $N=N_m$ as well as $T_0=\emptyset$
and $P_0=\{p\in P\mid\forall t\in T:p\notin\post{t}\}$.

Let $N$ be a Petri net with deactivation pattern $D$,
$\mu\in\firingsequences{N_D}$ be a firing sequence and
$k\in\set{1,\dots,|\bc{N}|}$.
We write
$\transitions_{\leq k}(\mu)\coloneqq\transitions(\mu)\cap T_k$ for the
transitions in the first $k$ cells and
$\transitions_{>k}(\mu)\coloneqq\transitions(\mu)\setminus T_k$ for
the transitions in the cells after the $k$-th cell as well as
$\places_{\leq k}(\mu)\coloneqq
m_0\cup(\bigcup_{t\in\transitions_{\leq k}(\mu)}\post{t})$
for the places reached after all transitions in the first $k$ cells
were fired.

\medskip

We will now construct auxiliary reward functions $R[k]$ that take
pairs of a set of places ($U\subseteq P_k$) and of transitions
($V\subseteq T\setminus T_k$) as input and return a reward.
Intuitively, $R[k](U,V)$ corresponds to the reward for reaching all
places in $U$ and then firing all transitions in $V$ afterwards where
reaching $U$ ensures that all transitions in $V$ can fire.

Starting with the reward function
$R[m]:\mathcal{P}(P)\times\set{\emptyset}\rightarrow\mR,
(M,\emptyset)\mapsto R(M)$, we iteratively compute reward functions
$R[k]:\mathcal{P}(P_k)\times\mathcal{P}(T\setminus T_k)\rightarrow\mR$
for $k\geq 0$:
\[R[k](U,V)\coloneqq
  \begin{cases}
    R[k+1](U,V) & \text{if }\mC_{k+1}\cap V=\emptyset \\
    \sum\limits_{\substack{U'\cap \post{t}\neq \emptyset \\
        U = U'\backslash \post{t}\cup \pre{t}}}
    R[k+1](U',V\backslash\{t\}) &
    \text{if }\mC_{k+1}\cap V=\set{t} \\
    0 & \text{otherwise}
  \end{cases}
\]

The first case thus describes a scenario where no transition from the
$(k+1)$-th cell is involved while the second case sums up all rewards
that are reached when some transition $t$ in the cell has to be fired
(that is, all rewards that are given when some of the places in
$\post{t}$ are reached).
We give non-zero values only to sets $V$ that contain at most one
transition of each cell and $U$ has to contain the full pre-set of $t$
of the transition $t$ removed from $V$.
This is done in order to ensure that in subsequent steps those
transitions that generate $\pre{t}$ are in the set to which we assign
the reward.
This guarantees that $V$ is always a configuration of $N$ after
marking $U$ while $R[k](U,V)$ is zero if the transitions in $V$ cannot
be fired after $U$.
In this way, rewards are ultimately only given to configurations and
to no other sets of transitions, enabling us later to compute the
probabilities of those configurations.

And if $N$ is an occurrence net, every entry in $R[k+1]$ produces at
most one entry in $R[k]$, meaning that $\supp(R[k])\le\supp(R[k+1])$.

Now we can prove that the value of a firing sequence is invariant
when rewriting the auxiliary reward functions as described above.

\begin{propositionrep}
  The auxiliary reward functions satisfy
  \[
    \sum_{V\subseteq\transitions_{>k}(\mu)}
    \sum_{U\subseteq\places_{\leq k}(\mu)}R[k](U,V) =
    \sum_{V\subseteq\transitions_{>k+1}(\mu)}
    \sum_{U\subseteq\places_{\leq k+1}(\mu)}R[k+1](U,V),
  \]
  for $k\in \set{0,\dots,|\bc{N}|-1}$.

  Hence, for every
  $\mu\in\firingsequences{N}$
  \[V(\places(\mu))
    =\sum_{U\subseteq\places(\mu)}R[|\bc{N}|](U,\emptyset)
    =\sum_{V\subseteq\transitions_{>k}(\mu)}
    \sum_{U\subseteq\places_{\leq k}(\mu)}R[k](U,V),
  \]
  which means that we obtain a reward function on transitions
  consistent with $R$ by defining
  $[R]:\mathcal{P}(T)\rightarrow\mR$ as
  \[[R](V)\coloneqq\sum_{U\subseteq m_0}R[0](U,V). \]
\end{propositionrep}

\begin{proof}
  Note that, due to safety of the net, we have
  \[\places_{\leq k+1}(\mu)=\begin{cases}
      \places_{\leq k}(\mu)\text{ }\dot{\cup}\text{ }\post{t}
      & \text{if }t\in\transitions(\mu)\cap\mC_{k+1}\neq\emptyset,\\
      \places_{\leq k}(\mu)&\text{if
      }\transitions(\mu)\cap\mC_{k+1}=\emptyset.
    \end{cases}\]
  As such, if $\transitions(\mu)\cap\mC_{k+1}=\emptyset$, i.e., no
  transition from the $(k+1)$-th cell fired in the sequence $\mu$, we
  have
  \[\sum_{V\subseteq\transitions_{>k}(\mu)}\sum_{U\subseteq
      \places_{\leq k}(\mu)}R[k](U,V)
    =\sum_{V\subseteq\transitions_{>k+1}(\mu)}
      \sum_{U\subseteq\places_{\leq k+1}(\mu)}R[k+1](U,V)\]
  If, on the other hand, $t\in\transitions(\mu)\cap\mC_{k+1}$ is the
  unique transition from $\mC_{k+1}$ that fired in $\mu$, we have
  \begin{align*}
    &\sum_{V\subseteq\transitions_{>k}(\mu)}
      \sum_{U\subseteq\places_{\leq k}(\mu)}R[k](U,V)\\
    &=\sum_{V\subseteq\transitions_{>k}(\mu) \setminus\set{t}}
      \sum_{U\subseteq\places_{\leq k}(\mu)}
      \left(R[k](U,V)+R[k](U,V\cup\set{t})\right)\\
    &=\sum_{V\subseteq\transitions_{>k+1}(\mu)}
      \sum_{U\subseteq\places_{\leq k}(\mu)}
      \Bigg(R[k+1](U,V) \mathop{+}
      \sum\limits_{
        \substack{U'\cap \post{t}\neq\emptyset \\
        U = U'\backslash \post{t}\cup \pre{t}}
      } R[k+1](U',V) \Bigg)\\
    &=\sum_{V\subseteq\transitions_{>k+1}(\mu)}
      \Bigg(\sum_{U\subseteq\places_{\leq k}(\mu)}
      R[k+1](U,V) \mathop{+} \sum\limits_{
        \substack{U\cap \post{t}\neq\emptyset \\
        U\backslash\post{t}\cup\pre{t}\subseteq\places_{\leq k}(\mu)}
      } R[k+1](U,V) \Bigg)\\
    &=\sum_{V\subseteq\transitions_{>k+1}(\mu)}
      \Bigg(\sum_{U\subseteq\places_{\leq k}(\mu)}
      R[k+1](U,V) \mathop{+}\sum\limits_{
        \substack{U\subseteq \places_{\leq k+1}(\mu) \\ U\cap
        \post{t}\neq \emptyset}
      } R[k+1](U,V)\Bigg)\\
    &=\sum_{V\subseteq\transitions_{>k+1}(\mu)}
      \sum_{U\subseteq\places_{\leq k}(\mu)}
      \left(R[k+1](U,V)+
      \sum_{\emptyset\neq O\subseteq\post{t}}R[k+1](U\cup O,V)
      \right)\\
    &=\sum_{V\subseteq\transitions_{>k+1}(\mu)}
      \sum_{U\subseteq\places_{\leq k}(\mu)
      \text{ }\dot{\cup}\text{ }\post{t}}R[k+1](U,V)\\
    &=\sum_{V\subseteq\transitions_{>k+1}(\mu)}
      \sum_{U\subseteq\places_{\leq k+1}(\mu)}R[k+1](U,V).
  \end{align*} \qed
\end{proof}

This leads to the following corollary:

\begin{corollaryrep}
  Given a net $N$ and a deactivation pattern $D$, we can calculate the
  expected value
  \[\val^D=\mE[ \valfunc\circ\places]
    =\sum_{\tau\subseteq T}\prod_{t\in\tau}\frac{\chi_{T\setminus
        D}(t)\cdot\Lambda(t)}{\sum_{t'\in\mC_t\setminus
        D}\Lambda(t')}[R](\tau). \]
\end{corollaryrep}

\begin{proof}
  Lemma~\ref{lem:consistent-reward-fct} gives us
  \[\val^D=\sum_{\tau\subseteq T}
    \mP^D(\transitions\supseteq\tau) \cdot [R](\tau)\] as $[R]$ is
  consistent with the reward function $R$.

  We also observe that
  \[ \mP^D(\transitions\supseteq\tau) =
    \prod_{t\in\tau}\frac{\Lambda(t)}{\sum_{t'\in\mC_t\setminus
        D}\Lambda(t')} \]
  for $\tau\in\mathcal{C}(N_D)$ since the probability of a
  configuration can be determined by multiplying the probabilities of
  all its transitions, where the probability of a transition is its
  normalized rate, where the normalization is performed wrt.\ to all
  other deactivated transitions in the cell $\mathbb{C}_t$ of $t$.

  Hence, the equality above can be extended to:
  \begin{align*}
    \val^D&=\sum_{\tau\in\mathcal{C}(N_D)}\prod_{t\in\tau}
      \frac{\Lambda(t)}{\sum_{t'\in\mC_t\setminus D}
        \Lambda(t')}\cdot [R](\tau)\\
    &=\sum_{\tau\in\mathcal{C}(N)}\prod_{t\in\tau}
      \frac{\chi_{T\setminus D}(t)\Lambda(t)}{
        \sum_{t'\in\mC_t\setminus D}\Lambda(t')}[R](\tau)\\
    &=\sum_{\tau\subseteq T}\prod_{t\in\tau}\frac{\chi_{T\setminus
        D}(t)\Lambda(t)}{\sum_{t'\in\mC_t\setminus
        D}\Lambda(t')}[R](\tau)
  \end{align*}
  where we use that $[R](\tau)=0$ for all
  $\tau\in\mathcal{P}(T)\setminus\mathcal{C}(N)$. \qed
\end{proof}

Checking whether some deactivation pattern $D$ exists such that this
term is greater than some bound $p$ can be checked by an SMT solver.

Note that, in contrast to the naive definition of $[R]$ only on
maximal configurations, this algorithm constructs a reward
function on configurations that, for occurrence nets, has a support
with at most $\supp(R)$ elements.
For arbitrary SAFC nets, the support of $[R]$ might be of exponential
size.

\begin{adjustbox}{valign=C,raise=\strutheight,minipage={1\linewidth}}
  \begin{wrapfigure}[18]{r}{0.35\linewidth} 
  \vspace*{-0.5cm}
  \centering
  \begin{tikzpicture}[node distance=0.9cm,>=stealth',bend angle=45,auto,every label/.style={align=left}]
  \tikzstyle{place}=[circle,thick,draw=black,fill=white,minimum size=5mm]
  \tikzstyle{transition}=[rectangle,thick,draw=black!75,
                fill=black!5,minimum size=4mm]
  \begin{scope}
    \node [place, tokens=1] (p1) [label=90:$p_1$]{};
    \node [place, tokens=1] (p2) [label=90:$p_2$, right of=p1, xshift=9mm, yshift=0mm]{};
    \node [place] (p3) [label=180:$p_3$, below of=p1, xshift=4.5mm, yshift=-9mm]{};
    \node [place] (p4) [label=0:$p_4$, below of=p2, xshift=-4.5mm, yshift=-9mm]{};
    \node [place] (p7) [label=270:$p_7$, below of=p4, fill=red!50, yshift=-9mm]{};
    \node [place] (p6) [label=270:$p_6$, left of=p7, fill=yellow!50, yshift=0mm]{};
    \node [place] (p5) [label=270:$p_5$, left of=p6, fill=yellow!50, yshift=0mm]{};

    \node [transition,double] (t1) [label=135:$t_1$, below of=p1, xshift=-4.5mm, yshift=0mm] {}
    edge [pre] (p1)
    edge [post] (p5);

    \node [transition] (t2) [label={[label distance=-0mm]85:$t_2$}, below of=p1, xshift=4.5mm, yshift=0mm] {}
    edge [pre] (p1)
    edge [post] (p3);

    \node [transition] (t3) [label={[label distance=-0mm]95:$t_3$}, below of=p2, xshift=-4.5mm, yshift=0mm] {}
    edge [pre] (p2)
    edge [post] (p4);

    \node [transition] (t4) [label=45:$t_4$, below of=p2, xshift=4.5mm, yshift=0mm] {}
    edge [pre] (p2);

    \node [transition,double] (t5) [label=180:$t_5$, below of=p3, xshift=0mm, yshift=0mm] {}
    edge [pre] (p3)
    edge [pre] (p4)
    edge [post] (p5)
    edge [post] (p6);

    \node [transition,double] (t6) [label=0:$t_6$, below of=p4, xshift=0mm, yshift=0mm] {}
    edge [pre] (p3)
    edge [pre] (p4)
    edge [post] (p6)
    edge [post] (p7);
  \end{scope}
  \end{tikzpicture}
  \vspace*{-0.8cm}
  \caption{A SAFC decision net. The goal is to mark one or both of the
    yellow places at some point without ever marking the red
    place.}\label{fig:SAFC}
\end{wrapfigure}
\strut{}
\vspace*{00cm} 
\begin{example}
  We take the Petri net from Figure~\ref{fig:SAFC} as an example
  (where all transitions have firing rate 1). The reward function $R$
  is given in the table below. By using the inclusion-exclusion
  principle we ensure that one obtains reward~$1$ if one or both of
  the yellow places are marked at some point without ever marking the
  red place.

  The optimal strategy is obviously to only deactivate the one
  transition ($t_6$) which would mark the red place.

  The net has three cells $\mC_1=\{t_1,t_2\},\mC_2=\{t_3,t_4\},$ and
  $\mC_3=\{t_5,t_6\}$ where $\mC_1,\mC_2\sqsubseteq\mC_3$.  As such,
  $R[3] = R$ with $R$ below and obtain $R[2]$ (due to
  $P_2=\{p_1,p_2,p_3,p_4,p_5\}$).  In the next step, we get (by removing
  $t_3$ and $t_4$) $R[1]$ and finally $R[0]$, from which we can
  derive $[R]$, the reward function on transitions, as described
  above.

  This allows us to write the value for a set $D$ of deactivated
  transitions as follows (where if both $t_5,t_6\in D$, we assume the
  last quotient to be zero)
  \[\val^D=\frac{\chi_{T\setminus D}(t_1)}{\chi_{T\setminus D}(t_1)+1}
    +\frac{1}{\chi_{T\setminus
        D}(t_1)+1}\frac{1}{2}\frac{\chi_{T\setminus
        D}(t_5)}{\chi_{T\setminus D}(t_5)+\chi_{T\setminus D}(t_6)}\]

  \medskip

  \hrule
  \vspace{-0.4cm}
  \begin{align*}
    R=&[\{p_5\}:1,\{p_6\}:1,\{p_5,p_6\}:-1,\{p_5,p_7\}:-1,
    \{p_6,p_7\}:-1,\{p_5,p_6,p_7\}:1] \\
    R[2]=&[(\{p_5\},\emptyset):1,(\{p_3,p_4\},\{t_5\}):1,
    (\{p_3,p_4,p_5\},\{t_6\}):-1] \\
    R[1]=&[(\{p_5\},\emptyset):1,(\{p_2,p_3\},\{t_3,t_5\}):1,(\{p_2,p_3,p_5\},\{t_3,t_6\}):-1]
    \\
    R[0]=&[(\{p_1\},\{t_1\}):1,(\{p_1,p_2\},\{t_2,t_3,t_5\}):1] \\
    [R]=&[\{t_1\}:1,\{t_2,t_3,t_5\}:1]
  \end{align*}
  \vspace{-0.5cm}
  \hrule

  \medskip

  Writing $x_i\coloneqq\chi_{T\setminus D}(t_i)\in\{0,1\},i=1,5,6,$
  the resulting inequality
  \[\frac{x_1}{x_1+1}+\frac{1}{2}\frac{1}{x_1+1}\frac{x_5}{x_5+x_6} >
    p\]
  can now be solved by an SMT solver with Boolean variables $x_1,x_5,$
  and $x_6$ (i.e., $x_1,x_5,x_6\in\{0,1\}$).

\end{example}
\end{adjustbox}

\paragraph*{Runtime results:}
To test the performance of our algorithm, we performed runtime tests
on specific families of simple stochastic decision Petri nets,
focussing on the impact of concurrency and backward-conflicts on its
runtime.
All families are based on a series of simple branching cells
each containing two transitions, one controllable and one
non-controllable, reliant on one place as a precondition.
Each non-controllable transition marks a place to which we randomly
assigned a reward according to a normal distribution (in particular,
it can be negative).
The families differ in how these cells are connected, testing
performance with concurrency, backward-conflicts, and sequential
problems,
respectively (for a detailed overview of the experiments see
\short{\arxive}\full{Appendix~\ref{app:runtime-results}}).

Rewriting the reward function (and, thus, solving the value problem)
produced expected results: Runtimes on nets with many
backward-conflicts are exponential while the rewriting of reward
functions of occurrence nets exhibits a much better performance,
reflecting its polynomial complexity.

To solve the policy problem based on the rewritten reward function, we
compared the performances of naively calculating the values of each
possible deactivation pattern with using an SMT solver (Microsoft's
z3, see also~\cite{z3}).
Tests showed a clear impact on the representation of the control
variables (describing the deactivation set $D$) as booleans or as
integers bounded by $0$ and $1$ with the latter showing a better
performance.
Furthermore, the runtime of solving the rewritten formula with an SMT
solver showed a high variance on random reward values.
Nonetheless, the results show the clear benefit of using the SMT
solver on the rewritten formula in scenarios with a high amount of
concurrency, with much faster runtimes than the brute force
approach.
In scenarios without concurrency, this benefit vanishes, and in
scenarios with many backward-conflicts, the brute force approach is
considerably faster than solving the rewritten function with an SMT
solver.
The latter effect can be explained by the rewritten reward function
$[R]$ having an exponential support in this scenario.

All in all, the runtime results reflect the well-known drawbacks and
benefits of most partial-order techniques, excelling in scenarios with
high concurrency while having a reduced performance if there are
backward- and self-conflicts.

\begin{toappendix}
  \section{Runtime results}\label{app:runtime-results}
  We performed runtime tests on three families of SAFC-SDPNs,
  each with a simple generation procedure with randomly chosen rewards
  and with a clear focus on either concurrency, high degree of self-
  and backward-conflicts, and absence of both, respectively.

  The first family $\mathcal{N}_1$ consists of Petri nets with $n$
  concurrent simple branching cells, each with one initially marked
  place on which two transitions depend.  One of these transitions is
  not controllable and leads to a place with a random reward sampled
  according to the standard normal distribution.  The other transition
  is controllable but marks no place.  Formally, a net of this family
  is thus given by $P=\{p_1,\dots,p_{2n}\},T=(t_1,\dots,t_{2n})$ with
  $\pre{t_{2k-1}}=\pre{t_{2k}}=\{p_{2k-1}\}$ and
  $\post{t_{2k-1}}=\{p_{2k}\},\post{t_{2k}}=\emptyset$ for
  $k=1,\dots,n$, $m_0=\{p_{2k-1}\mid k=1,\dots,n\}$, $\Lambda\equiv1$,
  $C=\{t_{2k}\mid k=1,\dots,n\}$, and $R$ only non-zero for
  $\{p_{2k}\}$ (randomly generated according to standard normal
  distribution) for $k=1,\dots,n$ (see also
  Figure~\ref{fig:concurrency}).  Generating these nets with random
  rewards for each (post-)place as well as a random bound $p$ for the
  policy problem (also sampled according to the standard normal
  distribution) allows for a variety of nets and problems (some of
  which might not be solvable) to test our algorithms with a focus on
  its performance on highly concurrent nets.  While the optimal
  strategy for each of these nets is to deactivate any transition that
  is in a cell with a positively rewarded place and activate all
  others, the random generation ensures that this optimal strategy
  results in different optimal sets $D$ of deactivated transitions.
  Note, however, that the corresponding MDP of these nets will have an
  exponential size due to the $2^n\cdot n!$ possible firing sequences.

  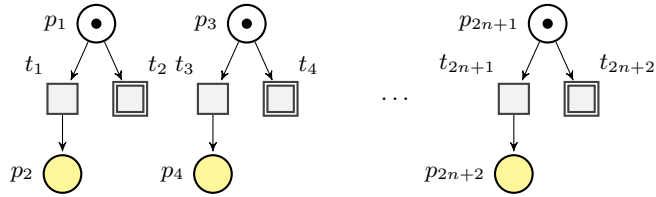
\begin{figure} [h]
  \centering
  \begin{tikzpicture}[node distance=1cm,>=stealth',bend angle=45,auto,every label/.style={align=left}]
  \tikzstyle{place}=[circle,thick,draw=black,fill=white,minimum size=5mm]
  \tikzstyle{transition}=[rectangle,thick,draw=black!75,
                fill=black!5,minimum size=4mm]
  \begin{scope}
    \node [place, tokens=1] (p1) [label=180:$p_1$]{};
    \node [place, tokens=1] (p3) [label=180:$p_3$, right of=p1, xshift=10mm, yshift=0mm]{};
    \node [place, tokens=1] (p5) [label=180:$p_{2n+1}$, right of=p3, xshift=30mm, yshift=0mm]{};

    \node [transition] (t1) [label=120:$t_1$, below of=p1, xshift=-4.5mm, yshift=0mm] {}
    edge [pre] (p1);

    \node [transition,double] (t2) [label=60:$t_2$, below of=p1, xshift=4.5mm, yshift=0mm] {}
    edge [pre] (p1);

    \node [transition] (t3) [label=120:$t_3$, below of=p3, xshift=-4.5mm, yshift=0mm] {}
    edge [pre] (p3);

    \node [transition,double] (t4) [label=60:$t_4$, below of=p3, xshift=4.5mm, yshift=0mm] {}
    edge [pre] (p3);

    \node [transition] (t5) [label=120:$t_{2n+1}$, below of=p5, xshift=-4.5mm, yshift=0mm] {}
    edge [pre] (p5);

    \node [transition,double] (t6) [label=60:$t_{2n+2}$, below of=p5, xshift=4.5mm, yshift=0mm] {}
    edge [pre] (p5);

    \node at ($(t4)!.5!(t5)$) {\ldots};

    \node [place] (p2) [label=180:$p_2$, below of=t1, fill=yellow!50]{}
    edge [pre] (t1);
    \node [place] (p4) [label=180:$p_4$, below of=t3, fill=yellow!50]{}
    edge [pre] (t3);
    \node [place] (p6) [label=180:$p_{2n+2}$, below of=t5, fill=yellow!50]{}
    edge [pre] (t5);

  \end{scope}
  \end{tikzpicture}
  \caption{The family $\mathcal{N}_1$ of free-choice occurrence SDPNs with a high amount
    concurrency.  Yellow places yellow are being rewarded (or
    punished).}\label{fig:concurrency}
\end{figure}

  The second family $\mathcal{N}_2$ consists, similar to $\mathcal{N}_1$,
  of $n$ branching cells.
  However, the post-place of one cell is set to be the initial place
  of the next cell, resulting in a sequential line of branching cells
  of the same type as above.
  Formally, a net of this family is thus given by
  $P=\{p_1,\dots,p_{n+1}\},T=(t_1,\dots,t_{2n})$ with
  $\pre{t_{2k-1}}=\pre{t_{2k}}=\{p_k\}$ and
  $\post{t_{2k-1}}=\{p_{k+1}\},\post{t_{2k}}=\emptyset$ for
  $k=1,\dots,n$, $m_0=\{p_1\}$, $\Lambda\equiv1$,
  $C=\{t_{2k}\mid k=1,\dots,n\}$, and $R$ only non-zero for
  $\{p_{k}\}$ (randomly generated according to standard normal
  distribution) for $k=2,\dots,n+1$ (see also Figure~\ref{fig:sequential}).
  Finding an optimal strategy for these nets is a bit more intricate
  as in $\mathcal{N}_1$ since firing the controllable transition in
  any of these cells results in no reward for all subsequent cells.

  \begin{figure} [h]
  \centering
  \begin{tikzpicture}[node distance=1cm,>=stealth',bend angle=45,auto,every label/.style={align=left}]
  \tikzstyle{place}=[circle,thick,draw=black,fill=white,minimum size=5mm]
  \tikzstyle{transition}=[rectangle,thick,draw=black!75,
                fill=black!5,minimum size=4mm]
  \begin{scope}
    \node [place, tokens=1] (p1) [label=90:$p_1$]{};

    \node [transition] (t1) [label=90:$t_1$, right of=p1, xshift=0mm, yshift=5mm] {}
    edge [pre] (p1);
    \node [transition, double] (t2) [label=270:$t_2$, below of=t1, xshift=0mm, yshift=0mm] {}
    edge [pre] (p1);

    \node [place] (p2) [label=90:$p_2$, right of=t1, xshift=0mm, yshift=-5mm, fill=yellow!50]{}
    edge [pre] (t1);

    \node [transition] (t3) [label=90:$t_3$, right of=p2, xshift=0mm, yshift=5mm] {}
    edge [pre] (p2);
    \node [transition, double] (t4) [label=270:$t_4$, below of=t3, xshift=0mm, yshift=0mm] {}
    edge [pre] (p2);

    \node [place] (p3) [label=90:$p_3$, right of=t3, xshift=0mm, yshift=-5mm, fill=yellow!50]{}
    edge [pre] (t3);

    \node [place] (p4) [label=90:$p_{n+1}$, right of=p3, xshift=20mm, yshift=0mm, fill=yellow!50]{};

    \node at ($(p3)!.5!(p4)$) {\ldots};

    \node [transition] (t5) [label=90:$t_{2n+1}$, right of=p4, xshift=0mm, yshift=5mm] {}
    edge [pre] (p4);
    \node [transition, double] (t6) [label=270:$t_{2n+2}$, below of=t5, xshift=0mm, yshift=0mm] {}
    edge [pre] (p4);

    \node [place] (p5) [label=90:$p_{n+2}$, right of=t5, xshift=0mm, yshift=-5mm, fill=yellow!50]{}
    edge [pre] (t5);

  \end{scope}
  \end{tikzpicture}
  \caption{The family $\mathcal{N}_2$ family of free-choice occurrence SDPNs. Yellow
  places are being rewarded (or punished).}\label{fig:sequential}
\end{figure}
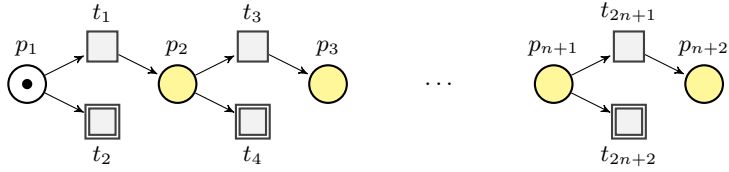

  The third and final family $\mathcal{N}_3$ also consists of $n$
  branching cells as the ones above.
  However, both transitions of one cell mark the initial place of the
  next cell (while, again, only the non-controllable transition marks
  the rewarded place).
  This ensures that all cells are fired (as in the concurrent family
  $\mathcal{N}_1$) but in sequence and, most importantly, with all but
  the first initial place having backward-conflicts.
  Formally, a net of this family is thus given by
  $P=\{p_1,\dots,p_{2n}\},T=(t_1,\dots,t_{2n})$ with
  $\pre{t_{2k-1}}=\pre{t_{2k}}=\{p_{2k-1}\}$ and
  $\post{t_{2k-1}}=\{p_{2k},p_{2k+1}\},\post{t_{2k}}=\{p_{2k+1}\}$
  for $k=1,\dots,n$, $m_0=\{p_1\}$, $\Lambda\equiv1$,
  $C=\{t_{2k}\mid k=1,\dots,n\}$, and $R$ only non-zero for
  $\{p_{2k}\}$ (randomly generated according to standard normal
  distribution) for $k=1,\dots,n$ (see also Figure~\ref{fig:backward-conflicts}).
  While an optimal strategy for this family of nets is the same as
  for the first one, deactivating exactly all controllable transitions
  in cells with positively rewarded places, the backward-conflicts
  result in an exponentially sized rewritten reward function $[R]$
  on the transitions with each of the $2^n$ possible configurations
  being rewarded.

  \begin{figure} [h]
  \centering
  \begin{tikzpicture}[node distance=1cm,>=stealth',bend angle=45,auto,every label/.style={align=left}]
  \tikzstyle{place}=[circle,thick,draw=black,fill=white,minimum size=5mm]
  \tikzstyle{transition}=[rectangle,thick,draw=black!75,
                fill=black!5,minimum size=4mm]
  \begin{scope}
    \node [place, tokens=1] (p1) [label=90:$p_1$]{};

    \node [transition] (t1) [label=90:$t_1$, right of=p1, xshift=0mm, yshift=5mm] {}
    edge [pre] (p1);
    \node [transition, double] (t2) [label=270:$t_2$, below of=t1, xshift=0mm, yshift=0mm] {}
    edge [pre] (p1);

    \node [place] (p2) [label=90:$p_2$, right of=t1, xshift=0mm, yshift=5mm, fill=yellow!50]{}
    edge [pre] (t1);
    \node [place] (p3) [label=90:$p_3$, right of=t1, xshift=0mm, yshift=-5mm]{}
    edge [pre] (t1)
    edge [pre] (t2);

    \node [transition] (t3) [label=90:$t_3$, right of=p3, xshift=0mm, yshift=5mm] {}
    edge [pre] (p3);
    \node [transition, double] (t4) [label=270:$t_4$, below of=t3, xshift=0mm, yshift=0mm] {}
    edge [pre] (p3);

    \node [place] (p4) [label=90:$p_4$, right of=t3, xshift=0mm, yshift=5mm, fill=yellow!50]{}
    edge [pre] (t3);
    \node [place] (p5) [label=90:$p_5$, right of=t3, xshift=0mm, yshift=-5mm]{}
    edge [pre] (t3)
    edge [pre] (t4);

    \node [place] (p6) [label=90:$p_{2n+1}$, right of=p5, xshift=20mm, yshift=0mm]{};

    \node at ($(p5)!.5!(p6)$) {\ldots};

    \node [transition] (t5) [label=90:$t_{2n+1}$, right of=p6, xshift=0mm, yshift=5mm] {}
    edge [pre] (p6);
    \node [transition, double] (t6) [label=270:$t_{2n+2}$, below of=t5, xshift=0mm, yshift=0mm] {}
    edge [pre] (p6);

    \node [place] (p7) [label=90:$p_{2n+2}$, right of=t5, xshift=0mm, yshift=5mm, fill=yellow!50]{}
    edge [pre] (t5);

  \end{scope}
  \end{tikzpicture}
  \caption{The family $\mathcal{N}_3$ of safe and acyclic free-choice SDPNs with
  backward-/self-conflicts. Yellow places are being rewarded (or punished).}\label{fig:backward-conflicts}
\end{figure}
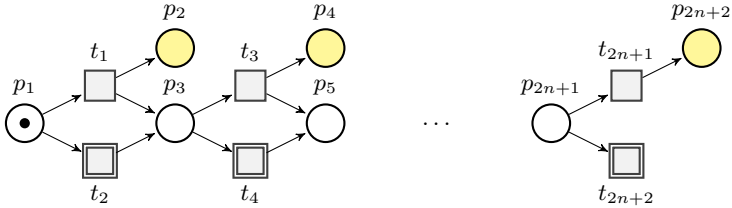

  For each of these families, we performed runtime tests on 25
  randomly generated nets (i.e., reward values) for each $n$ (as long
  as it was feasible).
  The tables and graphs below show the runtimes for the rewriting
  algorithm (i.e., solving the value problem), and solving the policy
  problem (again with randomly generated bound according to the
  standard normal distribution) both by iterating over all possible
  deactivation sets (brute force) or by using the z3 SMT solver.

  The runtimes of solving the value problem or rewriting the reward
  function using the algorithm as described in
  Section~\ref{sec:numerical} are as expected (see
  Figure~\ref{fig:rewriting}): For the family $\mathcal{N}_3$
  containing many backward-conflicts (in particular, a family of
  non-occurrence nets), the runtimes rise exponentially with the
  amount of branching cells (see also
  Table~\ref{tab:backward-rewriting}) while the algorithm performs
  much better on both families of free-choice occurrence nets
  $\mathcal{N}_1$ and $\mathcal{N}_2$ and, in particular, independent
  of the amount of concurrency present (in contrast to the expected
  time of solving the value problem using the corresponding MDP which
  would grow exponentially for family $\mathcal{N}_1$).  Furthermore,
  it is noteworthy that the performance of the rewriting algorithm
  mainly depends on the net structure, not on the randomly generated
  reward values which is reflected by the relatively small variance
  (see Tables~\ref{tab:concurrency-rewriting},
  \ref{tab:sequential-rewriting} and~\ref{tab:backward-rewriting}).

  \begin{table}
    \centering
    \begin{tabular}{r || r | r | r | r | }
      Size & Median & Mean & St.Dev. & 90\% quantile \\
      \hline\hline
      1   & 0     & 0.40      & 0.49    & 1     \\
      2   & 2     & 2.04      & 0.20    & 2     \\
      3   & 6     & 6.12      & 0.32    & 7     \\
      4   & 14    & 14.28     & 0.53    & 15    \\
      5   & 28    & 28.24     & 2.10    & 31    \\
      6   & 49    & 48.72     & 1.76    & 50    \\
      7   & 79    & 79.16     & 1.93    & 80    \\
      8   & 121   & 121.68    & 2.28    & 123   \\
      9   & 179   & 179.24    & 2.90    & 183   \\
      10  & 247   & 247.28    & 8.25    & 257   \\
      11  & 338   & 339.40    & 7.98    & 350   \\
      12  & 464   & 461.88    & 9.62    & 474   \\
      13  & 604   & 600.76    & 12.09   & 614   \\
      14  & 764   & 761.60    & 14.35   & 777   \\
      15  & 971   & 970.32    & 7.59    & 978   \\
      16  & 1196  & 1190.56   & 16.57   & 1208  \\
      17  & 1476  & 1471.68   & 21.47   & 1497  \\
      18  & 1799  & 1795.56   & 20.33   & 1817  \\
      19  & 2136  & 2145.96   & 46.15   & 2213  \\
      20  & 2589  & 2594.56   & 38.30   & 2640  \\
      21  & 3072  & 3069.04   & 30.38   & 3100  \\
      22  & 3632  & 3621.68   & 38.97   & 3665  \\
      23  & 4238  & 4230.68   & 30.44   & 4257  \\
      24  & 4956  & 4937.44   & 51.32   & 4985  \\
      25  & 5703  & 5698.40   & 53.57   & 5738  \\ \hline
    \end{tabular}
    \medskip
    \caption{Runtime results of the reward rewriting algorithm for family $\mathcal{N}_1$ in ms.}
    \label{tab:concurrency-rewriting}
  \end{table}
  \begin{table}
    \centering
    \begin{tabular}{r || r | r | r | r | }
      Size & Median & Mean & St.Dev. & 90\% quantile \\
      \hline\hline
      1   & 0     & 0.36      & 0.48    & 1     \\
      2   & 2     & 2.00      & 0.00    & 2     \\
      3   & 6     & 5.76      & 0.43    & 6     \\
      4   & 13    & 13.00     & 0.28    & 13    \\
      5   & 25    & 25.16     & 1.62    & 26    \\
      6   & 42    & 42.56     & 1.58    & 43    \\
      7   & 69    & 69.20     & 2.68    & 71    \\
      8   & 104   & 103.44    & 2.77    & 106   \\
      9   & 149   & 149.20    & 4.24    & 154   \\
      10  & 210   & 209.48    & 4.60    & 216   \\
      11  & 284   & 283.56    & 4.05    & 289   \\
      12  & 378   & 377.12    & 6.36    & 383   \\
      13  & 486   & 483.96    & 10.17   & 494   \\
      14  & 612   & 610.20    & 9.62    & 621   \\
      15  & 757   & 759.28    & 23.27   & 770   \\
      16  & 931   & 927.96    & 9.15    & 938   \\
      17  & 1133  & 1145.88   & 57.85   & 1189  \\
      18  & 1362  & 1362.92   & 27.37   & 1402  \\
      19  & 1632  & 1621.60   & 19.12   & 1639  \\
      20  & 1922  & 1916.00   & 18.51   & 1936  \\
      21  & 2260  & 2252.72   & 20.29   & 2269  \\
      22  & 2625  & 2618.68   & 28.96   & 2651  \\
      23  & 3140  & 3104.12   & 359.93  & 3299  \\
      24  & 3544  & 3535.00   & 44.97   & 3582  \\
      25  & 4061  & 4065.24   & 43.59   & 4118  \\ \hline
    \end{tabular}
     \medskip
   \caption{Runtime results of the reward rewriting algorithm for family $\mathcal{N}_2$ in ms.}
    \label{tab:sequential-rewriting}
  \end{table}
  \begin{table}
    \centering
    \begin{tabular}{r || r | r | r | r | }
      Size & Median & Mean & St.Dev. & 90\% quantile \\
      \hline\hline
      1   & 0     & 0.44      & 0.50    & 1     \\
      2   & 2     & 2.32      & 0.55    & 3     \\
      3   & 8     & 7.92      & 0.27    & 8     \\
      4   & 22    & 22.44     & 1.42    & 23    \\
      5   & 52    & 52.24     & 1.58    & 53    \\
      6   & 121   & 135.32    & 28.71   & 190   \\
      7   & 263   & 262.16    & 5.41    & 271   \\
      8   & 583   & 587.48    & 20.51   & 614   \\
      9   & 1281  & 1283.52   & 21.44   & 1314  \\
      10  & 2772  & 2762.00   & 29.15   & 2790  \\
      11  & 6018  & 6018.56   & 43.54   & 6069  \\
      12  & 13193 & 14144.36  & 1821.25 & 17375 \\
      13  & 41182 & 41155.08  & 478.37  & 41672 \\
      14  & 90170 & 90159.60  & 795.00  & 91114 \\ \hline
    \end{tabular}
    \medskip
    \caption{Runtime results of the reward rewriting algorithm for family $\mathcal{N}_3$ in ms.}
    \label{tab:backward-rewriting}
  \end{table}

  Turning our attention to the performance of solving the policy
  problem based on the rewritten reward function, we notice first and
  foremost that, in all three families, using an SMT solver produces
  highly varying runtimes.  This can be seen in both the standard
  deviations as seen in
  Tables~\ref{tab:concurrency-solving}, \ref{tab:sequential-solving},
  and~\ref{tab:backward-solving}, as well as in the boxplot diagrams
  in Figures~\ref{fig:smt-concurrency}, \ref{fig:smt-sequential},
  and~\ref{fig:smt-backward}.

  Boxplots are often used graphical representation to represent
  statistical results where, here, the \enquote*{box} describes the
  quartiles (i.e., 25\%-quantile to 75\%-quantile) with the yellow bar
  signifying the median.
  The whiskers have a maximal size of 1.5-times the length of the box
  (being smaller if no other values are present) and outliers (i.e.,
  values outside of the maximal whisker length) are marked as circles.

  Taking a look at the results for the highly concurrent family
  $\mathcal{N}_1$, we notice immediately that while using the brute
  force approach of comparing the values for all deactivation patterns
  produces exponential runtimes as expected, using the SMT solver has
  a much lower runtime where, despite growing variance, even the
  worst-case runtime seems to be polynomial apart from rare outliers
  (which are still much more performant than the brute force approach;
  see Table~\ref{tab:concurrency-solving} and
  Figure~\ref{fig:smt-concurrency}).  While this example is restricted
  to the simplest of concurrent SDPNs, this clearly reflects the
  strength of partial-order techniques to deal with concurrency where
  solving the corresponding MDP would necessarily produce exponential
  runtimes.

  The results for family $\mathcal{N}_2$ which lacks all concurrency
  but is still an occurrence net shows the extremely high variance in
  the runtime of the SMT solver (see
  Table~\ref{tab:sequential-solving} or
  Figure~\ref{fig:smt-sequential}).  As mentioned in the description
  of the family $\mathcal{N}_2$ above, finding an optimal deactivation
  pattern in this scenario is much more intricate than for the other
  two families and while, in the best case, the SMT solver is much
  faster than the brute force approach, most notably the unusually low
  median for $n=13$, it can also take a multitude longer in the
  \enquote*{more difficult} scenarios.

  Finally, while the results for family $\mathcal{N}_1$ showed the
  benefits of partial-order techniques, the results for family
  $\mathcal{N}_3$ reflect their drawbacks.  Note that the rewritten
  reward function $[R]$ on configurations already has an exponential
  support (w.r.t. $n$) which not only leads to a longer computation
  time of the value, explaining the much higher runtime of the brute
  force approach.  The exponential support also results in a much more
  complex SMT expression, the effect of which being that the z3 solver
  can only find answers efficiently for very small $n$ as can be seen
  in Table~\ref{tab:backward-solving} and
  Figure~\ref{fig:smt-backward}.

  \begin{table}
    \centering
    \begin{tabular}{ r||r | r | r | r || r | r | r | r |}
      & \multicolumn{4}{c||}{Brute Force} & \multicolumn{4}{c}{SMT Solver}\\
      Size & Median & Mean & St.Dev. & 90\% & Median & Mean & St. Dev. & 90\% \\
      \hline\hline
      1   & 0     & 0.00      & 0.00    & 0     & 4   & 4.20    & 0.40    & 5\\
      2   & 0     & 0.64      & 3.14    & 0     & 5   & 4.68    & 0.73    & 5\\
      3   & 0     & 0.60      & 2.94    & 0     & 5   & 5.00    & 0.89    & 6\\
      4   & 0     & 1.52      & 4.44    & 4     & 5   & 5.92    & 2.28    & 8\\
      5   & 15    & 8.88      & 7.73    & 16    & 5   & 6.56    & 2.35    & 10\\
      6   & 15    & 12.96     & 5.67    & 16    & 6   & 7.88    & 2.83    & 12\\
      7   & 31    & 30.28     & 4.58    & 33    & 8   & 9.60    & 5.36    & 15\\
      8   & 63    & 65.24     & 6.03    & 78    & 9   & 12.04   & 7.66    & 21\\
      9   & 149   & 148.76    & 8.02    & 157   & 8   & 10.52   & 5.95    & 18\\
      10  & 337   & 336.96    & 7.32    & 344   & 6   & 9.28    & 5.67    & 20\\
      11  & 780   & 774.00    & 16.96   & 791   & 10  & 16.12   & 13.01   & 29\\
      12  & 1701  & 1695.00   & 20.39   & 1718  & 9   & 13.20   & 10.71   & 24\\
      13  & 3759  & 3748.76   & 43.43   & 3796  & 8   & 18.36   & 18.04   & 49\\
      14  & 8264  & 8250.76   & 79.50   & 8331  & 7   & 14.12   & 12.10   & 31\\
      15  &       &           &         &       & 10  & 15.20   & 15.17   & 38\\
      16  &       &           &         &       & 10  & 18.84   & 20.87   & 45\\
      17  &       &           &         &       & 10  & 21.88   & 21.46   & 54\\
      18  &       &           &         &       & 15  & 32.00   & 33.81   & 85\\
      19  &       &           &         &       & 9   & 28.24   & 29.26   & 82\\
      20  &       &           &         &       & 22  & 39.68   & 38.46   & 99\\
      21  &       &           &         &       & 10  & 98.20   & 283.86  & 204\\
      22  &       &           &         &       & 15  & 90.64   & 241.91  & 157\\
      23  &       &           &         &       & 15  & 24.36   & 21.66   & 53\\
      24  &       &           &         &       & 21  & 227.88  & 602.75  & 230\\
      25  &       &           &         &       & 11  & 39.08   & 61.17   & 99\\
      \hline
    \end{tabular}
    \medskip
    \caption{Runtime results of solving the policy problem for family $\mathcal{N}_1$ based on its rewritten reward function in ms.}
    \label{tab:concurrency-solving}
  \end{table}
  \begin{table}
    \centering
    \begin{tabular}{ r||r | r | r | r || r | r | r | r |}
      & \multicolumn{4}{c||}{Brute Force} & \multicolumn{4}{c}{SMT Solver}\\
      Size & Median & Mean & St. Dev. & 90\% & Median & Mean & St. Dev. & 90\% \\
      \hline\hline
      1   & 0     & 0.04      & 0.20    & 0     & 3     & 3.24      & 0.81      & 4\\
      2   & 0     & 0.24      & 0.43    & 1     & 4     & 3.88      & 0.52      & 4\\
      3   & 1     & 0.64      & 0.48    & 1     & 4     & 4.68      & 0.93      & 6\\
      4   & 2     & 2.00      & 0.00    & 2     & 6     & 8.72      & 5.35      & 18\\
      5   & 6     & 5.60      & 0.49    & 6     & 10    & 14.92     & 12.64     & 33\\
      6   & 14    & 14.72     & 1.22    & 15    & 37    & 52.48     & 50.65     & 114\\
      7   & 35    & 35.44     & 1.13    & 36    & 68    & 166.88    & 187.83    & 477\\
      8   & 86    & 85.92     & 1.41    & 87    & 104   & 335.16    & 476.87    & 1207\\
      9   & 199   & 198.88    & 5.57    & 206   & 518   & 3064.64   & 4734.68   & 9352\\
      10  & 605   & 602.56    & 87.71   & 709   & 1506  & 6250.28   & 7470.27   & 14787\\
      11  & 1440  & 1415.28   & 101.87  & 1562  & 1291  & 10406.56  & 13460.78  & 32612\\
      12  & 2725  & 2824.12   & 248.15  & 3182  & 24194 & 28207.96  & 34708.02  & 62627\\
      13  & 6392  & 6301.04   & 360.13  & 6703  & 31    & 27871.40  & 48568.88  & 80423\\
      14  & 14295 & 14354.60  & 507.15  & 14898 & 23752 & 80045.08  & 110595.76 & 285170\\
      15  & 32547 & 32540.48  & 923.48  & 33695 & 94357 & 242712.04 & 386185.05 & 880569\\ \hline
    \end{tabular}
    \medskip
    \caption{Runtime results of solving the policy problem for family $\mathcal{N}_2$ based on its rewritten reward function in ms.}
    \label{tab:sequential-solving}
  \end{table}
  \begin{table}
    \centering
    \begin{tabular}{ r||r | r | r | r || r | r | r | r |}
      & \multicolumn{4}{c||}{Brute Force} & \multicolumn{4}{c}{SMT Solver}\\
      Size & Median & Mean & St. Dev. & 90\% & Median & Mean & St. Dev. & 90\% \\
      \hline\hline
      1   & 0     & 0.64      & 3.14    & 0     & 7   & 7.15    & 1.46    & 8\\
      2   & 0     & 1.96      & 5.19    & 10    & 9   & 9.95    & 5.13    & 18\\
      3   & 15    & 9.24      & 7.66    & 16    & 33  & 61.4    & 61.10   & 163\\
      4   & 32    & 36.40     & 6.98    & 48    & 237 & 724.25  & 915.76  & 1683\\
      5   & 160   & 160.64    & 11.89   & 175   &&&& \\
      6   & 634   & 637.80    & 16.86   & 653   &&&& \\
      7   & 2675  & 2677.64   & 17.46   & 2702  &&&& \\
      8   & 11441 & 11453.40  & 68.32   & 11543 &&&& \\
      9   & 49029 & 52364.28  & 4530.39 & 59670 &&&& \\ \hline
    \end{tabular}
    \medskip
    \caption{Runtime results of solving the policy problem for family $\mathcal{N}_{3}$ based on its rewritten reward function in ms.}
    \label{tab:backward-solving}
  \end{table}

  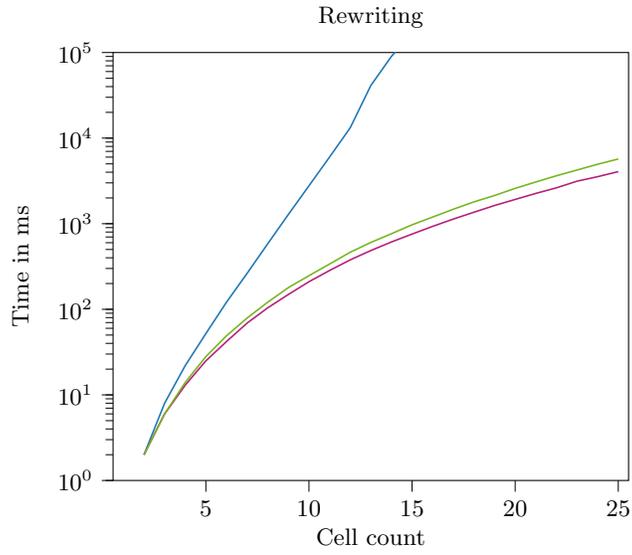
\begin{figure}
   \centering
\begin{tikzpicture}

\definecolor{darkgray176}{RGB}{176,176,176}
\definecolor{blue}{RGB}{31,119,180}
\definecolor{red}{RGB}{180,31,119}
\definecolor{green}{RGB}{119,180,31}

\begin{axis}[
log basis y={10},
tick align=outside,
tick pos=left,
title={Rewriting},
x grid style={darkgray176},
xmin=0.5, xmax=25.5,
xlabel={Cell count},
xtick style={color=black},
y grid style={darkgray176},
ymin=1, ymax=100000,
ymode=log,
ylabel={Time in ms},
ytick style={color=black}
]
\addplot [semithick, blue]
table {%
1 0
2 2
3 8
4 22
5 52
6 121
7 263
8 583
9 1281
10 2772
11 6018
12 13193
13 41182
14 90170
15 170000
};
\addplot [semithick, red]
table {%
0 0
1 0
2 2
3 6
4 13
5 25
6 42
7 69
8 104
9 149
10 210
11 284
12 378
13 486
14 612
15 757
16 931
17 1133
18 1362
19 1632
20 1922
21 2260
22 2625
23 3140
24 3544
25 4061
};
\addplot [semithick, green]
table {%
0 0
1 0
2 2
3 6
4 14
5 28
6 49
7 79
8 121
9 179
10 247
11 338
12 464
13 604
14 764
15 971
16 1196
17 1476
18 1799
19 2136
20 2589
21 3072
22 3632
23 4238
24 4956
25 5703
};
\end{axis}

\end{tikzpicture}
   \caption{Runtime of the rewriting of the reward function on
     $\mathcal{N}_1$ (green), $\mathcal{N}_2$ (red), and
     $\mathcal{N}_3$ (blue).}
   \label{fig:rewriting}
  \end{figure}

  \begin{figure}
   \centering
\begin{tikzpicture}

\definecolor{darkgray176}{RGB}{176,176,176}
\definecolor{blue}{RGB}{31,119,180}
\definecolor{red}{RGB}{180,31,119}
\definecolor{green}{RGB}{119,180,31}

\begin{axis}[
log basis y={10},
tick align=outside,
tick pos=left,
title={Brute Force},
xlabel={Cell count},
x grid style={darkgray176},
xmin=-0.75, xmax=15.75,
xtick style={color=black},
y grid style={darkgray176},
ymin=1, ymax=5000,
ymode=log,
ylabel={Time in ms},
ytick style={color=black}
]
\addplot [semithick, green]
table {%
0 0
1 0
2 0
3 0
4 0
5 15
6 15
7 31
8 63
9 149
10 337
11 780
12 1701
13 3759
14 8264
};
\addplot [semithick, blue]
table {%
0 0
1 0
2 0
3 15
4 32
5 160
6 634
7 2675
8 11441
9 49029
};
\addplot [semithick, red]
table {%
0 0
1 0
2 0
3 1
4 2
5 6
6 14
7 35
8 86
9 199
10 605
11 1440
12 2725
13 6392
14 14295
15 32547
};
\end{axis}

\end{tikzpicture}
   \caption{Runtime of the solving the policy problem by brute force
     (iterating over all possible deactivation patterns) on
     $\mathcal{N}_1$ (green), $\mathcal{N}_2$ (red), and
     $\mathcal{N}_3$ (blue).}
   \label{fig:brute-force}
  \end{figure}
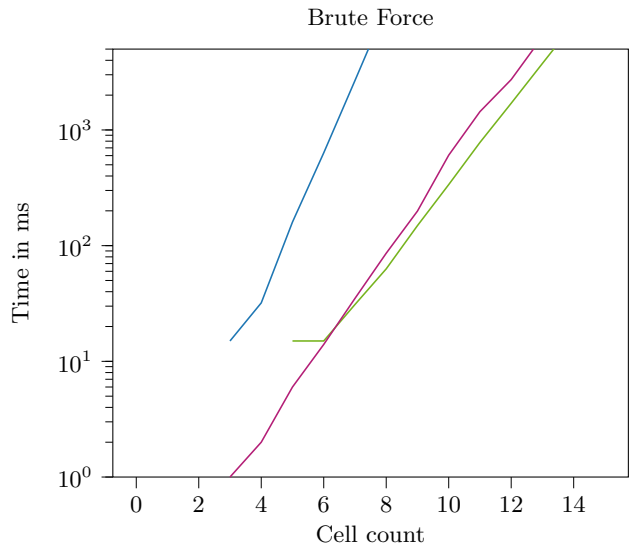

  \begin{figure}
   \centering
   \includegraphics[width=\textwidth]{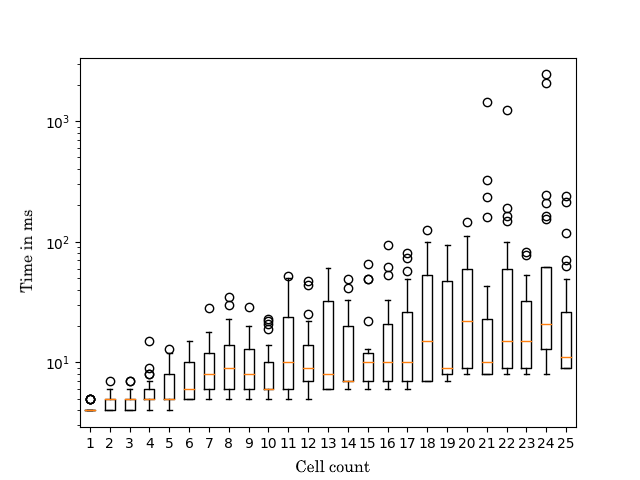}
   \caption{Runtime of the solving the policy problem using the z3
    SMT solver on the rewritten reward function on $\mathcal{N}_1$.}
   \label{fig:smt-concurrency}
  \end{figure}

  \begin{figure}
   \centering
   \includegraphics[width=\textwidth]{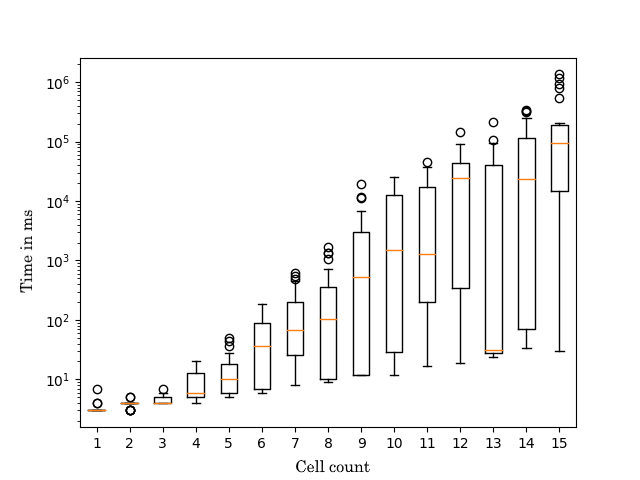}
   \caption{Runtime of the solving the policy problem using the z3
    SMT solver on the rewritten reward function on $\mathcal{N}_2$.}
   \label{fig:smt-sequential}
  \end{figure}

  \begin{figure}
   \centering
   \includegraphics[width=\textwidth]{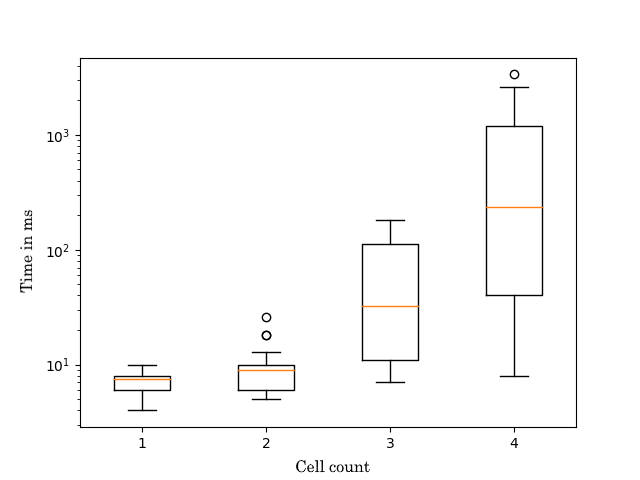}
   \caption{Runtime of the solving the policy problem using the z3
    SMT solver on the rewritten reward function on $\mathcal{N}_3$.}
   \label{fig:smt-backward}
  \end{figure}
\end{toappendix}

\section{Conclusion}\label{sec:conclusion}

We have introduced the formalism of stochastic decision Petri nets and
defined its semantics via an encoding into Markov decision
processes.
It turns out that finding optimal policies for a model that
incorporates concurrency, probability and decisions, is a non-trivial
task.
It is computationally hard even for restricted classes of nets and
constant policies.
However, we remark that workflow nets are often SAFC nets and a
constant deactivation policy is not unreasonable, given that one
cannot monitor and control a system all the time.
We have also presented an algorithm for the studied subproblem, which
we view as a step towards efficient partial-order techniques for
stochastic (decision) Petri nets.

\paragraph*{Related Work:}

Petri nets~\cite{r:petri-nets} are a well-known and widely studied
model of concurrent systems based on consumption and generation of
resources.
Several subclasses of Petri nets have received attention, among them
free-choice nets~\cite{DeselEsparza95} and occurrence nets, where the
latter are obtained by unfolding Petri nets for verification
purposes~\cite{eh:unfoldings-book}.

Our notion of stochastic decision Petri nets is an extension of the
well-known model of stochastic Petri nets~\cite{Marsan88}.
This model and a variety of generalizations are used for the
quantitative analyses of concurrent systems.
Stochastic Petri nets come in a continuous-time and in a discrete-time
variant, as treated in this paper.
That is, using the terminology of~\cite{Tolver2016}, we consider the
corresponding Markov chain of jumps, while in the continuous-time
case, firing rates determine not only the probability which transition
fires next, but also how fast a transition will fire dependent on the
marking.
These firing times are exponentially distributed, a distribution that
is memoryless, meaning that the probability of a transition firing is
independent on its waiting time.\todo{B: remove this sentence if we need space.}

Our approach was motivated by extending the probabilistic model of
stochastic Petri nets by a mechanism for decision making, as in the
extension of Markov chains~\cite{Tolver2016} to Markov decision
processes (MDPs)~\cite{bellman1957markovian}.
Since the size of a stochastic Petri net might be exponentially
smaller than the Markov chain that it generates, the challenge is to
provide efficient methods for determining optimal strategies,
preferably partial order methods that avoid the explicit
representation of concurrent events in an interleaving semantics.
Our complexity results show that the quest for such methods is
non-trivial, but some results can be achieved by suitably restricting
the considered Petri nets.

A different approach to include decision-making in Petri nets was\todo[color=green!50]{R2: If possible, discuss how the models compare to each other}
described by Beccuti et al.\ as Markov decision Petri
nets~\cite{Beccuti07,badsbf:markov-decision-petri-nets}.
Their approach, based on a notion of well-formed Petri nets,
distinguishes explicitly between a probabilistic part and a
non-deterministic part of the Petri net as well as a set of components
that control the transitions.
They use such nets to model concurrent systems and obtain experimental
results.
In a similar vein, graph transformation systems -- another model of
concurrent systems into which Petri nets can be encoded -- have been
extended to probabilistic graph transformation systems, including
decisions in the MDP sense~\cite{kg:prob-gts}.
The decision is to choose a set of rules with the same left-hand side
graph and a match, then a randomized choice is made among these rules.
Again, the focus is on modelling and to our knowledge neither of these
approaches provides complexity results.

Another problem related to the ones considered in this paper is the
computation of the expected execution time of a timed probabilistic
Petri net as described in~\cite{MeyerEO19}.
The authors treated timed probabilistic workflow nets (TPWNs) which
assumes that every transition requires a fixed duration to fire,
separate from the firing probability.
They showed that approximating the expected time of a sound SAFC TPWN
is $\hashp$-hard which is the functional complexity class
corresponding to $\pp$.
While the problems studied in their paper and in our paper are
different, the fact that both papers consider SAFC nets and obtain a
$\hashp$- respectively $\pp$-hardness result seems interesting and
deserves further study.

Our complexity results are closely connected with the analysis of
Bayesian networks~\cite{Pearl00}, which are a well-known graphical
formalism to represent conditional dependencies among random variables
and can be employed to reason about and compactly represent
probability distributions.  The close relation between Bayesian
networks and occurrence nets was observed in~\cite{BruniMM20}, which
gives a Bayesian network semantics for occurrence nets, based on the
notion of branching cells
from~\cite{ab:true-concurrency-probabilistic} that were introduced in
order to reconcile partial order methods -- such as unfoldings -- and
probability theory.  We took inspiration from this reduction in
Proposition~\ref{fig:bn-safc-reduction} and another of our reductions
(Proposition~\ref{prop:safc-np-pp-hard}) -- encoding Petri nets as
Bayesian networks -- is a transformation going into the other
direction, from Bayesian networks to SAFC nets.

In our own
work~\cite{chhk:update-ce-nets-bayesian,bchk:uncertainty-reasoning} we
considered a technique for uncertainty reasoning, combining both Petri
nets and Bayesian networks, albeit in a rather different
setting.
There we considered Petri nets with uncertainty, where one has only
probabilistic knowledge about the current marking of the net.
In this setting Bayesian networks are used to compactly store this
probabilistic knowledge and the main challenge is to update
respectively rewrite Bayesian networks representing such knowledge
whenever the Petri net fires.

\paragraph*{Future Work:}

As future work we plan to consider more general classes of Petri nets,
lifting some of the restrictions imposed in this paper.
In particular, it would be interesting to extend the method from
Section~\ref{sec:numerical} to nets that allow infinite runs.
Furthermore, dropping the free-choice requirement is desirable, but
problematic.
While the notion of branching cells does exist for stochastic nets
(see~\cite{ab:true-concurrency-probabilistic,BruniMM20}), it does not
accurately reflect the semantics of stochastic nets (see e.g.\ the
discussion on confusion in the introduction of~\cite{BruniMM20}).

As already detailed in the introduction, partial-order methods for
analyzing probabilistic systems, modelled for instance by stochastic
Petri nets, are in general poorly understood.
Hence, it would already be a major result to obtain scalable methods
for computing payoffs values for a stochastic net without decisions,
but with a high degree of concurrency.

In addition we plan to use the encoding of Petri nets into Bayesian
networks from~\cite{BruniMM20} (on which we based the proof of
Proposition~\ref{prop:safc-bn-reduction}) and exploit it to analyze
such nets by using dedicated methods for reasoning on Bayesian
networks.

Naturally, it would be interesting to extend analysis techniques in
such a way that they can deal with uncertainty and derive policies
when we have only partial knowledge, as in partially observable Markov
decision process (POMDPs), first studied
in~\cite{a:mdp-incomplete-state}.
However, this seems complex, given the fact that determining the best
strategy for POMDPs is a non-trivial problem in
itself~\cite{c:algorithms-mdps}.

Similarly, it is interesting to introduce a notion of time as in
continuous-time Markov chains~\cite{Tolver2016}, enabling us to
compute expected execution times as in~\cite{MeyerEO19}.

Last but not least, our complexity analysis and algorithm focus on
finding optimal constant policies.
A natural step would be to instead consider the problem of finding
optimal positional strategies as defined in Chapter~\ref{sec:sdpns},
which is the focus of most works on Markov decision processes (see for
example~\cite{c:algorithms-mdps}).

\bibliographystyle{plain}
\bibliography{references}

\end{document}